\documentclass[onecolumn,twoside]{IEEEtran}
\usepackage[left=0.67in,right=0.67in,top=0.66in,bottom=0.67in]{geometry}
\usepackage[T1]{fontenc}
\usepackage[latin9]{inputenc}
\usepackage{amsthm}
\usepackage{amsmath} 
\usepackage{graphicx} 
\usepackage{color} 
\usepackage{cite}
\usepackage{algpseudocode}
\usepackage{algorithm}
\usepackage{amssymb}
\usepackage{subfigure}
\usepackage{esint}
\usepackage{algpseudocode}
\usepackage{epstopdf}
\usepackage{multirow}
\usepackage{makecell}
\usepackage{float}
\usepackage{lipsum}
\usepackage{balance}
\usepackage{tikz}
\usetikzlibrary{shapes ,arrows, positioning} % For FSTD
\usepackage{adjustbox}
\usepackage{bbm}
\usepackage{hyperref}
\usepackage{mathrsfs}
\usepackage{bm}
\usepackage{booktabs} % added for \addlinespace within tables

\newtheorem{definition}{Definition}
\newtheorem*{definition*}{Definition}
\newtheorem{remark}{Remark}

\newtheorem{example}{Example}

\theoremstyle{plain}
\theoremstyle{plain}
\newtheorem{theorem}{Theorem}
\newtheorem{lemma}{Lemma}

\newcommand{\comment}[1]{}

\IEEEoverridecommandlockouts

\algnewcommand\algorithmicforeach{\textbf{for each}}
\algdef{S}[FOR]{ForEach}[1]{\algorithmicforeach\ #1\ \algorithmicdo}

\def\thickhline{\noalign{\hrule height1.4pt}}

\usepackage{colortbl}
\usepackage{hhline}

\begin{document}
\title{Exploiting the Degrees of Freedom: Multi-Dimensional Spatially-Coupled Codes \\ Based on Gradient Descent}

\author{
\IEEEauthorblockN{Ata Tanr{\i}kulu$^*$, Mete Y{\i}ld{\i}r{\i}m$^*$, and Ahmed Hareedy,~\IEEEmembership{Member,~IEEE}}

\thanks{$^*$Ata Tanr{\i}kulu and Mete Y{\i}d{\i}r{\i}m have equal contribution to this work.

Ata Tanr{\i}kulu, Mete Y{\i}ld{\i}r{\i}m, and Ahmed Hareedy are with the Department of Electrical and Electronics Engineering, Middle East Technical University, 06800 Ankara, Turkey (e-mail: ata.tanrikulu@metu.edu.tr, mete.yildirim@metu.edu.tr, and ahareedy@metu.edu.tr).

This work was supported in part by the T\"{U}B\.{I}TAK 2232-B International Fellowship for Early Stage Researchers. Parts of this paper were presented at the IEEE International Symposium on Information Theory (ISIT) 2024 \cite{reins_gdmd}.
}
}


\markboth{}%
{}



% for Transactions on Magnetics papers, we must declare the abstract and
% index terms PRIOR to the title within the \IEEEtitleabstractindextext
% IEEEtran command as these need to go into the title area created by

%%%%%%%%%%%%%%%%%%%%%%%%%%%%%%%%%%%%%
\IEEEtitleabstractindextext{%
\begin{abstract}

Because of their excellent asymptotic and finite-length (FL) performance, spatially-coupled (SC) codes are a class of low-density parity-check (LDPC) codes that is gaining increasing attention. Multi-dimensional (MD) SC codes are constructed by connecting copies of an SC code via relocations in order to mitigate various sources of non-uniformity and improve performance in many data storage and data transmission systems. As the number of degrees of freedom in the MD-SC code design increases, appropriately exploiting them becomes more difficult because of the complexity growth of the design process. In this paper, we propose a probabilistic framework for the MD-SC code design, which is based on the gradient-descent (GD) algorithm, to design high performance MD codes where this challenge is addressed. In particular, we express the expected number of detrimental objects, which we seek to minimize, in the graph representation of the code in terms of entries of a probability-distribution matrix that characterizes the MD-SC code design. We then find a locally-optimal probability distribution, which serves as the starting point of the FL algorithmic optimizer that produces the final MD-SC code. We adopt a recently-introduced Markov chain Monte Carlo (MCMC) FL algorithmic optimizer that is guided by the proposed GD algorithm. We apply our framework to various objects of interest. We start from simple short cycles, and then we develop the framework to address more sophisticated cycle concatenations, aiming at finer-grained optimization. We offer the theoretical analysis as well as the design algorithms. Next, we present experimental results demonstrating that our MD codes, conveniently called GD-MD codes, have notably lower numbers of targeted detrimental objects compared with the available state-of-the-art. Moreover, we show that our GD-MD codes exhibit significant improvements in error-rate performance compared with MD-SC codes obtained by a uniform distribution.

\end{abstract}

%%%%%%%%%%%%%%%%%%%%%%%%%%%%%%%%%%%%%

}

\maketitle

\IEEEdisplaynontitleabstractindextext

% \fi
\IEEEpeerreviewmaketitle

%%%%%%%%%%%%%%%%%%%%%%%%%%%%%%%%%%%%%
\section{Introduction}\label{sec_intro}

Low-density parity-check (LDPC) codes were introduced by Gallager in \cite{gal_th}, and they attracted significant attention nearly thirty years later, after their reinvention by MacKay and Neal \cite{mackay_neal,mackay_neal2}. Since then, LDPC codes have been among the most widely used error-correction techniques in several different applications. The performance of an LDPC code is negatively affected by detrimental configurations in the graph representation of the code, namely absorbing sets \cite{absorbing} formed by a combination of short cycles \cite{fossorier}. Spatially-coupled (SC) codes \cite{zigangirov} are a class of LDPC codes that offer excellent asymptotic performance \cite{lentmaier, kudekar, kudekar2}, superior finite-length (FL) performance \cite{costello, mitchell,battaglioni, GRADE}, and low-latency windowed decoding that is suitable for streaming applications \cite{siegel, iyengar}. Here, we focus on time-invariant SC codes \cite{battaglioni, amir}. The FL performance advantage of SC codes results from the additional degrees of freedom they offer the code designer via partitioning, which are added to what lifting enables for the case of protograph-based SC codes. Those additional degrees of freedom rapidly grow with the memory of the SC code. There are various results on high performance SC code designs with low memories in the literature \cite{battaglioni, rosnes, battaglioni-globe}. Some of these designs are optimal with respect to the number of detrimental objects \cite{oocpo}; others demonstrated notable performance gains in data storage and transmission systems \cite{channel_aware}. A recent idea introduced probabilistic design of SC codes with high memories, which is based on the gradient-descent (GD) algorithm, to find locally-optimal solutions with respect to the number of detrimental objects \cite{GRADE}. Throughout the rest of the paper, by a short cycle, we mean a cycle of length $6$ or $8$, and by a detrimental object, we mean either a short cycle or a concatenation of two such cycles. These short cycles and their concatenations are common graphical substructures in a wide array of detrimental absorbing sets \cite{GRADE,rohith2}.

Multi-dimensional spatially-coupled (MD-SC) codes are a new class of LDPC codes that has excellent performance in all regions \cite{ohashi, truhachev-mitchell-lentmaier, liu}. These codes can mitigate various sources of, possibly multi-dimensional, non-uniformity \cite{dahandeh, cai} in many applications. There are various effective MD-SC code designs introduced in recent literature \cite{schmalen, truhachev}. One idea is to use informed MD relocations in order to connect $M$ copies of a high performance SC code to construct a circulant-based, or a protograph-based, MD-SC code \cite{homa-hareedy, homa-lev, rohith2}, which is also the idea we focus on in this work. Recently, we introduced a Markov chain Monte Carlo (MCMC) method to perform FL optimization of LDPC codes, where the optimization algorithm samples from a distribution selected such that the probability is inversely proportional to the number of detrimental objects \cite{reins_mcmc}. We will employ this MCMC method for SC partitioning, MD relocation, and lifting, where for the first two, the method is guided by the output of our GD algorithms.

An MD-SC code can be constructed using three matrices: the partitioning matrix $\mathbf{K}$, the relocation matrix $\mathbf{M}$, and the lifting matrix $\mathbf{L}$. These matrices, as defined in Section~\ref{sec: prelim}, represent the operations performed in each respective stage of the code design process. We summarize in the following lines how one path of the most recent literature on the topic has evolved:
\begin{itemize}
\item[1)] A GD-based probabilistic framework for SC codes, which constructs optimized SC codes with high memory $m$, was proposed in~\cite{GRADE}. The key idea is designing a partitioning matrix $\mathbf{K}$ to achieve a reduced count of detrimental objects based on a locally-optimal probability distribution for edges across component matrices.
\item[2)]  This approach was extended in \cite{reins_gdmd} to a GD-based probabilistic framework for MD-SC codes. The key idea is designing a relocation matrix $\mathbf{M}$ (for a given underlying SC code) to achieve a reduced count of short cycles based on a locally-optimal probability distribution for edges across auxiliary matrices, which allow MD coupling, and an FL algorithmic optimizer. The FL optimizer relies on a voting method that appeared earlier in \cite{homa-hareedy} and \cite{rohith2}, then was updated in~\cite{homa-lev}. Observe that \cite{reins_gdmd} is the preliminary version of the work in this paper.
\item[3)] Recently, a novel GD-guided MCMC framework was proposed \cite{reins_mcmc} to design both a partitioning matrix $\mathbf{K}$ based on a locally-optimal probability distribution, which is obtained according to~\cite{GRADE}, and a lifting matrix $\mathbf{L}$. For FL optimization, this approach turned out to be more effective than the proposed algorithm in~\cite{GRADE} (and in fact more effective than the available state-of the-art FL optimization algorithms).
\end{itemize}

In this paper, we present the first probabilistic MD-SC code design framework that enables exploiting the abundance of degrees of freedom that comes with high SC memory $m$ and high number of MD auxiliary matrices (including the diagonal one) $M$. We provide a general theoretical analysis for our design framework that is applicable to arbitrary subgraphs of the Tanner graph of an MD-SC code. This includes simple objects detrimental to code performance, particularly short cycles, as well as more sophisticated configurations such as concatenations of two short cycles. Since focusing on cycles could cause unnecessary degrees of freedom being exhausted on isolated cycles that are much less problematic, this broader approach allows us to use the available degrees of freedom provided by MD-SC code design more effectively.

Our idea is to build a probability-distribution matrix, the entries of which represent the probability that a non-zero element in the base matrix, i.e., an edge in the underlying protograph, will be associated with the corresponding component matrix at the corresponding auxiliary matrix in the MD-SC construction. We develop the objective functions representing the expected number of short cycles as well as two-cycle concatenations in terms of the aforementioned probabilities, obtain their gradients, and derive the form of the solution (see Section~\ref{sec: framework} and Section~\ref{sec: theory}). We note that while the focus is on concatenations of two short cycles, our theoretical analysis is general for a detrimental configuration with a cycle basis having any number of cycles with arbitrary lengths. We then introduce the MD gradient-descent (MD-GRADE) algorithm (Section~\ref{subsec: MD-GRADE algo}) and the finite-length MCMC algorithm (Section~\ref{subsec: MCMC_algo}). Via our MD-GRADE algorithm, we find a locally-optimal distribution matrix, minimizing the expected number of target objects, and input it to the MCMC algorithm that converges on excellent FL MD-SC designs as demonstrated by the experimental results.

Our experimental results demonstrate that the new codes, which we call GD-MD codes, have significantly lower populations of short cycles compared with the available state-of-the-art (see Table~\ref{table: cycle_counts} and Table~\ref{table: cycle_counts_with_depth}), as well as significantly lower numbers of detrimental objects compared with their counterparts generated via a straightforward (uniform) distribution (see Table~\ref{table: object_counts}). Furthermore, we show that the reduction in the number of detrimental objects has a clear and remarkably positive effect on the error-rate performance, where the performance gains can reach up to $2.81$ orders of magnitude in favor of GD-MD codes (see Fig.~\ref{fig: shorter_code}). This improvement is a direct result of our framework, whereas the comparison codes exhibit an error floor or/and a degraded waterfall due to the high population of detrimental objects whose counts we aim to minimize.

The rest of the paper is organized as follows. In Section~\ref{sec: prelim}, we provide the necessary preliminary information related to SC and MD-SC codes together with the notation that will be used throughout the paper for these classes of codes. In Section~\ref{sec: framework}, we introduce our probabilistic framework focusing on short cycles, which we then generalize in Section~\ref{sec: theory} through the analysis of arbitrary subgraphs and complement by a special case study on the analysis of two short-cycle concatenations. We present the algorithms used to obtain locally-optimal probability-distributions and to construct FL GD-MD codes in Section~\ref{sec: algorithms}. We then present high performance GD-MD codes designed using the proposed framework, and we provide comparisons with state-of-the-art MD-SC codes from the literature as well as MD-SC codes generated using uniform distributions, in terms of the number of detrimental objects and error-rate performance, in Section~\ref{sec: numeric}. Finally, we conclude the paper in Section~\ref{sec: conclusion}. The partitioning, lifting, and relocation matrices of some of our GD-MD codes are provided in the paper appendix to ensure reproducibility of the results.


%%%%%%%%%%%%%%%%%%%%%%%%%%%%%%%%%%%%%
\section{Preliminaries}\label{sec: prelim}

Let $\mathbf{H}_{\textup{SC}}$ in \eqref{sc matrix} be the parity-check matrix of a circulant-based (CB) SC code with parameters $(\gamma, \kappa, z, L, m)$, where $\gamma \geq 3$ and $\kappa > \gamma$ are the column and row weights of the underlying block code, respectively. The SC coupling length and memory are $L$ and $m$, respectively.

\vspace{-0.7em}\begin{gather} \label{sc matrix}
\mathbf{H}_{\textup{SC}} =
\begin{bmatrix}
\mathbf{H}_0 & \mathbf{0}  &  & & & \mathbf{0} \vspace{-0.5em}\\
\mathbf{H}_1 & \mathbf{H}_0 & &  &  & \vdots \vspace{-0.5em}\\
\vdots & \mathbf{H}_1 & \ddots &  &  & \vdots \vspace{-0.5em}\\
\mathbf{H}_m & \vdots & \ddots & \ddots & & \mathbf{0} \vspace{-0.5em}\\
\mathbf{0} & \mathbf{H}_m & \ddots & \ddots & \ddots  & \mathbf{H}_0 \vspace{-0.5em}\\
\vdots & \mathbf{0} & & \ddots & \ddots & \mathbf{H}_1 \vspace{-0.5em}\\
\vdots & \vdots &  & &  \ddots & \vdots \vspace{-0.0em}\\
\mathbf{0} & \mathbf{0}& &  &  & \mathbf{H}_m
\end{bmatrix}.
\end{gather}
$\mathbf{H}_{\textup{SC}}$ is obtained from some binary matrix $\mathbf{H}^{\textup{g}}_{\textup{SC}}$ by replacing each non-zero (zero) entry in the latter by a $z \times z$ circulant (zero) matrix, $z \in \mathbb{N}$. Throughout the paper, the notation ``$\textup{g}$'' in the superscript of any matrix refers to its protograph matrix, where each circulant matrix is replaced by $1$ and each zero matrix is replaced by $0$. $\mathbf{H}^{\textup{g}}_{\textup{SC}}$ is called the protograph matrix of the SC code, and it has a convolutional structure composed of $L$ replicas $\mathbf{R}_{i}^{\textup{g}}$ of size $(L+m)\gamma \times \kappa$ each, where 
\[
\mathbf{R}_i^{\textup{g}} = \left[(\mathbf{0}_{\gamma i  \times \kappa })^{\textup{T}} \textup{} (\mathbf{H}_0^{\textup{g}})^{\textup{T}} \textup{ } (\mathbf{H}_1^{\textup{g}})^{\textup{T}} \textup{ } \dots \textup{ } (\mathbf{H}_m^{\textup{g}})^{\textup{T}} \textup{ } (\mathbf{0}_{\gamma (L-1-i)  \times \kappa })^{\textup{T}}\right]^{\textup{T}},
\]
for $i\in\lbrace 0,1,\dots, L-1\rbrace$. Here, $\mathbf{H}_y^{\textup{g}}$'s are all of size $\gamma  \times \kappa$, $\gamma, \kappa \in \mathbb{N}$, and they are all pairwise disjoint. The sum $\mathbf{{H}^{\textup{g}}} = \sum_{y=0}^m \mathbf{H}_y^{\textup{g}}$ is called the protograph base matrix, abbreviated to the base matrix in this paper, which is the protograph matrix of the underlying block matrix $\mathbf{H} = \sum_{y=0}^m \mathbf{H}_y$. In this paper, $\mathbf{{H}^{\textup{g}}}$ is taken to be all-one matrix, and the SC codes have quasi-cyclic structure. 
\par
The $\gamma \times \kappa$ matrix $\mathbf{K}$ whose entry at $(i,j)$ is $a \in \{0,1,\dots,m\}$ when $\mathbf{H}_a^{\textup{g}}=1$ at that entry is called the partitioning matrix. The $\gamma \times \kappa$ matrix $\mathbf{L}$ with $\mathbf{L}(i,j)=f_{i,j} \in \{0,1,\dots,z-1\}$, where $\boldsymbol{\sigma}^{f_{i,j}}$ is the circulant in $\mathbf{H}$ lifted from the entry $\mathbf{{H}^{\textup{g}}}(i,j)$, is called the lifting matrix. Here, $\boldsymbol{\sigma}$ is the $z \times z$ identity matrix cyclically shifted one unit to the left. Observe that $m+1$ is the number of component matrices $\mathbf{H}_y^{\textup{g}}$ (or $\mathbf{H}_y$) and $L$ is the number of replicas in $\mathbf{H}^{\textup{g}}_{\textup{SC}}$ (or $\mathbf{H}_{\textup{SC}}$). 
\par
We now define the MD-SC code. The parity-check matrix $\mathbf{H}_{\textup{MD}}$ of the MD-SC code has the following form:
% \vspace{-0.1em}
\begin{gather} \label{md matrix}
\mathbf{H}_{\textup{MD}} =
\begin{bmatrix}
\mathbf{H}'_{\textup{SC}} & \mathbf{X}_{M-1} & \mathbf{X}_{M-2} & \dots & \mathbf{X}_{2} & \mathbf{X}_{1} \vspace{-0.0em}\\
\mathbf{X}_{1} & \mathbf{H}'_{\textup{SC}} & \mathbf{X}_{M-1} & \dots & \mathbf{X}_{3} & \mathbf{X}_{2} \vspace{-0.0em}\\
\mathbf{X}_{2} & \mathbf{X}_{1} & \mathbf{H}'_{\textup{SC}} & \dots & \mathbf{X}_{4} & \mathbf{X}_{3} \vspace{-0.0em}\\
\vdots & \vdots & \vdots & \ddots & \vdots & \vdots \vspace{-0.0em}\\
\mathbf{X}_{M-2} & \mathbf{X}_{M-3} & \mathbf{X}_{M-4} & \dots & \mathbf{H}'_{\textup{SC}} & \mathbf{X}_{M-1} \vspace{-0.0em}\\
\mathbf{X}_{M-1} & \mathbf{X}_{M-2} & \mathbf{X}_{M-3} & \dots & \mathbf{X}_{1} & \mathbf{H}'_{\textup{SC}}
\end{bmatrix},
\end{gather}
where 
\vspace{-0.5em}\begin{align} \label{hmd constraint}\mathbf{H}_{\textup{SC}} = \mathbf{H}'_{\textup{SC}} + \sum_{\ell=1}^{M-1} \mathbf{X}_\ell.
\end{align}  
Here, $\mathbf{X}_\ell$'s are called the auxiliary matrices. The MD-SC code is obtained from $M$ copies of $\mathbf{H}_{\textup{SC}}$ on the diagonal by relocating some of its circulants from each replica of every copy of $\mathbf{H}_{\textup{SC}}$ to the corresponding locations in $\mathbf{X}_{\ell}$ copies (simply by shifting the circulants $\ell L\kappa z$ units cyclically to the left), and then by coupling them in a sliding manner as in \eqref{md matrix}. For convention, we use $\mathbf{X}_0 = \mathbf{H}_{\textup{SC}}'$.

\begin{definition} The $M (L+m)\gamma  \times M L \kappa $ matrix $\mathbf{H}^{\textup{g}}_{\textup{MD}}$ obtained by replacing $\mathbf{H}'_{\textup{SC}}$ by $\mathbf{H}'^{\textup{g}}_{\textup{SC}}$ and $ \mathbf{X}_\ell$'s with $ \mathbf{X}^{\textup{g}}_\ell$'s in \eqref{md matrix} is called the MD protograph of $\mathbf{H}_{\textup{MD}}$. Here, $\mathbf{X}^{\textup{g}}_{\ell}$'s are uniquely determined by the following properties:
\begin{enumerate} 
\item The equation $\mathbf{H}^{\textup{g}}_{\textup{SC}} = \mathbf{H}'^{\textup{g}}_{\textup{SC}} + \sum_{\ell=1}^{M-1} \mathbf{X}^{\textup{g}}_\ell$ holds.
\item $\mathbf{X}_{\ell}$'s are obtained by replacing each non-zero (zero) entry in $\mathbf{X}^{\textup{g}}_{\ell}$ by the $z \times z$ circulant $\boldsymbol{\sigma}^{f_{i,j}}$ (the zero matrix $\mathbf{0}_{z \times z}$), $z \in \mathbb{N}$, that has the appropriate power $f_{i,j}$ from the lifting matrix $\mathbf{L}$. 
\end{enumerate}
\end{definition} 

Relocations are represented by an MD mapping as follows: 
$$F: \{\mathcal{C}_{i,j} \,|\, 1 \leq i \leq \gamma, 1 \leq j \leq \kappa \} \rightarrow \{0,1,\dots, M-1\},$$
where $\mathcal{C}_{i,j}$ is the circulant corresponding to $1$ at entry $(i,j)$ of the base matrix $\mathbf{{H}^{\textup{g}}}$, and $F(\mathcal{C}_{i,j})$ is the index of the auxiliary matrix to which $\mathcal{C}_{i,j}$ is located. $F^{-1}(0)$ is the set of non-relocated circulants, and the percentage $\mathcal{T}$ of relocated circulants is called the MD density. Conveniently, we define the $\gamma \times \kappa$ matrix $\mathbf{M}$ with $\mathbf{M}(i,j)=F(\mathcal{C}_{i,j})$ as the relocation matrix. An MD-SC code is uniquely determined by the matrices $\mathbf{K}$, $\mathbf{L}$, and $\mathbf{M}$, and the parameters of an MD-SC code are described by the tuple $(\gamma, \kappa, z, L, m, M)$. 
\par
A cycle of length $2g$ (cycle-$2g$) in the Tanner graph of $\mathbf{H}_{\textup{MD}}$ is a $2g$-tuple $(i_1,j_1,i_2,j_2,\dots,i_g,j_g)$, where $(i_k,j_k)$ and $(i_k,j_{k+1})$, $1 \leq k \leq g, \textup{ } j_{g+1}=j_1$, are the positions of its non-zero entries in $\mathbf{H}_{\textup{MD}}$. Short cycles are detrimental to the performance of an LDPC code as they result in dependencies in decoding, negatively affecting the waterfall region, and they combine to create absorbing sets, negatively affecting the error-floor region \cite{absorbing, channel_aware}. We first focus on cycles-$k$ and then concatenations of a cycle-$k$ and a cycle-$k'$, where $k, k' \in \{6,8\}$, in this paper. We denote such concatenations via the cycle lengths as follows $k$-$k'$, and we focus on $6$-$6$, $6$-$8$, and $8$-$8$ configurations. Relocations are done in order to eliminate such objects when the underlying SC code has girth, i.e., shortest cycle length, $6$ or $8$. The necessary and sufficient condition for a cycle-$2g$ in $\mathbf{H}_{\textup{SC}}$ to create a cycle-$2g$ in $\mathbf{H}_{\textup{MD}}$, i.e., remains \textit{active}, after relocations is studied in \cite{homa-hareedy} and also given in the second part of Lemma~\ref{lemma:par_rel_con}. Practically, a cycle candidate that remains active creates multiple cycle candidates associated with the same cycle length, but we say it becomes ``a cycle candidate'' from an ontological point of view.

\begin{lemma} \label{lemma:par_rel_con} Each cycle-$2g$ in the Tanner graph of an MD-SC code corresponds to a cycle-$2g$ candidate in the protograph base matrix, $\mathbf{H}^{\textup{g}}$.\footnote{A cycle candidate is a way of traversing a protograph pattern to reach a cycle after partitioning/relocations/lifting (see \cite{channel_aware}).} We label a cycle-$2g$ candidate in $\mathbf{H}^{\textup{g}}$ as $ \left(i_1,j_1, i_2,j_2, \ldots, i_g, j_g\right) $, where $ \left(i_k, j_k\right) $ and $ \left(i_k, j_{k+1}\right) $, $ 1 \leq k \leq g $, $ j_{g+1} = j_1 $, are nodes (entries) of this cycle candidate. The partitioning and relocation matrices are $ \mathbf{K} $ and $ \mathbf{M} $, respectively. Then, this cycle candidate becomes a cycle-$2g$ candidate in the SC protograph, i.e., remains active, if and only if the following condition is satisfied \cite{battaglioni}:
\begin{align} \label{eq:partition_condition}
\sum_{k=1}^g \left[\mathbf{K}(i_k, j_k) - \mathbf{K}(i_k, j_{k+1})\right] = 0. 
\end{align}
This cycle candidate becomes a cycle-$2g$ candidate in the MD protograph, i.e., remains active, if and only if it satisfies \eqref{eq:partition_condition} and the following condition\cite{homa-hareedy}:
\begin{align} \label{eq:relocation_condition}
\sum_{k=1}^g \left[\mathbf{M}(i_k, j_k) - \mathbf{M}(i_k, j_{k+1})\right] \equiv 0  \pmod{M}.
\end{align}
\end{lemma}

Finally, we also provide the definition of unlabeled absorbing and trapping sets \cite{absorbing,karimi}, which are essential for our analysis. Here, ``unlabeled'' means all edge weights are set to $1$, which is natural for binary codes.

\begin{definition} (Unlabeled absorbing and trapping sets). \label{def:uas-uts}
Let $V$ be a subset of variable nodes (VNs) in the unlabeled graph of a code. Consider the subgraph induced by $V$. Let $\mathcal{O}$ denote the set of odd-degree neighboring check nodes (CNs) of $V$, and let $\mathscr{T}$ denote the set of even-degree neighboring CNs of $V$. This subgraph configuration is called an $(a, d)$ unlabeled trapping set if $|V| = a$ and $|\mathcal{O}| = d$. Moreover, an $(a, d)$ unlabeled trapping set is said to be an unlabeled absorbing set if, in addition, each VN in $V$ has strictly more neighbors in $\mathscr{T}$ than in $\mathcal{O}$.  Note that for a binary code, all edges are unlabeled by default. An object is said to be elementary if the degree of each associated CN is at most two. Elementary unlabeled absorbing and trapping sets will be abbreviated as UAS and UTS, respectively, for the remaining of the paper.
\end{definition}

%%%%%%%%%%%%%%%%%%%%%%%%%%%%%%%%%%%%%
\section{Novel Framework for Probabilistic MD-SC Code Design}\label{sec: framework}

In this section, we present our novel probabilistic design framework that searches for a locally-optimal probability distribution for auxiliary matrices of an MD-SC code. We start by focusing on cycles. We define a joint probability distribution over the component matrix and the auxiliary matrix (including $\mathbf{X}_0$) to which a $1$ in the base matrix is assigned. This distribution is characterized by an $(m+1) \times M$ so-called probability-distribution matrix. We then formulate an optimization problem and find its solution form to obtain the locally-optimal probability-distribution matrix. Finally, we express the expected number of short cycles in $\mathbf{H}_{\textup{MD}}$ as a function of this probability distribution. The matrix $\mathbf{H}_{\textup{MD}}$ will be designed based on this distribution later on in Section~\ref{sec: algorithms}. Here, our vectors are row vectors by default.

\subsection{Probabilistic Derivations}
In this subsection, we relate the joint probability distribution to the expected number of cycle candidates in the MD protograph. In particular, we derive the probability of cycle activeness, then use it to find the expectation, both in terms of this distribution.

\begin{definition} (Joint probability-distribution matrix). Let $m \geq 1$ and $M \geq 2$. Then, 
\begin{align} \mathbf{P} = \begin{bmatrix} p_{0,0} & p_{0,1} & \dots & p_{0,M-1} \\ p_{1,0} & p_{1,1} & \dots & p_{1,M-1} \\ 
\vdots & \vdots & \ddots & \vdots 
\\ p_{m,0} & p_{m,1} & \dots & p_{m,M-1} \end{bmatrix}_{(m+1) \times M},
\end{align}
where 
\begin{align} \label{constraint_overall} \sum_{i=0}^{m} \sum_{j=0}^{M-1} p_{i,j} = 1 
\end{align}
and each $p_{i,j} \in [0,1]$ specifies the probability that a $1$ in $\mathbf{H}_{\textup{MD}}^{\textup{g}}$ is assigned to the $i^{\textup{th}}$ component of the $j^{\textup{th}}$ auxiliary matrix. Thus, $\mathbf{P}$ is referred to as the \textbf{(joint) probability-distribution matrix}. Recall that $\mathbf{X}_0 = \mathbf{H}_{\textup{SC}}'$. The two-variable \textbf{coupling polynomial} of an MD-SC code associated with the probability-distribution matrix $\mathbf{P}$ is
\vspace{-0.1em}\begin{align} \label{eq: f}
f(X,Y;\mathbf{P}) = \sum_{i=0}^{m} \sum_{j=0}^{M-1} p_{i,j}X^i Y^j,
\end{align}
which is abbreviated as $f(X,Y)$ when the context is clear.

\end{definition}
\begin{definition} \label{definition: a-con} For a matrix $\mathbf{A}$ of size $k \times \ell$, the vector $\mathbf{a}^{\textup{con}}=[a_{r}]_{1 \times k\ell}$ is defined as the vector obtained by concatenating the rows of the matrix $\mathbf{A}_{k \times \ell}$ from top to bottom, i.e., $a_{\ell i+j}$ is the entry of $\mathbf{A}$ at $(i,j)$ for $0 \leq i \leq k-1$ and $0 \leq j \leq \ell-1$.
\end{definition}

\begin{definition} The vector $\mathbf{p}^{\textup{con}}$ corresponding to the (joint) probability-distribution matrix $\mathbf{P}$ is called the \textbf{probability-distribution vector}. 
\end{definition}

\begin{remark} \label{rmk: order} In the FL construction of $\mathbf{H}_{\textup{MD}}$ via the MCMC Algorithm in Section \ref{sec: algorithms}, relocations are performed after partitioning and lifting. The respective order in the set-up of Theorems \ref{thm: expected number of cycle-6}, \ref{thm: expected number of cycle-8}, and \ref{thm: forecast} below, however, is partitioning, relocations, then lifting. This change in order, of course, does not affect $\mathbf{H}_{\textup{MD}}$ nor the (expected) cycle counts.
\end{remark}

We note that expectations in Theorem~\ref{thm: expected number of cycle-6} and Theorem~\ref{thm: expected number of cycle-8} are $L$-invariant and $M$-invariant.

\begin{theorem} \label{thm: expected number of cycle-6} Let $[\cdot]_{i,j}$ denote the coefficient of $X^iY^j$ in a two-variable polynomial. Denote by $P_6(\mathbf{p}^{\textup{con}})$ the probability of a cycle-$6$ candidate in the base matrix becoming a cycle-$6$ candidate in the MD protograph $\mathbf{H}_{\textup{MD}}^{\textup{g}}$ under random partitioning and relocations with respect to the probability-distribution vector $\mathbf{p}^{\textup{con}}$. Then, we have
\vspace{-0.1em}\begin{align} \label{exp: prob6} P_6(\mathbf{p}^{\textup{con}}) = \sum_{M \vert b}  \left[ f^3(X,Y)f^3(X^{-1},Y^{-1}) \right]_{0,b}.
\end{align}
Thus, the expected number $N_6(\mathbf{p}^{\textup{con}})$ of cycle-$6$ candidates in  $\mathbf{H}_{\textup{MD}}^{\textup{g}}$ is given by
\begin{align} \label{eqn: expected number of cycle-6} N_6(\mathbf{p}^{\textup{con}})= 6 \binom{\gamma}{3} \binom{\kappa}{3} \sum_{M \vert b}  \left[ f^3(X,Y)f^3(X^{-1},Y^{-1}) \right]_{0,b}.
\end{align}
\end{theorem}

\begin{proof} 
Denote a cycle-$6$ candidate $C$ in the base matrix $\mathbf{{H}^{\textup{g}}}$ by $(i_1,j_1,i_2,j_2,i_3,j_3)$, where $(i_{k},j_{k})$ and $(i_{k},j_{k+1})$, $1\leq k\leq 3$, $j_{4}=j_1$, are edges of $C$ in $\mathbf{{H}^{\textup{g}}}$. From the partitioning condition in \eqref{eq:partition_condition} and the relocation condition in \eqref{eq:relocation_condition}, $P_6(\mathbf{p}^{\textup{con}})$ is given by the joint probability of both condition equations being satisfied, and this probability is derived in \eqref{eqn: prob6}. Note that the number of cycle-$6$ candidates in $\mathbf{{H}^{\textup{g}}}$ is $6\binom{\gamma}{3} \binom{\kappa}{3}$, and each becomes a cycle-$6$ candidate in $\mathbf{H}_{\textup{MD}}^{\textup{g}}$ with probability $P_6(\mathbf{p}^{\textup{con}})$. Therefore, $N_6(\mathbf{p}^{\textup{con}})= 6\binom{\gamma}{3} \binom{\kappa}{3} P_6(\mathbf{p}^{\textup{con}})$.
\end{proof}

\begin{figure*}[ht!] 
\noindent\makebox[\linewidth]{\rule{\textwidth}{1,5pt}} \\
\begin{align} \label{eqn: prob6}
&P_6(\mathbf{p}^{\textup{con}})=\mathbb{P}\left[\sum_{k=1}^{3}[\mathbf{K}(i_{k},j_{k})-\mathbf{K}(i_{k},j_{k+1})]=0, \hspace{+0.3em}  \sum_{k=1}^{3} [\mathbf{M}({i_k,j_k})-\mathbf{M}({i_k,j_{k+1}})] \equiv 0 \pmod{M} \right] \nonumber \\
=&\sum\limits_{\substack{\sum\nolimits_{k=1}^{3}(x_k-y_k) = 0, \\ \sum\nolimits_{k=1}^{3} (u_k-v_k) \equiv 0 \pmod{M}}} \prod_{k=1}^3\mathbb{P}\left[ \mathbf{K}(i_{k},j_{k})=x_k, \mathbf{M}({i_k,j_k})=u_k, \mathbf{K}(i_{k},j_{k+1})=y_k, \mathbf{M}({i_k,j_{k+1}})=v_k \right] \nonumber \\
=&\sum\limits_{\substack{\sum\nolimits_{k=1}^{3}(x_k-y_k) = 0, \\ \sum\nolimits_{k=1}^{3} (u_k-v_k) \equiv 0 \pmod{M}}} p_{x_1,u_1}p_{x_2,u_2}p_{x_3,u_3}p_{y_1,v_1}p_{y_2,v_2}p_{y_3,v_3} \nonumber \\
=&\sum_{M \vert b} \left[\sum\limits_{\substack{x_k,y_k \in \{0,1, \dots, m\}, \\ u_k,v_k \in \{0,1, \dots, M-1\}}} p_{x_1,u_1}p_{x_2,u_2}p_{x_3,u_3}p_{y_1,v_1}p_{y_2,v_2}p_{y_3,v_3} X^{x_1+x_2+x_3-y_1-y_2-y_3}Y^{u_1+u_2+u_3-v_1-v_2-v_3} \right]_{0,b} \nonumber \\
=&\sum_{M \vert b} \left[f^3(X,Y)f^3(X^{-1},Y^{-1})\right]_{0,b}.
\end{align}
\noindent\makebox[\linewidth]{\rule{\textwidth}{1,5pt}}
\end{figure*}

\begin{theorem} \label{thm: expected number of cycle-8}
Denote by $N_8(\mathbf{p}^{\textup{con}})$ the expected number of cycle-$8$ candidates (see Fig.~\ref{fig: cycle8_patterns}) in the MD protograph. Then,
\begin{align}\label{eqn: expected number of cycle-8}
N_8(\mathbf{p}^{\textup{con}}) 
= \sum_{M \vert b} &\hspace{+0.2em}\Big\{ w_1 \left[ f^2(X^2,Y^2)f^2(X^{-2},Y^{-2}) \right]_{0,b} 
+ w_2 \left[ f(X^2,Y^2)f(X^{-2},Y^{-2})f^2(X,Y)f^2(X^{-1},Y^{-1}) \right]_{0,b} \nonumber \\
&+ w_3  \left[ f(X^2,Y^2)f^2(X,Y)f^4(X^{-1},Y^{-1}) \right]_{0,b}
+ w_4  \left[ f^4(X,Y)f^4(X^{-1},Y^{-1}) \right]_{0,b} \Big\},
\end{align}
where $w_1 = \binom{\gamma}{2} \binom{\kappa}{2}$, $w_2 = 3\binom{\gamma}{2} \binom{\kappa}{3}+ 3\binom{\gamma}{3} \binom{\kappa}{2}$, $w_3= 18 \binom{\gamma}{3} \binom{\kappa}{3}$, $ w_4 = 6 \binom{\gamma}{2} \binom{\kappa}{4}+6\binom{\gamma}{4} \binom{\kappa}{2} +36 \binom{\gamma}{3} \binom{\kappa}{4}+36\binom{\gamma}{4} \binom{\kappa}{3}+ 72 \binom{\gamma}{4} \binom{\kappa}{4}$ if $\gamma \geq 4$, and $w_4 = 6 \binom{\gamma}{2} \binom{\kappa}{4}+36 \binom{\gamma}{3} \binom{\kappa}{4}$ if $\gamma = 3$, where $\kappa \geq 4$.
\end{theorem}

\begin{proof}
Observe that there are nine protograph patterns that can produce cycles-$8$ after partitioning/relocations/lifting \cite{channel_aware} (see Fig.~\ref{fig: cycle8_patterns}). Following the logic in Theorem \ref{thm: expected number of cycle-6}, \cite[Theorem 1]{GRADE}, and their proofs, we can derive the result above (see also \cite[Chapter~5]{ahh_phd} for more details).
\end{proof}

\begin{remark} In this paper, we do not discuss cycles-$4$ in the analysis since removing all of them via informed lifting is easy. We also consider prime $z$, which implies that protograph cycles-$4$ cannot produce cycles-$8$ after lifting. Therefore, we take the weight of the term corresponding to cycle-$8$ candidate obtained by traversing a cycle-$4$ twice (see Fig.~\ref{fig: cycle8_patterns}) $w_1=0$ in the version of Algorithm~\ref{algo: MD-GRADE8} that handles cycles-$8$ in Section \ref{subsec: MD-GRADE algo} to obtain the probability-distribution matrix. Observe also that the major contribution to cycle-$8$ counts comes from the candidates with $8$ distinct entries (associated with the term multiplied by $w_4$ in \eqref{eqn: expected number of cycle-8} of Theorem~\ref{thm: expected number of cycle-8}).
\end{remark}

%%%%%%%%%%%%%%%%%%%%%%%%%%%%%%%%%%%%%
\subsection{Gradient-Descent MD-SC Solution Form} 

In this section, we derive the gradient-descent MD-SC solution form by analyzing the Lagrangian of the objective functions presented above. This solution form is a key ingredient in the MD gradient-descent (MD-GRADE) algorithm in Section \ref{subsec: MD-GRADE algo}. Our algorithm produces near-optimal relocation percentages for gradient-descent MD-SC (GD-MD) codes with arbitrary memory of the underlying SC code and arbitrary number of auxiliary matrices.

\begin{remark} The vector $\mathbf{p}^*$ of the probabilities $p_i^*$'s in the constraints of Lemma~\ref{lemma: lagrange6} and Lemma~\ref{lemma: lagrange8} introduced below is set to a locally-optimal edge distribution for the underlying SC code obtained as in \cite{GRADE}. Since the $p_i^*$'s add up to $1$, \eqref{constraint_overall} is automatically satisfied. We chose this approach in order to apply MD relocations to the best underlying SC code. These probabilities are used as an input to algorithms in Section~\ref{subsec: MD-GRADE algo}. Separate design for partitioning, relocations, and lifting was shown to be the best approach in the literature \cite{GRADE,channel_aware,rohith2,homa-lev}.
\end{remark} 

\begin{lemma} \label{lemma: lagrange6} For $P_6(\mathbf{p}^{\textup{con}})$ to be locally minimized subject to the constraints 
\begin{align} \label{eq: row sum constraint}
 p_{i,0}+p_{i,1} + \dots + p_{i,M-1} = p_i^*,
\end{align}
for all $i \in \{0,1,\dots,m\}$, where each $p_{i,j} \in [0,1]$, it is necessary that the following equations hold for some $c_i \in \mathbb{R}$:
\begin{align} \label{eq_6}
\sum_{M \vert b} \left[ f^3(X,Y)f^2(X^{-1},Y^{-1}) \right]_{i,b+j} = c_i,
\end{align}
for all $i \in \{0,1,\dots, m\}$ and $j \in \{0,1, \dots, M-1\}$.
\end{lemma}

\begin{proof} Since the constraints of the optimization problem in hand satisfy the linear independence constraint qualification (LICQ), Karush-Kuhn-Tucker (KKT) conditions must hold. Note that the vector $\mathbf{{r}}^{\textup{con}}$ to appear below is obtained by concatenating the rows of the $(m+1) \times M$ matrix 
\vspace{-0.3em}\begin{align} 
\mathbf{R} = \begin{bmatrix} r_0 & r_0 & \dots & r_0 \\ r_1 & r_1 & \dots & r_1 \\ 
\vdots & \vdots & \ddots & \vdots 
\\ r_m & r_m & \dots & r_m \end{bmatrix}.
\end{align}
We now consider the Lagrangian $L_6(\mathbf{p}^{\textup{con}}) = P_6(\mathbf{p}^{\textup{con}})+ \sum_{i=0}^m r_i (p_i^*-p_{i,0}-p_{i,1} - \dots - p_{i,M-1})$ and compute its gradient as follows:
\begin{align} &\nabla_{\mathbf{p}^{\textup{con}}}  L_6(\mathbf{p}^{\textup{con}}) 
=  \nabla_{\mathbf{p}^{\textup{con}}} \left ( P_6(\mathbf{p}^{\textup{con}})+ \sum_{i=0}^m r_i (p_i^*-p_{i,0}-p_{i,1} - \dots - p_{i,M-1}) \right ) \nonumber \\ 
&= \nabla_{\mathbf{p}^{\textup{con}}} \left( \sum_{M \vert b} \left[ f^3(X,Y)f^3(X^{-1},Y^{-1}) \right]_{0,b} \right) - \mathbf{{r}}^{\textup{con}}
= \sum_{M \vert b} \left[ \nabla_{\mathbf{p}^{\textup{con}}}  (f^3(X,Y)f^3(X^{-1},Y^{-1}))\right]_{0,b} - \mathbf{r}^{\textup{con}} \nonumber \\
&= \sum_{M \vert b} \Big\{ 3 \left[f^3(X,Y)f^2(X^{-1},Y^{-1})\nabla_{\mathbf{p}^{\textup{con}}}f(X^{-1},Y^{-1}) \right]_{0,b} 
+  3 \left[ f^2(X,Y)f^3(X^{-1},Y^{-1})\nabla_{\mathbf{p}^{\textup{con}}}f(X,Y) \right]_{0,b} \Big\} - \mathbf{r}^{\textup{con}} \nonumber \\
&= 6  \sum_{M \vert b}  \left[ f^3(X,Y)f^2(X^{-1},Y^{-1})\mathbf{v}_1^{\textup{con}}) \right]_{0,b} - \mathbf{r}^{\textup{con}}, 
\end{align} 
where the last equality follows from the polynomial-symmetry observation for any polynomial $r(X,Y)$
\[
\sum_{M \vert b} \left[ r(X,Y)\nabla_{\mathbf{p}^{\textup{con}}}f(X^{-1},Y^{-1}) \right]_{0,b} = \sum_{M \vert b} \left[ r(X^{-1},Y^{-1})\nabla_{\mathbf{p}^{\textup{con}}}f(X,Y) \right]_{0,b}.
\]
Here, $\mathbf{v}_k^{\textup{con}}$ is the vector of length $(m+1)M$ obtained by concatenating the rows of the following $(m+1) \times M$ matrix (of monomials)
\begin{align}
 \mathbf{V}_k = \begin{bmatrix} 1 & Y^{-k} & \dots & Y^{-(M-1)k} \\ X^{-k} & X^{-k}Y^{-k} & \dots & X^{-k}Y^{-(M-1)k} \\ 
\vdots & \vdots & \ddots & \vdots 
\\ X^{-mk} & X^{-mk}Y^{-k} & \dots & X^{-mk}Y^{-(M-1)k} \end{bmatrix}
\end{align}
for $k \in \mathbb{Z}$. When $P_6(\mathbf{p}^{\textup{con}})$ reaches its local minimum, $\nabla_{\mathbf{p}^{\textup{con}}}  L_6(\mathbf{p}^{\textup{con}}) = \mathbf{0}_{1\times (m+1)M},$
which directly leads to \eqref{eq_6} by defining $c_i = r_i /6$ for all $i \in \{0,1,\dots,m\}$.
\end{proof}

\begin{lemma} \label{lemma: lagrange8} For $N_8(\mathbf{p}^{\textup{con}})$ to be locally minimized subject to the constraints
\begin{align} 
p_{i,0}+p_{i,1} + \dots + p_{i,M-1} = p_i^*,
\end{align}
for all $i \in \{0,1,\dots,m\}$, where each $p_{i,j} \in [0,1]$, it is necessary that the following equations hold for some $c_i \in \mathbb{R}$:
\begin{align} \label{eq_8}
\sum_{M \vert b} &\hspace{+0.2em}\Big\{ w_1 \left[ 4f^2(X^2,Y^2)f(X^{-2},Y^{-2}) \right]_{2i,b+2j} 
+ w_2 \left[ 2f(X^2,Y^2)f^2(X,Y)f^2(X^{-1},Y^{-1}) \right]_{2i,b+2j} \nonumber \\
&+ w_2 \left[ 4f(X^2,Y^2)f(X^{-2},Y^{-2})f^2(X,Y)f(X^{-1},Y^{-1}) \right]_{i,b+j}
+ w_3 \left[ f^2(X,Y)f^4(X^{-1},Y^{-1}) \right]_{-2i,b-2j} \nonumber \\ 
&+ w_3 \left[ 2f(X^2,Y^2)f(X,Y)f^4(X^{-1},Y^{-1}) \right]_{-i,b-j} 
+ w_3 \left[ 4f(X^2,Y^2)f^2(X,Y)f^3(X^{-1},Y^{-1}) \right]_{i,b+j} \nonumber \\ 
&+  w_4 \left[ 8f^4(X,Y)f^3(X^{-1},Y^{-1}) \right]_{i,b+j} \Big\}
= c_i,
\end{align}
for all $i \in \{0,1, \dots, m\}$ and $j \in \{0,1, \dots, M-1\}$.
\end{lemma}

\begin{proof}
We consider the Lagrangian $L_8(\mathbf{p}^{\textup{con}}) = N_8(\mathbf{p}^{\textup{con}})+ \sum_{i=0}^m r_i (p_i^*-p_{i,0}-p_{i,1} - \dots - p_{i,M-1})$ and compute its gradient as follows:
\begin{align}\label{eqn: cycle_8_grad}
&\nabla_{\mathbf{p}^{\textup{con}}}  L_8(\mathbf{p}^{\textup{con}})
= \nabla_{\mathbf{p}^{\textup{con}}} (N_8(\mathbf{p}^{\textup{con}})+ \sum_{i=0}^m r_i (p_i^*-p_{i,0}-p_{i,1} - \dots - p_{i,M-1})) \nonumber \\ 
&= \sum_{M \vert b} \nabla_{\mathbf{p}^{\textup{con}}} \Big\{ w_1 \left[ f^2(X^2,Y^2)f^2(X^{-2},Y^{-2}) \right]_{0,b}
+ w_2 \left[ f(X^2,Y^2)f(X^{-2},Y^{-2})f^2(X,Y)f^2(X^{-1},Y^{-1}) \right]_{0,b} \nonumber \\
&\hspace{+4.7em}+ w_3 \left[ f(X^2,Y^2)f^2(X,Y)f^4(X^{-1},Y^{-1}) \right]_{0,b}
+  w_4 \left[ f^4(X,Y)f^4(X^{-1},Y^{-1}) \right]_{0,b} \Big\} - \mathbf{r}^{\textup{con}} \nonumber \\
&= \sum_{M \vert b} \Big\{ w_1 \left[ 4f^2(X^2,Y^2)f(X^{-2},Y^{-2})\mathbf{v}_2^{\textup{con}} \right]_{0,b} 
+ w_2 \left[ 2f(X^2,Y^2)f^2(X,Y)f^2(X^{-1},Y^{-1})\mathbf{v}_2^{\textup{con}} \right]_{0,b} \nonumber \\
&\hspace{+2.7em}+ w_2 \left[ 4f(X^2,Y^2)f(X^{-2}\hspace{-0.3em},Y^{-2})f^2(X,Y)f(X^{-1}\hspace{-0.3em},Y^{-1})\mathbf{v}_1^{\textup{con}} \right]_{0,b}
+ w_3 \left[ f^2(X,Y)f^4(X^{-1},Y^{-1})\mathbf{v}_{-2}^{\textup{con}} \right]_{0,b} \nonumber \\ 
&\hspace{+2.7em}+ w_3 \left[ 2f(X^2,Y^2)f(X,Y)f^4(X^{-1},Y^{-1}) \mathbf{v}_{-1}^{\textup{con}} \right]_{0,b} 
+ w_3 \left[ 4f(X^2,Y^2)f^2(X,Y)f^3(X^{-1},Y^{-1}) \mathbf{v}_1^{\textup{con}} \right]_{0,b} \nonumber \\ 
&\hspace{+2.7em}+  w_4 \left[ 8f^4(X,Y)f^3(X^{-1},Y^{-1})\mathbf{v}_1^{\textup{con}} \right]_{0,b} \Big\} - \mathbf{r}^{\textup{con}}.
\end{align}
When $N_8(\mathbf{p}^{\textup{con}})$ reaches its local minimum, $\nabla_{\mathbf{p}^{\textup{con}}}  L_8(\mathbf{p}^{\textup{con}}) = \mathbf{0}_{1 \times (m+1)M},$ which directly leads to \eqref{eq_8} by defining $c_i = r_i$ for all $i \in \{0,1,\dots,m\}$.
\end{proof}

%%%%%%%%%%%%%%%%%%%%%%%%%%%%%%%%%%%%%
\subsection{Outcomes Forecasted} \label{subsec: forecast}

\begin{figure}
\centering
\includegraphics[width=0.40\textwidth]{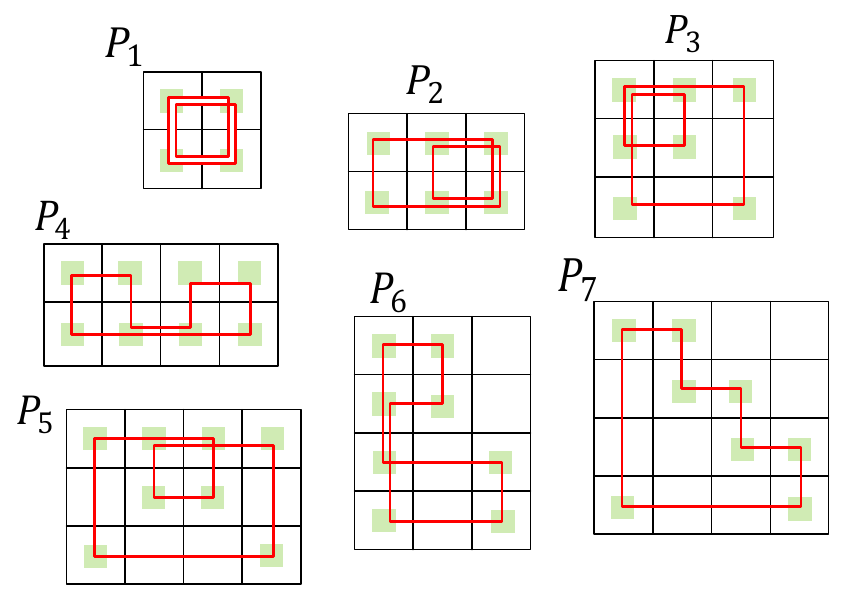} \vspace{-0.8em}
\caption{Seven (out of nine) protograph patterns that can result in cycles-$8$ in the Tanner graph after partitioning/relocations/lifting. The transposes of $P_2$ and $P_4$ are excluded from the list for brevity. Light green squares are the non-zero entries specifying the cycle-candidate edges.} \vspace{-1.0em}
\label{fig: cycle8_patterns}
\end{figure}

In this section, we compute the expected number of short cycles-$k$, $k \in \{6,8\}$, under a specific MD probability distribution, taking into account $L$, $M$, and $z$ (which gives an estimate, typically a useful upper bound, of what the MCMC algorithm can produce; see Section~\ref{subsec: MCMC_algo}). These numbers inform us what to expect from incorporating the probability-distribution matrix in designing GD-MD codes under random partitioning, relocations, and lifting.

\begin{theorem} \label{thm: forecast} After random partitioning and relocations based on a probability-distribution vector $\mathbf{p}^{\textup{con}}$, and random lifting based on a uniform distribution, the expected number $E_6(\mathbf{p}^{\textup{con}})$ of cycles-$6$ in the Tanner graph of $\mathbf{H}_{\textup{MD}}$ is
\begin{align} \label{Eobj6mid} E_6(\mathbf{p}^{\textup{con}}) \approx N_6(\mathbf{p}^{\textup{con}})\cdot \frac{(2L-m)}{2}\cdot M,
\end{align}
and the expected number $E_8(\mathbf{p}^{\textup{con}})$ of cycles-$8$ is 
\begin{align} \label{Eobj8mid}
E_8(\mathbf{p}^{\textup{con}}) \approx N_8(\mathbf{p}^{\textup{con}}) \cdot (L-m) \cdot M,
\end{align}
where $L > \chi$, $\chi$ is the maximum number of consecutive SC protograph replicas a cycle pattern of interest, specifically its VNs, can span, and $\kappa \geq 4$ (see \cite[Lemma~1]{oocpo} and \cite{channel_aware}).
\end{theorem}

\begin{proof}
Let $E_{\textup{obj}}$ be the expected total number of cycle-$k$ candidates out of all cycle-$k$ candidates $\{C_i \,|\, i \in \mathcal{I}\}$ in the all-one base matrix $\mathbf{{H}^{\textup{g}}}$ that remain active after random partitioning, random relocations, and random lifting (see Remark \ref{rmk: order}). Here, $\mathcal{I}$ is the set of candidate indices. Define $\mathbbm{1}(C_i)$ as an indicator function on $C_i$ being active. Then,
\begin{align} 
E_{\textup{obj}} &=\mathbb{E} \bigg[ \sum_{i \in \mathcal{I}} \mathbbm{1}(C_i) \bigg] \nonumber \\
&= \sum_{i \in \mathcal{I}}  \mathbb{E} [ \mathbbm{1}(C_i) ] \nonumber \\
&= \sum_{P \in \mathcal{K}(C)} \sum_{i \in \mathcal{I}(P)}  \mathbb{E}_{P} [ \mathbbm{1}(C_i) ] ,
\end{align}
where $\mathcal{K}(C)$ is the set of all protograph patterns that can produce the cycle-$k$ of interest denoted by $C$, $\mathcal{I}(P)$ is the set of indices of all cycle candidates in $\mathbf{{H}^{\textup{g}}}$ associated with pattern $P$. Moreover,
\begin{align}
\mathbb{E}_{P} [ \mathbbm{1}(C_i) ]= \mathbb{P}(K^{C_i}_{P}, F^{C_i}_{P}) \mathbb{P}(L^{C_i}_{P} \mid  K^{C_i}_{P}, F^{C_i}_{P}),
\end{align}
where $\mathbb{P}(K^{C_i}_{P}, F^{C_i}_{P})$ is the probability that $C_i$ remains active after partitioning and relocations, while $\mathbb{P}(L^{C_i}_{P} \mid  K^{C_i}_{P}, F^{C_i}_{P})$ is the probability that $C_i$ remains active after lifting given that it arrives at this stage active ($K^{C_i}_{P}$, $F^{C_i}_{P}$, and $L^{C_i}_{P}$ are the associated events). Observe that what happens at the lifting stage depends on partitioning and relocations.
\par 
First, note that there is only one pattern $P$, which is a cycle-$6$ itself, in the base matrix that can produce cycles-$6$ in $\mathbf{H}_{\textup{MD}}$, i.e., $|\mathcal{K}(C)|=1$. Furthermore,
\begin{enumerate} 
\item[a)] $\mathbb{P}(K^{C_i}_{P}, F^{C_i}_{P})=P_6(\mathbf{p}^{\textup{con}})$, 
\item[b)] $|\mathcal{I}|= |\mathcal{I}(P)| = 6 \binom{\gamma}{3} \binom{\kappa}{3}$, and
\item[c)] $\mathbb{P}(L^{C_i}_{P} \mid  K^{C_i}_{P}, F^{C_i}_{P})=1/z$. 
\end{enumerate}
Hence and using \eqref{eqn: expected number of cycle-6}, $ E_{\textup{obj}}$ can be expressed as $N_6(\mathbf{p}^{\textup{con}})/z$. Combining this with the fact that an active candidate $C_i$ produces multiple copies after each stage (for example, $z$ copies as a result of lifting), we obtain 
\begin{align} N_6(\mathbf{p}^{\textup{con}})&\cdot (L-s_{\textup{max}}+1)\cdot M \nonumber  \\
& \leq E_6(\mathbf{p}^{\textup{con}}) \nonumber \\
& \leq N_6(\mathbf{p}^{\textup{con}})\cdot (L-s_{\textup{min}}+1)\cdot M,
\end{align}
where $s_{\textup{min}}$ and $s_{\textup{max}}$ are the minimum and maximum spans, in terms of replicas, of a cycle-$6$ pattern in $\mathbf{H}^{\textup{g}}_{\textup{SC}}$. Note that $s_{\textup{min}}=1$ and $s_{\textup{max}}=\chi=m+1$. By assumption, $L\gg m$, and thus we can take $E_6(\mathbf{p}^{\textup{con}})\approx N_6(\mathbf{p}^{\textup{con}})\cdot \frac{(2L-m)}{2}\cdot M$ via averaging $L$ and $L-m$ as an estimate, which is \eqref{Eobj6mid}.
\par 
For the second part of the statement, observe that there are nine patterns in the base matrix that can produce cycles-$8$ in $\mathbf{H}_{\textup{MD}}$ \cite{GRADE, channel_aware}, i.e., $|\mathcal{K}(C)|=9$. We will find the expected number of cycles-$8$ in $\mathbf{H}_{\textup{MD}}$ for pattern $P_4$ shown in Fig.~\ref{fig: cycle8_patterns} as this is a quite rich case.

We first note that
\begin{enumerate} 
\item[a)] $\mathbb{P}(K^{C_i}_{P_4}, F^{C_i}_{P_4})= \sum_{M \vert b} [ f^4(X,Y)f^4(X^{-1},Y^{-1})]_{0,b}$ and
\item[b)] $|\mathcal{I}(P_4)| = 6\binom{\gamma}{2} \binom{\kappa}{4}$.
\end{enumerate}
Next, we show that $(z-1)(z-2)/z^3 \leq \mathbb{P}(L^{C_i}_{P_4} \mid  K^{C_i}_{P_4}, F^{C_i}_{P_4}) \leq 1/z$ for lifting.
\par
Note that $P_4$ is formed via three basic (fundamental) cycles-$4$. Following the arguments of \cite{fossorier} and \cite[Theorem~3]{rohith2}, we consider the following cases for $P_4$ after partitioning and relocations:
\begin{enumerate} 
\item[{2.a)}] The three basic cycles-$4$ are kept.
\item[{2.b)}] Only two basic cycles-$4$ are kept. 
\item[{2.c)}] No cycles-$4$ are kept.
\end{enumerate}
Suppose that the $4$ edges of $P_4$ on the first row are labeled N1, N2, N3, N4, starting from the left, and the $4$ edges on the second row are labeled N5, N6, N7, N8, starting from the left (see Fig.~\ref{fig: cycle8_patterns}). We access these edges in the following order N1, N5, N6, N2, N3, N7, N8, N4.
\par
In Case 2.a), suppose we start from the edge N1, and we trace the possible lifting choices for edges N2, N7, and N4 at the end of each basic cycle traversal. N2 has $z-1$ options for a power of $\boldsymbol{\sigma}$ in order not to form a cycle-$4$ on the left. N7 has $z-2$ options in order not to form a cycle-$4$ from the disjunctive union of the two left-most basic cycles. Finally, N4 has a single option to close the walk and form a cycle-$8$ \cite{fossorier}. Hence, $\mathbb{P}(L^{C_i}_{P_4} \mid  K^{C_i}_{P_4}, F^{C_i}_{P_4})= (z-1)(z-2)(1)/z^3=(z-1)(z-2)/z^3$.
\par
In Case 2.b), suppose we start from the edge N1, and we trace the possible lifting choices for edges N2, N7, and N4 at the end of each basic cycle traversal. We assume the cycle that ends with N7 is the one that no longer exists. N2 has $z-1$ options for a power of $\boldsymbol{\sigma}$ in order not to form a cycle-$4$ on the left. N7 is free to receive any of the $z$ options. Finally, N4 has a single option to close the walk and form a cycle-$8$ \cite{fossorier}. Hence, $\mathbb{P}(L^{C_i}_{P_4} \mid  K^{C_i}_{P_4}, F^{C_i}_{P_4})= (z-1)z(1)/z^3=(z-1)/z^2$. 
\par
In Case 2.c), it is clear that $\mathbb{P}(L^{C_i}_{P_4} \mid  K^{C_i}_{P_4}, F^{C_i}_{P_4})=1/z$. 
\par
Note that $\mathbb{P}(L^{C_i}_{P_4} \mid  K^{C_i}_{P_4}, F^{C_i}_{P_4})$ is closer to $1/z$ because it is more likely to have no cycles-$4$, for $m > 1$, after random partitioning (in which case, $C_i$ becomes a true cycle-$8$ in the MD protograph). Therefore, we take $\mathbb{P}(L^{C_i}_{P_4} \mid  K^{C_i}_{P_4}, F^{C_i}_{P_4})=1/z$ as an estimate. The other patterns can be treated in a similar fashion, and thus we take $\mathbb{P}(L^{C_i}_{P} \mid  K^{C_i}_{P}, F^{C_i}_{P})=1/z$ for all possible cycle-$8$ protograph patterns.\footnote{Note that in all cases where the probability of remaining active is not exactly $1/z$, it still is in the order of $O(1/z)$. Thus, for sufficiently large $z$, using $1/z$ as an approximation is generally justified.} Hence, $ E_{\textup{obj}}$ is approximated by $N_8(\mathbf{p}^{\textup{con}})/z$. Consequently, we obtain 
\begin{align} N_8(\mathbf{p}^{\textup{con}})&\cdot (L-s_{\textup{max}}+1)\cdot M \nonumber  \\
& \leq E_8(\mathbf{p}^{\textup{con}}) \nonumber \\
& \leq N_8(\mathbf{p}^{\textup{con}})\cdot (L-s_{\textup{min}}+1)\cdot M,
\end{align}
where $s_{\textup{min}}$ and $s_{\textup{max}}$ are the minimum and maximum spans of a cycle-$8$ pattern in $\mathbf{H}^{\textup{g}}_{\textup{SC}}$. Note that $s_{\textup{min}}=1$ and $s_{\textup{max}}=\chi=2m+1$ due to the existence of structures $P_6$ and $P_7$ (only $P_7$) after partitioning in the case of $\gamma \geq 4$ ($\gamma = 3$). By assumption, $L\gg m$, and thus we can take $E_8(\mathbf{p}^{\textup{con}}) \approx N_8(\mathbf{p}^{\textup{con}})\cdot (L-m)\cdot M$ via averaging $L$ and $L-2m$ as an estimate, which is \eqref{Eobj8mid}.
\end{proof}

\begin{remark}\label{rmk4} Note that the upper bound of $E_8(\mathbf{p}^{\textup{con}})$ in the proof above is exact. One obtains an exact lower bound as well by studying the probabilities $\mathbb{P}(L^{C_i}_{P} \mid  K^{C_i}_{P}, F^{C_i}_{P})$ for each cycle-$8$ pattern, and it is in fact lower than $N_8(\mathbf{p}^{\textup{con}}) \cdot (L-2m) \cdot M$. However, the expected number $E_8(\mathbf{p}^{\textup{con}})$ is always notably closer to the upper bound.
\end{remark}

We now discuss the value of Theorem~\ref{thm: forecast} to the coding theory community.
\begin{enumerate}
\item The GD-MD codes we design will have lower numbers of cycles (of length $=$ code girth) than the upper bound of the expected number when the partitioning and lifting matrices are designed as in \cite{GRADE}, \cite{oocpo} and \cite{reins_mcmc}. This is expected to hold in general when we use optimized partitioning and lifting matrices designed for the underlying SC~codes. As we shall see, our estimate can get close to the actual FL cycle counts.
\item The expected number of cycles-$k$, $k \in \{6,8\}$, in $\mathbf{H}_{\textup{MD}}$ decreases with the relocation percentage $\mathcal{T}$. However, the rate of reduction becomes close to $0$ after some $\mathcal{T}$ value. Hence, one can locate a threshold vector $\mathbf{p}_{\textup{thold}}$ where $E_k(\mathbf{p}_{\textup{thold}})$ is at most, say, $5\%$ more than the minimum $E_k(\mathbf{p}^{\textup{con}})$, which is reached by minimizing $N_k(\mathbf{p}^{\textup{con}})$. This allows us to specify an informed MD density $\mathcal{T}$ for the MCMC algorithm (within the margin governed by decoding latency) to potentially decrease its computational complexity.
\item $E_{\textup{obj}}$ in Theorem~\ref{thm: forecast} proof gives the expected number of  cycle-$k$ candidates in the all-one base matrix $\mathbf{{H}^{\textup{g}}}$ that remain active after random partitioning, random relocation, and random lifting. In case this number is single-digit (at $\mathbf{p}_{\textup{thold}}$ for example), it is likely that the MCMC algorithm will eliminate all active cycles and produce an MD-SC code with girth $\geq k+2$. Provided that cycles (of length $=$ code girth) can in fact be entirely eliminated, $E_{\textup{obj}}$ leads us to a threshold for the MD density to achieve this goal.
\item Having the MD-SC degrees of freedom $m$, $L$, and $M$, we can choose these parameters to achieve best performance subject to the design constraints such as the length and rate. For example, if one has an allowed range for the length of the MD-SC code to design while the rate is fixed, $M$ can be determined so that $E_{\textup{obj}}$, and hence $E_k(\mathbf{p}^{\textup{con}})$, is as low as possible while the length is within the allowed range. This, of course, requires the code designer to run the MD-GRADE algorithm multiple times (see Section~\ref{subsec: MD-GRADE algo}).
\end{enumerate}

%%%%%%%%%%%%%%%%%%%%%%%%%%%%%%%%%%%%%%%%%%%%%%%%%%%%%%%%%%%%%%%%%%%%
\section{Theoretical Framework for Objects} \label{sec: theory}

In this section, we generalize the probabilistic framework previously introduced for minimizing cycle counts in MD-SC codes. In particular, we extend our analysis to arbitrary subgraphs in the Tanner graph of the code. We introduce this generalization because, although analyzing cycles is useful, cycles are not always as detrimental as other structures, particularly in codes with VN degree $4$ or higher, where isolated cycles may have limited impact on the performance. Therefore, targeting more harmful structures during optimization can yield greater performance improvements. This approach also enables more efficient exploitation of the degrees of freedom provided by the relocation operation.

\subsection{General Case Analysis}

In this subsection, we provide a closed-form expression for the expected number of instances of any subgraph with an arbitrary topology after random partitioning and relocation operations, given the probability-distribution matrix $\mathbf{P}$. From this point onward, we will refer to such subgraphs as objects, which is customary for LDPC codes.

An object pattern is a generalization of a cycle pattern, where it represents an assignment of CNs and VNs of an object to the rows and columns of the base matrix $\mathbf{{H}^{\textup{g}}}$. Each such assignment induces an equivalence relation over the CNs and VNs, and consequently over the edges of the object. Hence, sets of equivalence classes over CNs, VNs, and edges are produced. Unlike the case of cycles, an object pattern inherently includes the way of edge traversal. Object patterns that define the same equivalence classes form what we call an object pattern class. The formal definitions of object pattern, object pattern class, and equivalence class are provided below (also given in \cite{GRADE} for SC codes).

\begin{definition}\label{def: object_pattern_class} (Object pattern class). 
Let an object be represented by a bipartite graph $G(V, C, E)$, where $V$ and $C$ denote the sets of VNs and CNs, respectively. The edge set $E$ consists of elements $e_{i,j}$ for $i \in C$ and $j \in V$, indicating an edge between CN $i$ and VN $j$ in the object graph. 

An object pattern is defined as an assignment $P = (g_V, g_C)$, where $g_V : V \rightarrow \{0,1,\dots,\kappa-1\}$ and $g_C : C \rightarrow \{0,1,\dots,\gamma-1\}$ ($\gamma$ and $\kappa$ are the dimensions of the base matrix of the code), subject to the following conditions:
\begin{enumerate}
    \item For every $c \in C$ and any pair $v_1, v_2 \in V$ such that $e_{c,v_1}, e_{c,v_2} \in E$, it holds that $g_V(v_1) \neq g_V(v_2)$.
    \item For every $v \in V$ and any pair $c_1, c_2 \in C$ such that $e_{c_1,v}, e_{c_2,v} \in E$, it holds that $g_C(c_1) \neq g_C(c_2)$.
\end{enumerate}
Here, the first (second) condition implies that $2$ VNs (CNs) that are connected to same CN (VN), cannot be mapped to same value in the base matrix.

The collection of all such valid object patterns that additionally satisfy the below conditions is called the object pattern class associated with $(\mathcal{V}, \mathcal{C})$, denoted by $\mathcal{P}(\mathcal{V}, \mathcal{C})$, where $\mathcal{V}$ and $\mathcal{C}$ denote the sets of equivalence classes (or partitions) over $V$ and $C$, respectively. The additional conditions are:
\begin{enumerate}
    \item For any $v_1, v_2 \in V$, we have $g_V(v_1) = g_V(v_2)$ if and only if $v_1 \sim v_2$ in the set $\mathcal{V}$, i.e., $v_1$ and $v_2$ are equivalent.
    \item For any $c_1, c_2 \in C$, we have $g_C(c_1) = g_C(c_2)$ if and only if $c_1 \sim c_2$ in the set $\mathcal{C}$, i.e., $c_1$ and $c_2$ are equivalent.
\end{enumerate}

We define the equivalence relation over the edges by $e_{c_1,v_1} \sim e_{c_2,v_2}$ if and only if $v_1 \sim v_2$ in $\mathcal{V}$ and $c_1 \sim c_2$ in $\mathcal{C}$, for all $v_1, v_2 \in V$ and $c_1, c_2 \in C$. The set of all equivalence classes induced by this relation is denoted by $\mathcal{E}(\mathcal{V}, \mathcal{C})$.
\end{definition}

For an object represented by the bipartite graph $G(V, C, E)$, and an object pattern class over $G$ denoted by $\mathcal{P}(\mathcal{V}, \mathcal{C})$, with the set of equivalence classes $\mathcal{E}(\mathcal{V}, \mathcal{C})$, we can say that two VNs (CNs) are equivalent under $\mathcal{V}$ ($\mathcal{C}$) if they correspond to the same column (row) in the base matrix according to the object pattern assignment. Two edges in the bipartite graph are considered equivalent if they correspond to the same entry in the base matrix under this definition. Each element of $\mathcal{E}(\mathcal{V}, \mathcal{C})$ thus represents a specific entry in the base matrix that is covered by the given object pattern. This leads to the notion of a matrix representation for an object pattern or an object pattern class. Examples illustrating matrix representations of various object patterns are provided later in this section.

Note that, whenever a set of edges corresponds to the same entry in the base matrix, we also consider they correspond to the same entry in the partitioning and relocation matrices. Consequently, these edges are assigned identically with respect to both partitioning and relocation, provided they belong to the same equivalence class.

\begin{definition} (Cycle basis). \label{def:cycle_basis}
A cycle basis of an object is a minimum-cardinality set of cycles such that every cycle in the object can be expressed as the disjunctive union of cycles in this set. The disjunctive union of cycles refers to the subgraph consisting of all edges that appear in an odd number of cycles, along with their neighboring VNs and CNs. We will refer to cycles in this cycle basis as fundamental cycles.
\end{definition}

The concept of a cycle basis is particularly important because the activeness of each object after partitioning and relocations can be determined by the activeness of its fundamental cycles, as established in \cite[Theorem~1]{rohith2}. Therefore, the probability of the object pattern being active can be expressed as the joint probability that all of its fundamental cycles are active.

To express this probability, we construct a polynomial referred to as the characteristic polynomial of the object pattern class. Summation of the certain coefficients of this polynomial correspond to the joint probability that all fundamental cycles remain active under random partitioning and relocations based on probability-distribution matrix $\mathbf{P}$, and hence the probability that the object pattern remains active under these conditions. The formal definition of the characteristic polynomial and the proof of this result are presented in Lemma~\ref{lemma:charpol}. The reader may also refer to Example~\ref{example1} in the next subsection to see how the characteristic polynomial is generated for a specific object pattern.

\begin{lemma} (Characteristic polynomial of an object pattern class). \label{lemma:charpol} 
Let $G(V, C, E)$ be the bipartite graph representation of an object, and let $\mathcal{P}(\mathcal{V}, \mathcal{C})$ denote the corresponding object pattern class. Suppose $\mathcal{E}(\mathcal{V}, \mathcal{C})$ is set of the equivalence classes on the edges $E$ induced by the partitions $\mathcal{V}$ and $\mathcal{C}$. Let $S$ represent the cycle basis of the graph $G$.

Define the mapping $\delta : E \times S \rightarrow \{-1, 0, 1\}$ as follows. For any cycle $s \in S$ given by the sequence $(c_1, v_1, c_2, v_2, \dots, c_g, v_g)$ and any edge $e \in E$,
\begin{align}
\delta_{e, s} = 
\begin{cases} 
1 & \text{if } e = e_{c_i,v_i} \text{ for } \exists i \in \lbrace 1,2,\dots,g\rbrace, \\ 
-1 & \text{if } e = e_{c_i,v_{i+1}} \text{ for } \exists i \in \lbrace 1,2,\dots,g\rbrace, \\ 
0 & \text{otherwise}.
\end{cases}
\end{align}
Note that the cycle $s$ is such that $c_i$ connects $v_i$ and $v_{i+1}$, for all $i \in \{1,2,\dots,g\}$, and $v_{g+1} = v_1$.

Let $h(\mathbf{X}, \mathbf{Y}; G \mid \mathcal{V}, \mathcal{C})$ be the characteristic polynomial associated with the object pattern class $\mathcal{P}(\mathcal{V}, \mathcal{C})$, defined as:
\begin{align}
h(\mathbf{X}, \mathbf{Y}; G \mid \mathcal{V}, \mathcal{C}) 
= \prod_{\bar{e} \in \mathcal{E}(\mathcal{V}, \mathcal{C})} 
f \left( \prod_{e \in \bar{e}} \prod_{s \in S} X_s^{\delta_{e, s}}, 
\prod_{e \in \bar{e}} \prod_{s \in S} Y_s^{\delta_{e, s}} \right).
\label{eq:char_func}
\end{align}
Then, the sum of the coefficients of the polynomial in \eqref{eq:char_func} corresponding to monomials with no $X$ components and with $Y$ exponents that are integer multiples of $M$ gives the probability that an object pattern from the class $\mathcal{P}(\mathcal{V}, \mathcal{C})$ remains active in the MD protograph after random partitioning and relocation operations.
\end{lemma}

\begin{proof}
For simplicity, we use $ \mathcal{E} $ instead of $\mathcal{E}(\mathcal{V}, \mathcal{C})$ in the proof. We define $ \mathbf{i} \in \lbrace 0,1,\ldots,m\rbrace^{|\mathcal{E}|}$ and  $ \mathbf{j} \in \lbrace 0,1,\ldots,M-1\rbrace^{|\mathcal{E}|}$, where the elements of each of these vectors are denoted by $ i_{\overline{e}} $ and $ j_{\overline{e}} $, respectively, for $\exists \overline{e} \in \mathcal{E} $.

Since all the edges that belong to the same equivalence class correspond to the same entry on the base, partitioning and relocation matrices, these edges are assigned identically for their partitioning and relocations. In particular, all the edges in the equivalence class $\overline{e}$ are assigned with $ i_{\overline{e}} $ for partitioning and with $ j_{\overline{e}} $ for relocations. Then,
\begin{align} \label{eq:char_open}
&h(\mathbf{X},\mathbf{Y}; G \mid \mathcal{V}, \mathcal{C}) = \prod_{\bar{e} \in \mathcal{E}} 
f \left( \prod_{e \in \bar{e}} \prod_{s \in S} X_s^{\delta_{e, s}}, \prod_{e \in \bar{e}} \prod_{s \in S} Y_s^{\delta_{e, s}} \right) = \prod_{\overline{e} \in \mathcal{E}} \left[\sum_{i=0}^m \sum_{j=0}^{M-1} p_{i,j} \cdot \prod_{e\in\overline{e}} \prod_{s\in S} X_s^{\delta_{e, s} i} Y_s^{\delta_{e, s} j} \right] \nonumber \\
&=\sum_{\substack{\mathbf{i} \in \lbrace 0,1,\ldots,m\rbrace^{|\mathcal{E}|}, \\ \mathbf{j} \in \lbrace 0,1,\ldots,M-1\rbrace^{|\mathcal{E}|}}} 
\prod_{\overline{e} \in \mathcal{E}} \left[ p_{i_{\overline{e}},j_{\overline{e}}}  
\prod_{e\in\overline{e}} \prod_{s\in S} X_s^{\delta_{e, s} i_{\overline{e}}} \, Y_s^{\delta_{e, s} j_{\overline{e}}} \right] = \sum_{\mathbf{i}, \mathbf{j}} \left( \prod_{\overline{e} \in \mathcal{E}} p_{i_{\overline{e}},j_{\overline{e}}} \right) 
\prod_{\overline{e} \in \mathcal{E}} \prod_{e \in \overline{e}} \prod_{s \in S} 
X_s^{\delta_{e, s} i_{\overline{e}}} \, Y_s^{\delta_{e, s} j_{\overline{e}}} \nonumber \\
&= \sum_{\mathbf{i}, \mathbf{j}} \left( \prod_{\overline{e} \in \mathcal{E}} p_{i_{\overline{e}},j_{\overline{e}}} \right) 
\prod_{s \in S} \left[ X_s^{\sum_{\overline{e} \in \mathcal{E}} i_{\overline{e}} \sum_{e \in \overline{e}} \delta_{e, s}} 
\, Y_s^{\sum_{\overline{e} \in \mathcal{E}} j_{\overline{e}} \sum_{e \in \overline{e}} \delta_{e, s}} \right] = \sum_{\mathbf{i}, \mathbf{j}} \left( \prod_{\overline{e} \in \mathcal{E}} p_{i_{\overline{e}},j_{\overline{e}}} \right) \prod_{s \in S} X_s^{k_s(\mathbf{i})} \, Y_s^{m_s(\mathbf{j})},
\end{align}
where
\begin{align}
k_s(\mathbf{i}) = \sum_{\overline{e} \in \mathcal{E}} i_{\overline{e}} \sum_{e\in\overline{e}} \delta_{e,s} = \sum_{e\in\overline{e}, \, \overline{e}\in \mathcal{E}} i_{\overline{e}} \, \delta_{e,s},
\label{eq:ks}
\end{align}
 
\begin{align}
m_s(\mathbf{j}) = \sum_{\overline{e} \in \mathcal{E}} j_{\overline{e}} \sum_{e\in\overline{e}} \delta_{e,s} = \sum_{e\in\overline{e}, \, \overline{e}\in \mathcal{E}} j_{\overline{e}} \, \delta_{e,s}.
\label{eq:ms}
\end{align}

Equations \eqref{eq:ks} and \eqref{eq:ms} are referred to as the alternating sums of the assignments $\mathbf{i}$ and $\mathbf{j}$ on cycle $s$ for partitioning and relocations, respectively. It can be observed that for cycle $s$, $k_s(\mathbf{i})$ equals the left-hand side of \eqref{eq:partition_condition}, and $m_s(\mathbf{j})$ equals the left-hand side of \eqref{eq:relocation_condition}. Additionally, the term $\prod_{\overline{e} \in \mathcal{E}} p_{i_{\overline{e}},j_{\overline{e}}}$ appearing in \eqref{eq:char_open} represents the joint probability of the assignments $\mathbf{i}$ and $\mathbf{j}$ for the object pattern, given that partitioning and relocations are performed randomly based on $\mathbf{P}$.

A cycle $s$ remains active after partitioning and relocation operations if and only if $k_s(\mathbf{i})$ equals $0$ and $m_s(\mathbf{j})$ equals an integer multiple of $M$. An object pattern remains active if and only if this is the case for all $s\in S$. 

Let $[h(\mathbf{X}, \mathbf{Y}; G \mid \mathcal{V}, \mathcal{C})]_{\mathbf{a}, \mathbf{b}}$ denote the coefficient of the polynomial corresponding to the monomial in which the degrees of the $X$ (resp., $Y$) variables are given by the vector $\mathbf{a} \in \mathbb{Z}^{|S|}$ (resp., $\mathbf{b} \in \mathbb{Z}^{|S|}$), with each element $a_s$ (resp., $b_s$) representing the power of $X_s$ (resp., $Y_s$). For brevity, we denote the situation that every element of the vector $\mathbf{b}$ is divisible by $M$ as $M \mid \mathbf{b}$. The probability that an object pattern from the class $\mathcal{P}(\mathcal{V}, \mathcal{C})$ remains active in the MD protograph after random partitioning and relocation operations is the summation of probabilities over all partitioning and relocation assignments that satisfy the conditions required for the object pattern to be active. This probability is then:
\begin{align}
\sum_{\substack{\mathbf{i}, \mathbf{j}: \forall s\in S, \\\ k_s(\mathbf{i}) = 0, \ M\mid m_s(\mathbf{j})}}  
\prod_{\overline{e} \in \mathcal{E}} p_{i_{\overline{e}},j_{\overline{e}}}
= \sum_{M \mid \mathbf{b}} \left[\sum_{\mathbf{i}, \mathbf{j}} \left( \prod_{\overline{e} \in \mathcal{E}} p_{i_{\overline{e}},j_{\overline{e}}} \right) \prod_{s \in S} X_s^{k_s(\mathbf{i})} \, Y_s^{m_s(\mathbf{j})}\right]_{\mathbf{0},\mathbf{b}}
= \sum_{M \mid \mathbf{b}} \left[ h(\mathbf{X}, \mathbf{Y}; G \mid \mathcal{V}, \mathcal{C}) \right]_{\mathbf{0},\mathbf{b}}, 
\label{eq:prob}
\end{align}
where the last equality follows from \eqref{eq:char_open}, and it completes the proof.
\end{proof}

In general, $[h(\mathbf{X}, \mathbf{Y}; G \mid \mathcal{V}, \mathcal{C})]_{\mathbf{a}, \mathbf{b}}$ is the probability that the alternating sums of fundamental cycles for partitioning (resp., relocations) take the values in the vector $\mathbf{a}$ (resp., $\mathbf{b}$). Our interest is of course in a specific case, where for partitioning, the only possible option for the alternating sums of fundamental cycles is to equal $0$. For relocations however, the available options are all integer multiples of $M$, since the activeness condition is defined modulo $M$. Hence, we perform the summation over all $\mathbf{b}$, where $M \mid b_i$ for all $i$.

As we discussed, each object pattern in a given class corresponds to an assignment of VNs and CNs to the sets $\{0,1,\dots,\kappa-1\}$ and $\{0,1,\dots,\gamma-1\}$, respectively. However, due to the inherent symmetries of the object patterns, certain permutations of assigned values can result in isomorphic object patterns. These symmetry-preserving transformations are automorphisms of $G$ under the object pattern class $\mathcal{P}(\mathcal{V}, \mathcal{C})$. It follows that any automorphism of $G$ under $\mathcal{P}(\mathcal{V}, \mathcal{C})$ is a graph isomorphism that preserves the equivalence classes denoted by $\mathcal{V}$ and $\mathcal{C}$. Therefore, each automorphism can be represented as a pair of permutations acting on the equivalence classes $\mathcal{V}$ and $\mathcal{C}$. The set of all such automorphisms under a given object pattern class forms a group, and it is formally defined in Definition~\ref{def:automorphism}.

\begin{definition} (Automorphism group of an object under an object pattern class).
\label{def:automorphism}
Let an object be represented by the bipartite graph $G(V, C, E)$, and let $\mathcal{P}(\mathcal{V}, \mathcal{C})$ denote an object pattern class over $G$. An automorphism of $G$ under $\mathcal{P}(\mathcal{V}, \mathcal{C})$ is defined as a pair of bijections $(\pi_V, \pi_C)$, where $\pi_V : V \rightarrow V$ and $\pi_C : C \rightarrow C$, such that the following conditions are satisfied:
\begin{enumerate}
    \item For all $v \in V$ and $c \in C$, we have $e_{c,v} \in E$ if and only if $e_{\pi_C(c),\pi_V(v)} \in E$.
    \item For all $v_1, v_2 \in V$, $v_1 \sim v_2$ if and only if $\pi_V(v_1) \sim \pi_V(v_2)$.
    \item For all $c_1, c_2 \in C$, $c_1 \sim c_2$ if and only if $\pi_C(c_1) \sim \pi_C(c_2)$.
\end{enumerate}

The set of all such automorphisms over $G$ under $\mathcal{P}(\mathcal{V}, \mathcal{C})$ forms a group under composition of permutations. This group is referred to as the \textit{automorphism group} of $G$ under $\mathcal{P}(\mathcal{V}, \mathcal{C})$, and denoted by $\mathrm{Aut}(G \mid \mathcal{V}, \mathcal{C})$.
\end{definition}

We can now use the notions of characteristic polynomial and automorphism group of an object pattern class to introduce the characteristic polynomial of an object by enumerating all patterns in the protograph base matrix. As in the case of object pattern classes, we can use this characteristic polynomial to express the expected number of active objects in the MD-SC protograph after random partitioning and relocations. The exact expression and its proof are provided in Theorem~\ref{thm: char_poly_object}. Here, expectations are $L$-invariant and $M$-invariant.

\begin{theorem} \label{thm: char_poly_object} (Characteristic polynomial of an object).
Given the bipartite graph $G(V, C, E)$ representing an object, let $\mathcal{B}(G)$ denote the set of all object pattern classes of $G$. Define the characteristic polynomial $h(\mathbf{X}, \mathbf{Y}; G)$ of the object $G$ as:
\begin{align}\label{eqn_char_obj}
&h(\mathbf{X}, \mathbf{Y}; G) = \sum_{\mathcal{P}(\mathcal{V}, \mathcal{C}) \in \mathcal{B}(G)} \frac{|\mathcal{V}|! |\mathcal{C}|!}{| \mathrm{Aut}(G \mid \mathcal{V}, \mathcal{C}) |} \binom{\kappa}{|\mathcal{V}|} \binom{\gamma}{|\mathcal{C}|} h(\mathbf{X}, \mathbf{Y}; G \mid \mathcal{V}, \mathcal{C}).
\end{align}
Then, the expression $\sum_{M \mid \mathbf{b}} [h(\mathbf{X}, \mathbf{Y}; G)]_{\mathbf{0},\mathbf{b}}$ equals the expected number of active object patterns of $G$ after partitioning and relocations in the MD-SC protograph.
\end{theorem}

\begin{proof}
The total number of possible assignments of node indices in $ G $ that result in the partitions specified by $ (\mathcal{V}, \mathcal{C}) $ is given by $|\mathcal{V}|! \, |\mathcal{C}|! \binom{\kappa}{|\mathcal{V}|} \binom{\gamma}{|\mathcal{C}|}$, since the MD-SC code is constructed from an all-one base matrix of size $\gamma\times\kappa$. However, each distinct assignment is counted multiple times, specifically exactly $ |\text{Aut}(G \mid \mathcal{V}, \mathcal{C})| $ times, due to the presence of automorphisms. Therefore, the cardinality of the object pattern class $ \mathcal{P}(\mathcal{V}, \mathcal{C}) $ is obtained by dividing the total number of assignments by the size of the automorphism group, leading to:
\begin{equation}\label{eqn_card_class}
| \mathcal{P}(\mathcal{V}, \mathcal{C}) | = \frac{|\mathcal{V}|! |\mathcal{C}|!}{| \mathrm{Aut}(G \mid \mathcal{V}, \mathcal{C}) |} \binom{\kappa}{|\mathcal{V}|} \binom{\gamma}{|\mathcal{C}|}.
\end{equation}

For any object pattern $P$ of $G$, let $X_P$ be a Bernoulli random variable indicating whether $P$ is active. Define the random variable $X = \sum_P X_P$, which is the total number of active patterns across all object patterns $P$ of $G$ in the protograph. Then,
\begin{align}
&\mathbb{E}[X] = \sum_{\substack{P \in \mathcal{P}(\mathcal{V}, \mathcal{C}), \\ \mathcal{P}(\mathcal{V}, \mathcal{C}) \in \mathcal{B}(G)}} \mathbb{E}[X_P] =  
\sum_{\substack{P \in \mathcal{P}(\mathcal{V}, \mathcal{C}), \\ \mathcal{P}(\mathcal{V}, \mathcal{C}) \in \mathcal{B}(G)}} \mathbb{P}[X_P = 1] = \sum_{\substack{P \in \mathcal{P}(\mathcal{V}, \mathcal{C}), \\ \mathcal{P}(\mathcal{V}, \mathcal{C}) \in \mathcal{B}(G)}} 
\sum_{M \mid \mathbf{b}} [h(\mathbf{X}, \mathbf{Y}; G \mid \mathcal{V}, \mathcal{C})]_{\mathbf{0},\mathbf{b}} \nonumber \\
&= \sum_{\mathcal{P}(\mathcal{V}, \mathcal{C}) \in \mathcal{B}(G)} | \mathcal{P}(\mathcal{V}, \mathcal{C}) | \sum_{M \mid \mathbf{b}} \left[h(\mathbf{X}, \mathbf{Y}; G \mid \mathcal{V}, \mathcal{C})\right]_{0,\mathbf{b}} = \sum_{M \mid \mathbf{b}} \; \sum_{\mathcal{P}(\mathcal{V}, \mathcal{C}) \in \mathcal{B}(G)} | \mathcal{P}(\mathcal{V}, \mathcal{C}) | \; \left[h(\mathbf{X}, \mathbf{Y}; G \mid \mathcal{V}, \mathcal{C})\right]_{\mathbf{0},\mathbf{b}} \nonumber \\
&= \sum_{M \mid \mathbf{b}} [h(\mathbf{X}, \mathbf{Y}; G)]_{\mathbf{0},\mathbf{b}},
\end{align}
where the last equality follows from \eqref{eqn_char_obj} and \eqref{eqn_card_class}, and it completes the proof.
\end{proof}

\subsection{Cycle Concatenation Analysis}

Our main goal is to reduce the number of absorbing and trapping sets in the Tanner graph of the MD-SC code, since these objects dominate the error profile of LDPC codes in the error-floor region~\cite{absorbing, richardson_error_floor}. These objects are also harmful in the waterfall region as they create decoding dependencies and they subsume (or they are substructures of) low-weight codewords. Instead of developing an optimization scheme that targets every possible absorbing set, which could be inefficient and is highly complicated, we concentrate on certain elementary structures formed by concatenating two short cycles (specifically, cycles of lengths $6$ and $8$), as these commonly occur as subgraphs within the most problematic absorbing sets.

Since absorbing sets have a cycle basis for each, focusing on cycles alone might initially seem logical. However, as discussed earlier, isolated cycles that are not directly connected to other cycles may not be very problematic, especially when the VN degrees are greater than three. Therefore, allocating all degrees of freedom offered by the partitioning and relocation operations solely to cycles is not efficient, which is also demonstrated in \cite{GRADE} and \cite{rohith2}.

Based on this discussion, we construct our optimization problem over configurations obtained by either concatenating two cycles-$6$, one cycle-$6$ and one cycle-$8$, or two cycles-$8$. Our main focus for these configurations is on the case where the two cycles concatenate through a chain of VN-CN-VN, since the population of these is higher compared with other types of concatenations such as VN-CN-VN-CN-VN, and they appear more frequently in detrimental absorbing sets. We refer to these objects as $6$-$6$, $6$-$8$, and $8$-$8$ configurations, respectively, where the respective bipartite graphs are also shown in Fig.~\ref{fig:concat}. Note that throughout the paper in the bipartite graph drawings, circles represent VNs and squares represent CNs. It should be also observed that for $6$-$6$, $6$-$8$, and $8$-$8$ configurations, concatenated cycles also form a basis for these objects.

\begin{figure}[hb!]
\vspace{-0.5em}
\centering
\includegraphics[width=0.5\textwidth]{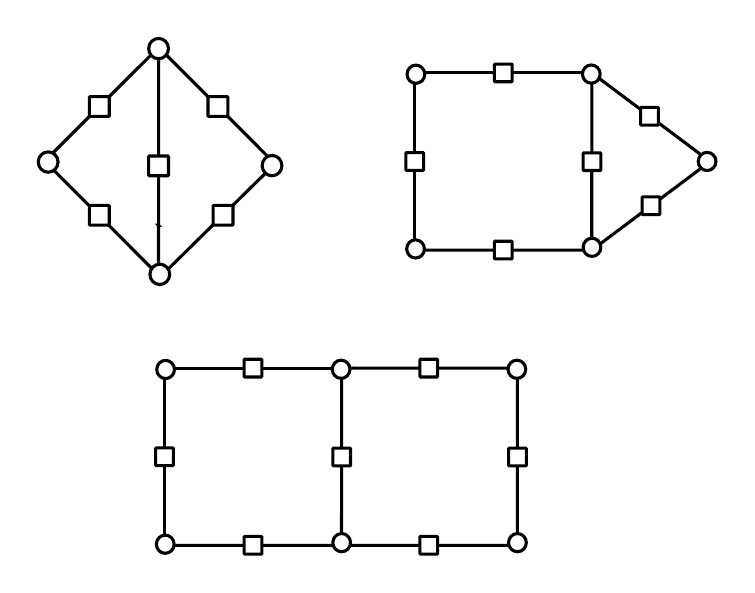} \vspace{-0.8em}
\caption{Bipartite graphs of $6$-$6$ (top left), $6$-$8$ (top right), and $8$-$8$ (bottom) configurations.} \vspace{-1.0em}
\label{fig:concat}
\end{figure}

\begin{example} \label{example1}
We now focus on an example object shown in Fig.~\ref{fig:68}, which depicts the bipartite graph of the object alongside the matrix representation of a pattern that belongs to one of its object pattern classes. In this example, the VNs and CNs of the object are partitioned into five equivalence classes, labeled as $v_i$ and $c_i$, respectively, for $i \in \{1,2,\dots,5\}$. Each VN class is assigned to a column and each CN class is assigned to a row in the matrix representation. All edges in the bipartite graph that belong to the same equivalence class correspond to the same entry in the matrix representation.

For each fundamental cycle, we present its cycle candidate in the matrix representation to illustrate how the corresponding cycle pattern is traversed. We use the variables $(X_1, Y_1)$ to represent the first fundamental cycle (a cycle-$8$ in this case), which is highlighted in red, and $(X_2, Y_2)$ to represent the second fundamental cycle (a cycle-$6$), which is highlighted in blue.

To construct the characteristic polynomial, we use the coupling polynomial $f(X,Y)$ for each equivalence class, that is, for each unique entry in the matrix representation, taking into account the two fundamental cycles. We label each edge on the bipartite graph with monomials $X_1^{i_1} X_2^{i_2}$ and $Y_1^{i_1} Y_2^{i_2}$, where the exponents $i_1$ and $i_2$ are determined based on the presence of the corresponding edge in the fundamental cycles. In particular, if an edge does not participate in cycle $s$, then $i_s = 0$; otherwise, the powers alternate between $1$ and $-1$ for consecutive edges of that cycle.

The inputs of the coupling polynomial are then evaluated by multiplying the $X_1^{i_1} X_2^{i_2}$ and $Y_1^{i_1} Y_2^{i_2}$ terms for all edges belonging to same equivalence class. The final expression for the characteristic polynomial is obtained by multiplying the coupling polynomials of all equivalence classes, i.e., of all unique matrix entries.

Following this procedure and referring to Fig.~\ref{fig:68}, where the powers of the $X$ variables are indicated for each edge in the bipartite graph and for each matrix entry ($Y$ powers are omitted for brevity), the characteristic polynomial corresponding to this object pattern class is given by $f(X_1 X_2,Y_1 Y_2) f(X_1^{-1} X_2^{-1},Y_1^{-1} Y_2^{-1}) f^3(X_1,Y_1) f^3(X_1^{-1},Y_1^{-1}) f^2(X_2,Y_2) f^2(X_2^{-1},Y_2^{-1})$.

\end{example}

\begin{figure*}[ht!]
\centering
\includegraphics[width=0.9\textwidth]{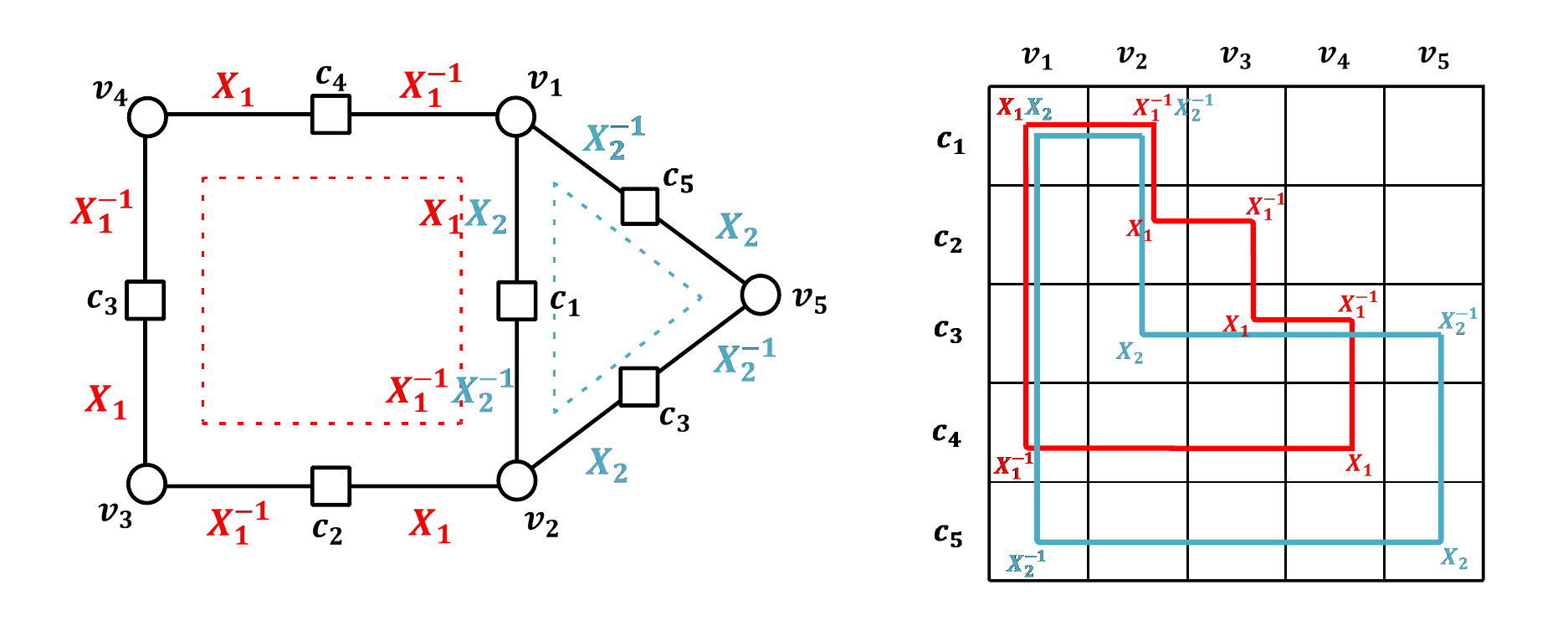}
\vspace{-0.8em}
\caption{Bipartite graph of a $6$-$8$ configuration (left) and matrix representation of one of its object patterns (right).}
\vspace{-1.0em}
\label{fig:68}
\end{figure*}

\begin{definition} (Dominant object pattern classes).
We denote the graphs corresponding to the $6$-$6$, $6$-$8$, and $8$-$8$ configurations by $G_{6,6}$, $G_{6,8}$, and $G_{8,8}$, respectively. An object pattern class $\mathcal{P}(\mathcal{V},\mathcal{C})$ with the set of equivalence classes $\mathcal{E}(\mathcal{V},\mathcal{C})$ for the object $G_{2k,2l}(V,C,E)$, where $V$, $C$, and $E$ denote the sets of VNs, CNs, and edges, respectively, and $(k,l) \in \{(3,3), (3,4), (4,4)\}$, is called a dominant object pattern class if it satisfies that $|E| = |\mathcal{E}(\mathcal{V},\mathcal{C})|$. This condition means that none of the edges in $E$ are equivalent to each other under any element of $\mathcal{E}(\mathcal{V},\mathcal{C})$, or equivalently, each edge is mapped to a distinct entry of the matrix representation of this object pattern class.
Lastly, we denote the set of all dominant object pattern classes for the object $G_{2k,2l}$ by $\mathcal{D}(G_{2k,2l})$. Here, $G_{2k,2l}$ is used to refer to the graph or to the object interchangeably.
\end{definition}

Dominant object pattern classes are central to our analysis of these configurations and to the optimization problem that we construct. The reason behind this choice, as well as the justification for naming these pattern classes ``dominant,'' will be explained later in this section. It can be observed that the object pattern class in Fig.~\ref{fig:68} is also dominant. 

Note that for a given object, there could be multiple pattern classes with the same characteristic polynomial. One example to this is the case of cycle-$8$, where patterns $P_4$, $P_5$, $P_6$, and $P_7$ (shown in Fig.~\ref{fig: cycle8_patterns}) share the characteristic polynomial $f^4(X,Y)f^4(X^{-1},Y^{-1})$. Another example is the case of $6$-$6$, $6$-$8$, and $8$-$8$ configurations, where dominant object pattern classes share the same characteristic polynomial. Proof of this claim with the exact expression of characteristic polynomial is given in Lemma~\ref{lem:dom_char}.

\begin{lemma} \label{lem:dom_char}
For the $2k$-$2l$ configurations, $(k,l) \in \{(3,3), (3,4), (4,4)\}$, all dominant object pattern classes of this object share the following characteristic polynomial:
\begin{align} \label{eq:char_pol_dom}
h(\mathbf{X}, \mathbf{Y}; G_{2k,2l} \mid \mathcal{V}, \mathcal{C}) = f(X_1 X_2,Y_1 Y_2) f(X_1^{-1} X_2^{-1},Y_1^{-1} Y_2^{-1}) f^{k-1}(X_1,Y_1) f^{k-1}(X_1^{-1},Y_1^{-1})f^{l-1}(X_2,Y_2) f^{l-1}(X_2^{-1},Y_2^{-1}).
\end{align}
\end{lemma}

\begin{proof}
The object $G_{2k,2l}$ is constructed by concatenating a cycle-$2k$ (denoted by $s_1$) and a cycle-$2l$ (denoted by $s_2$), connecting them through a chain of VN-CN-VN. These cycles are also the fundamental cycles of the object. Due to this construction, the fundamental cycles share exactly two edges on the bipartite graph, while the remaining edges are divided between $2k-2$ edges that belong solely to the first cycle and $2l-2$ edges that belong solely to the second cycle.

By the definition of dominant object pattern classes, the characteristic polynomial simplifies to:
\begin{align}
h(\mathbf{X}, \mathbf{Y}; G_{2k,2l} \mid \mathcal{V}, \mathcal{C}) = \prod_{e \in E} 
f \left( \prod_{s \in S} X_s^{\delta_{e, s}}, \prod_{s \in S} Y_s^{\delta_{e, s}} \right),
\end{align}
since no two distinct edges are equivalent in a dominant object pattern class.

For the two shared edges between the fundamental cycles, we assign $\delta$ values such that for one edge, $\delta_{e,s_1} = \delta_{e,s_2} = 1$, and for the other edge, $\delta_{e,s_1} = \delta_{e,s_2} = -1$, as was done in Example~\ref{example1}. For the remaining $2k-2$ edges belonging to the first cycle, we assign $\delta_{e,s_2} = 0$, and set $\delta_{e,s_1} = 1$ for half of them (i.e., $k-1$ edges) and $\delta_{e,s_1} = -1$ for the other half. Similarly, for the $2l-2$ edges belonging to the second cycle, we assign $\delta_{e,s_1} = 0$, and set $\delta_{e,s_2} = 1$ for half of them (i.e., $l-1$ edges) and $\delta_{e,s_2} = -1$ for the other half. With these degree assignments, the resulting characteristic polynomial matches the expression given in \eqref{eq:char_pol_dom}.
\end{proof}

\begin{lemma}
We denote the probability that any dominant object pattern of the object $G_{2k,2l}$ remains active after partitioning and relocation by $P_{2k,2l}^{\mathcal{D}}(\mathbf{p}^{\textup{con}})$, for a given distribution matrix $\mathbf{P}$. From Lemma~\ref{lemma:charpol} and Lemma~\ref{lem:dom_char}, it follows that:
\begin{align}
&P_{2k,2l}^{\mathcal{D}}(\mathbf{p}^{\textup{con}})= \nonumber \\
&\sum_{M \mid b_1, M \mid b_2} \left[ f(X_1 X_2,Y_1 Y_2) f(X_1^{-1} X_2^{-1},Y_1^{-1} Y_2^{-1}) f^{k-1}(X_1,Y_1) f^{k-1}(X_1^{-1},Y_1^{-1}) f^{l-1}(X_2,Y_2) f^{l-1}(X_2^{-1},Y_2^{-1}) \right]_{0,0,b_1,b_2}.
\end{align}
Here, $[\cdot]_{a_1, a_2, b_1, b_2}$ denotes the coefficient of the polynomial corresponding to the monomial in which the degrees of the variables $X_1$, $X_2$, $Y_1$, and $Y_2$ are given by $a_1$, $a_2$, $b_1$, and $b_2 \in \mathbb{Z}$, respectively.

Additionally, the gradient of $P_{2k,2l}^{\mathcal{D}}(\mathbf{p}^{\textup{con}})$ with respect to $\mathbf{p}^{\textup{con}}$ is expressed as:
\vspace{-0.1em}\begin{align} \label{eq:dom_grad}
&\nabla_{\mathbf{p}^{\textup{con}}} P_{2k,2l}^{\mathcal{D}}(\mathbf{p}^{\textup{con}})=  \nonumber \\
&\sum_{M \mid b_1, M \mid b_2} \hspace{-1.0em} \left[ \nabla_{\mathbf{p}^{\textup{con}}} \big ( f(X_1 X_2,Y_1 Y_2) f(X_1^{-1} X_2^{-1},Y_1^{-1} Y_2^{-1}) f^{k-1}(X_1,Y_1) f^{k-1}(X_1^{-1},Y_1^{-1}) f^{l-1}(X_2,Y_2) f^{l-1}(X_2^{-1},Y_2^{-1}) \big ) \right]_{0,0,b_1,b_2} \nonumber \\
&= 2\sum_{M \mid b_1, M \mid b_2} \Big ( [f(X_1 X_2,Y_1 Y_2) f^{k-1}(X_1,Y_1) f^{k-1}(X_1^{-1},Y_1^{-1}) f^{l-1}(X_2,Y_2) f^{l-1}(X_2^{-1},Y_2^{-1}) \mathbf{v}_{1,1}^{\textup{con}}]_{0,0,b_1,b_2} \nonumber \\
&+ [(k-1) f(X_1 X_2,Y_1 Y_2) f(X_1^{-1} X_2^{-1},Y_1^{-1} Y_2^{-1}) f^{k-1}(X_1,Y_1) f^{k-2}(X_1^{-1},Y_1^{-1}) f^{l-1}(X_2,Y_2) f^{l-1}(X_2^{-1},Y_2^{-1}) \mathbf{v}_{1,0}^{\textup{con}}]_{0,0,b_1,b_2} \nonumber \\
&+ [(l-1) f(X_1 X_2,Y_1 Y_2) f(X_1^{-1} X_2^{-1},Y_1^{-1} Y_2^{-1}) f^{k-1}(X_1,Y_1) f^{k-1}(X_1^{-1},Y_1^{-1}) f^{l-1}(X_2,Y_2) f^{l-2}(X_2^{-1},Y_2^{-1}) \mathbf{v}_{0,1}^{\textup{con}}]_{0,0,b_1,b_2} \Big ).
\end{align}
Here, $\mathbf{v}_{i,j}^{\textup{con}}$ is the vector of length $(m+1)M$ obtained by concatenating the rows of the following $(m+1) \times M$ matrix (of monomials)
\begin{align}
\mathbf{V}_{i,j} = \begin{bmatrix}
1 & Y_1^{-i} Y_2^{-j} & \dots & Y_1^{-(M-1)i} Y_2^{-(M-1)j} \\
X_1^{-i} X_2^{-j} & X_1^{-i} X_2^{-j} Y_1^{-i} Y_2^{-j} & \dots & X_1^{-i} X_2^{-j} Y_1^{-(M-1)i} Y_2^{-(M-1)j} \\
\vdots & \vdots & \ddots & \vdots \\
X_1^{-mi} X_2^{-mj} & X_1^{-mi} X_2^{-mj} Y_1^{-i} Y_2^{-j} & \dots & X_1^{-mi} X_2^{-mj} Y_1^{-(M-1)i} Y_2^{-(M-1)j}
\end{bmatrix}.
\end{align}
\end{lemma}

\begin{proof}
The lemma is self-explanatory given all the previous results and discussions.
\end{proof}

We now proceed to analyze the expected number of active $6$-$6$, $6$-$8$, and $8$-$8$ configurations after partitioning and relocations. We first need to determine the cardinality of each possible object pattern class. This is accomplished by examining the matrix representations of all of the relevant object pattern classes. 

Since these configurations consist of two cycles connected through a VN-CN-VN chain, in the matrix representation of all their pattern classes, the cycle candidates of the fundamental cycles (i.e., the traversal paths of the fundamental cycles over matrix entries) share two entries located in the same row (corresponding to the shared CN) but in different columns (corresponding to the adjacent VNs). Furthermore, in both cycle candidates, these entries are directly connected in the matrix. An example of this can be seen in Fig.~\ref{fig:68}, where the entries labeled with $(c_1,v_1)$ and $(c_1,v_2)$ are shared between the two cycle candidates marked in red and blue. 

Another condition is that the two cycle candidates must not share any additional entry directly connected to these shared entries. Otherwise, the cycles would at least share an additional CN at the connection point, forming a CN-VN-CN-VN chain, which is not the configuration we target. We exclude such cases from our analysis, as their population is relatively low and they do not frequently appear among the most detrimental absorbing sets.

We group each object pattern class according to the triplet $( |\mathcal{E}|,|\mathcal{V}|, |\mathcal{C}|)$, which represents the number of spanned entries, columns, and rows in the matrix representation, respectively. To count these configurations, we developed a software algorithm for each $2k$-$2l$ configuration, $(k,l) \in \{(3,3), (3,4), (4,4)\}$. This algorithm lists all cycles of length $2k$ and length $2l$ in an all-one submatrix of size $|\mathcal{C}| \times |\mathcal{V}|$ for each possible $(|\mathcal{C}|, |\mathcal{V}|)$ pair. It then iteratively checks for each distinct cycle pair whether it satisfies the conditions described above, and whether all the columns and rows of this all-one submatrix are spanned by this pair of cycles. If a pair meets the required conditions, then the corresponding configuration on the matrix constitutes a valid matrix representation of the desired configuration. The algorithm counts and groups these instances according to their number of rows, columns, and total entries (edges).

It can be observed that these counts correspond to the number of object patterns associated with the $( |\mathcal{E}|,|\mathcal{V}|, |\mathcal{C}|)$ triplets generated from a submatrix of size $|\mathcal{C}| \times |\mathcal{V}|$. Consequently, the actual number of these object patterns in the base matrix is obtained by multiplying these counts by the combinatorial factor $\binom{\kappa}{|\mathcal{V}|} \binom{\gamma}{|\mathcal{C}|}$.

The detailed counts are presented in Table~\ref{table: base_matrix_counts}, and a sample matrix representation for each triplet is illustrated in Fig.~\ref{fig:base matrix patterns}. It should be noted that we list only the triplets corresponding to the top three highest $|\mathcal{E}|$ values (the highest shows last) for each configuration, as the counts for the remaining instances are significantly lower in comparison. Additionally, we only include cases with $|\mathcal{C}| \leq 4$, since in our codes we assume $\gamma \leq 4$, and $|\mathcal{C}|$ cannot exceed $\gamma$. Having said that, the same procedure can be adopted to enumerate such configurations for $\gamma \geq 5$.

\begin{table}[h!]
\caption{Number of Object Patterns of The $6$-$6$, $6$-$8$, and $8$-$8$ Configurations Grouped by Number of Rows, Columns and Spanned Entries of The Corresponding Matrix Representations}
\centering
\scalebox{1.00}
{
\renewcommand{\arraystretch}{1.2}
\begin{tabular}{|c|c|c|c|c|}
\hline
\makecell{Configuration type} & \makecell{$\left|\mathcal{E}\right|$} & \makecell{$\left|\mathcal{V}\right|$} & \makecell{$\left|\mathcal{C}\right|$} & \makecell{Number of instances \\ of object pattern classes} \\
\hline
\multirow{4}{*}{$6$-$6$} & $8$ & $3$ & $3$ & $9 \, \binom{\kappa}{3} \binom{\gamma}{3}$\\
\cline{2-5}
& $9$ & $3$ & $4$ & $72 \, \binom{\kappa}{3} \binom{\gamma}{4}$ \\
\cline{2-5}
& \multirow{2}{*}{$10$} & $4$ & $3$ & $36 \, \binom{\kappa}{4} \binom{\gamma}{3}$ \\
\cline{3-5}
& & $4$ & $4$ & $288 \binom{\kappa}{4} \, \binom{\gamma}{4}$ \\
\thickhline
\multirow{6}{*}{$6$-$8$} & $10$ & $4$ & $4$ & $576 \, \binom{\kappa}{4} \binom{\gamma}{4}$ \\
\cline{2-5}
& \multirow{2}{*}{$11$} & $4$ & $3$ & $144 \, \binom{\kappa}{4} \binom{\gamma}{3}$ \\
\cline{3-5}
& & $4$ & $4$ & $2{,}880 \, \binom{\kappa}{4} \binom{\gamma}{4}$ \\
\cline{2-5}
& \multirow{3}{*}{$12$} & $4$ & $4$ & $1{,}152 \, \binom{\kappa}{4} \binom{\gamma}{4}$ \\
\cline{3-5}
& & $5$ & $3$ & $360 \, \binom{\kappa}{5} \binom{\gamma}{3}$ \\
\cline{3-5}
& & $5$ & $4$ & $11{,}520 \, \binom{\kappa}{5} \binom{\gamma}{4}$ \\
\thickhline
\multirow{11}{*}{$8$-$8$} & \multirow{4}{*}{$12$} & $4$ & $3$ &  $198 \, \binom{\kappa}{4} \binom{\gamma}{3}$ \\
\cline{3-5}
& & $4$ & $4$ & $2{,}952 \, \binom{\kappa}{4} \binom{\gamma}{4}$ \\
\cline{3-5}
& & $5$ & $3$ & $1{,}080 \, \binom{\kappa}{5} \binom{\gamma}{3}$ \\
\cline{3-5}
& & $5$ & $4$ & $10{,}080 \, \binom{\kappa}{5} \binom{\gamma}{4}$ \\
\cline{2-5}
& \multirow{3}{*}{$13$} & $4$ & $4$ & $1{,}728 \, \binom{\kappa}{4} \binom{\gamma}{4}$ \\
\cline{3-5}
& & $5$ & $3$ & $2{,}520 \, \binom{\kappa}{5} \binom{\gamma}{3}$ \\
\cline{3-5}
& & $5$ & $4$ & $53{,}280 \, \binom{\kappa}{5} \binom{\gamma}{4}$ \\
\cline{2-5}
& \multirow{4}{*}{$14$} & $4$ & $4$ & $864 \, \binom{\kappa}{4} \binom{\gamma}{4}$ \\
\cline{3-5}
& & $5$ & $4$ & $17{,}280 \, \binom{\kappa}{5} \binom{\gamma}{4}$ \\
\cline{3-5}
& & $6$ & $3$ & $5{,}400 \, \binom{\kappa}{6} \binom{\gamma}{3}$ \\
\cline{3-5}
& & $6$ & $4$ & $120{,}960 \, \binom{\kappa}{6} \binom{\gamma}{4}$ \\
\hline
\end{tabular}}
\label{table: base_matrix_counts}
\vspace{-0.5em}
\end{table}

\begin{figure*}
\centering
\includegraphics[width=0.9\textwidth]{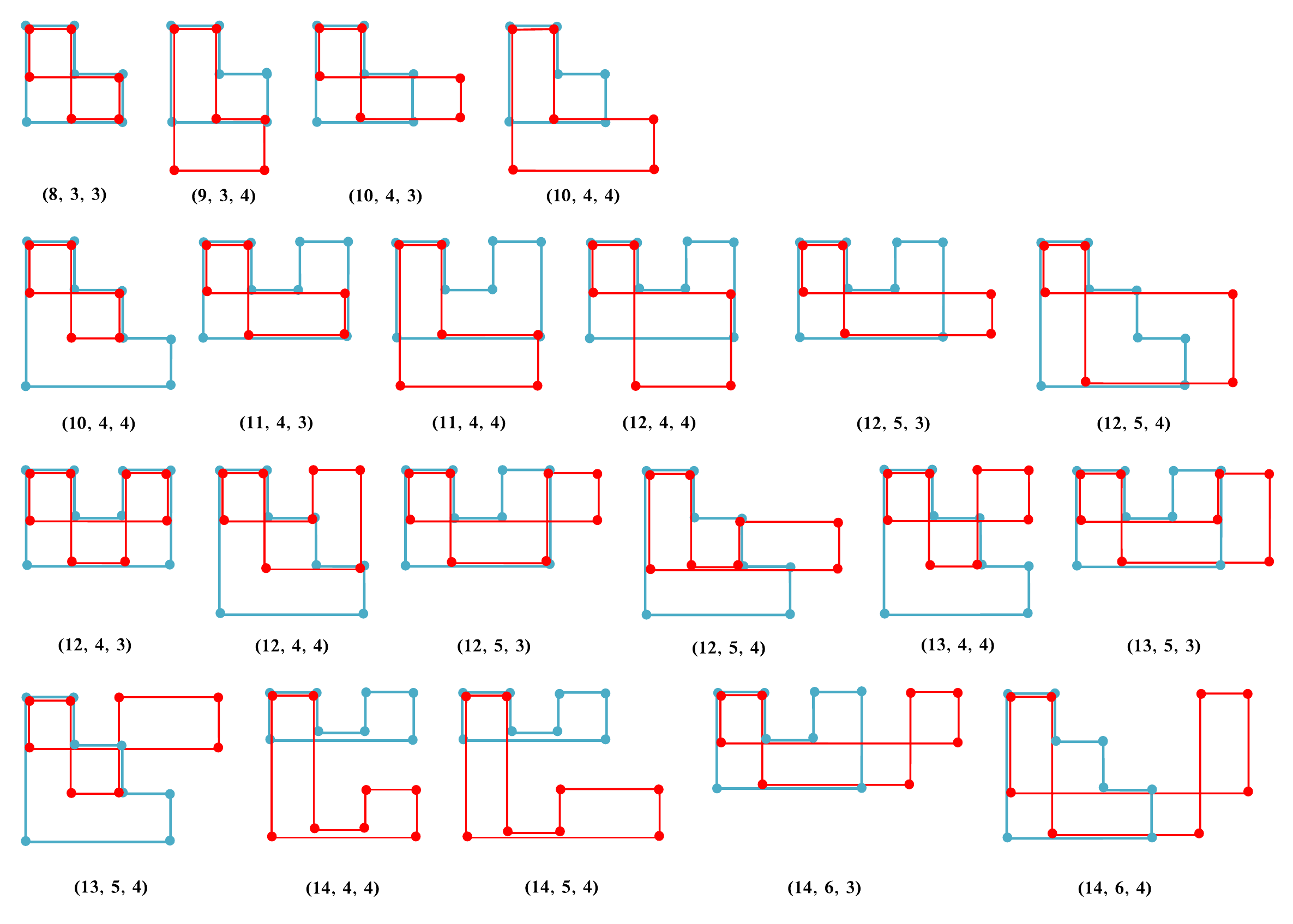} \vspace{-0.8em}
\caption{Example object pattern class matrix representations of the $6$-$6$, $6$-$8$, and $8$-$8$ configurations grouped by the values in $( |\mathcal{E}|,|\mathcal{V}|, |\mathcal{C}|)$.} \vspace{-1.0em}
\label{fig:base matrix patterns}
\end{figure*}

Since for dominant object patterns we have $|E| = |\mathcal{E}|$, it can be observed from Table~\ref{table: base_matrix_counts} that all cases with the maximum $|\mathcal{E}|$ correspond to dominant object pattern cases. Hence, the table justifies the use of the term "dominant," as the number of these patterns is considerably higher compared with other object patterns, especially for $\gamma = 4$.

Accordingly, in our algorithms, we simplify the characteristic polynomial computations for the $6$-$6$, $6$-$8$, and $8$-$8$ configurations by considering only the dominant object patterns. This simplification significantly reduces the computational complexity of our algorithms without causing any considerable deviation from the exact results; further details will be discussed in Section~\ref{sec: algorithms}.

Based on this discussion, we introduce the new notation $\tilde{h}(\mathbf{X},\mathbf{Y};G_{2k,2l})$ for the approximated characteristic polynomial, which is expressed as follows:
\begin{align}
&\tilde{h}(\mathbf{X}, \mathbf{Y}; G_{2k,2l}) = \sum_{\mathcal{P}(\mathcal{V}, \mathcal{C}) \in \mathcal{D}(G_{2k,2l})} |\mathcal{P}(\mathcal{V}, \mathcal{C})| \, h(\mathbf{X}, \mathbf{Y}; G_{2k,2l} \mid \mathcal{V}, \mathcal{C}) \nonumber  \\ &= \Lambda_{2k,2l} \cdot f(X_1 X_2, Y_1 Y_2) f(X_1^{-1} X_2^{-1}, Y_1^{-1} Y_2^{-1}) f^{k-1}(X_1, Y_1) f^{k-1}(X_1^{-1}, Y_1^{-1}) f^{l-1}(X_2, Y_2) f^{l-1}(X_2^{-1}, Y_2^{-1}).
\end{align}
Here, $\Lambda_{2k,2l} = \sum_{\mathcal{P}(\mathcal{V}, \mathcal{C}) \in \mathcal{D}(G_{2k,2l})} |\mathcal{P}(\mathcal{V}, \mathcal{C})|$, which can be obtained from the counts provided in Table~\ref{table: base_matrix_counts} for each configuration. 
We define the expected number of active $6$-$6$, $6$-$8$, and $8$-$8$ configurations with dominant object patterns as follows:
\begin{align}
&N_{2k,2l}^{\mathcal{D}}(\mathbf{p}^{\textup{con}}) = \nonumber \\
&\Lambda_{2k,2l} \cdot \hspace{-0.6em} \sum_{M \mid b_1, M \mid b_2} \hspace{-0.5em} \left[ f(X_1 X_2,Y_1 Y_2) f(X_1^{-1} X_2^{-1},Y_1^{-1} Y_2^{-1}) f^{k-1}(X_1,Y_1) f^{k-1}(X_1^{-1},Y_1^{-1}) f^{l-1}(X_2,Y_2) f^{l-1}(X_2^{-1},Y_2^{-1}) \right]_{0,0,b_1,b_2},
\end{align}
for $(k,l)\in\lbrace (3,3), (3,4), (4,4)\rbrace$, where the coefficients $\Lambda_{2k,2l}$ are given by:
\begin{align*}
\Lambda_{6,6} &= 
\begin{cases}
36\binom{\kappa}{4}\binom{\gamma}{3}, & \text{for } \gamma = 3, \\
36\binom{\kappa}{4}\binom{\gamma}{3} + 288\binom{\kappa}{4}\binom{\gamma}{4}, & \text{for } \gamma = 4,
\end{cases} \\
\Lambda_{6,8} &= 
\begin{cases}
360\binom{\kappa}{5}\binom{\gamma}{3}, & \text{for } \gamma = 3, \\
360\binom{\kappa}{5}\binom{\gamma}{3} + 1{,}152\binom{\kappa}{4}\binom{\gamma}{4} + 1{,}520\binom{\kappa}{5}\binom{\gamma}{4}, & \text{for } \gamma = 4,
\end{cases} \\
\Lambda_{8,8} &= 
\begin{cases}
5{,}400\binom{\kappa}{6}\binom{\gamma}{3}, & \text{for } \gamma = 3, \\
5{,}400\binom{\kappa}{6}\binom{\gamma}{3} + 864\binom{\kappa}{4}\binom{\gamma}{4} + 17{,}280\binom{\kappa}{5}\binom{\gamma}{4} + 120{,}960\binom{\kappa}{6}\binom{\gamma}{4}, & \text{for } \gamma = 4.
\end{cases}
\end{align*}
These coefficients are all obtained from Table~\ref{table: base_matrix_counts}.

Finally, we define a joint expression for the expected number of $6$-$6$, $6$-$8$, and $8$-$8$ configurations with dominant object patterns by considering a weighted sum of the individual expectations as follows:
\begin{align} \label{eqn: expected number of cycle-con}
N^{\mathcal{D}}(\mathbf{p}^{\textup{con}}) = w_{6,6} N_{6,6}^{\mathcal{D}}(\mathbf{p}^{\textup{con}}) + w_{6,8} N_{6,8}^{\mathcal{D}}(\mathbf{p}^{\textup{con}}) + w_{8,8} N_{8,8}^{\mathcal{D}}(\mathbf{p}^{\textup{con}}),
\end{align}
where $(w_{6,6}, w_{6,8}, w_{8,8})$ is a triplet denoting the weight values assigned to each corresponding configuration. Gradient of this expression can be evaluated using \eqref{eq:dom_grad}. 

%%%%%%%%%%%%%%%%%%%%%%%%%%%%%%%%%%%%%%%%%%%%%%%%%%%%%%%%%%%%%%
\section{Design and Optimization Algorithms} \label{sec: algorithms} 

In this section, we present our optimization algorithm based on gradient descent, namely MD-GRADE, which is designed to search for a locally-optimal probability-distribution matrix for the optimization problem introduced in the previous sections. We also introduce our FL optimizer, which aims to find an optimal or sub-optimal relocation arrangement for constructing the desired GD-MD code. For the FL optimization, we adapt an MCMC algorithm \cite{reins_mcmc}, which we recently introduced in the literature for other stages of code design such as partitioning and lifting. The MCMC algorithm is initialized with a relocation matrix generated randomly according to the probability-distribution matrix obtained by the MD-GRADE algorithm, which guides the FL optimizer to reach a better result.
   
\subsection{Multi-Dimensional Gradient-Descent Distributor}\label{subsec: MD-GRADE algo}

\begin{definition} (Coefficient array). For every multi-variable polynomial $g(x_1, x_2, \dots, x_n)$ with $n$ variables, we define an $n$-dimensional array $\mathsf{G}$ as the coefficient array such that:
\begin{align}
g(x_1, x_2, \dots, x_n) = \sum_{i_1, i_2, \dots, i_n} \mathsf{G}\left[i_1, i_2, \dots, i_n\right] x_1^{i_1} x_2^{i_2} \dots x_n^{i_n} .
\end{align} 
Each polynomial can be represented using its coefficient array. For every variable $x_k$, where $k \in \lbrace 1, 2, \dots, n\rbrace $, the index $i_k$ spans from the smallest to the largest exponent of $x_k$ with a non-zero coefficient (coefficients associated with exponents in between are allowed to be zeros). We refer to the cardinality of the set of integers that $i_k$ spans as the range of $\mathsf{G}$ along the $k$-th dimension, and the multiplication of range values along all different dimensions as the size of $\mathsf{G}$.
\end{definition}

The coefficient array of an $n$-variable polynomial resulting from the multiplication of two $n$-variable polynomials can be computed using the $n$-dimensional discrete convolution of their respective coefficient arrays. Similarly, the coefficient array of an $n$-variable polynomial obtained by multiplying a scalar with another $n$-variable polynomial can be computed by multiplying the scalar with each entry.

Here, we represent $n$-dimensional discrete convolution between two $n$-dimensional arrays with $\mathsf{G} = \mathsf{G}_1 * \mathsf{G}_2$, where these coefficient arrays satisfy the given relation:
\begin{align} \label{eq: discrete conv}
\mathsf{G}\left[i_1, i_2, \dots, i_n\right] = \sum_{j_1, j_2, \dots, j_n} \mathsf{G}_1\left[j_1, j_2, \dots, j_n\right] \cdot \mathsf{G}_2\left[i_1-j_1, i_2-j_2, \dots, i_n-j_n\right],
\end{align}
for every possible $(i_1, i_2, \dots, i_n)$.

In our objective functions and their gradients, we frequently use the function $f(X, Y; \mathbf{P})$ for cycles, and $f(X_1 X_2, Y_1 Y_2; \mathbf{P})$ for cycle concatenations, often involving different powers of the $(X, Y)$ pairs.  To simplify the notation, we introduce the following convention: $\mathsf{F}_i$ denotes the coefficient array of the polynomial $f(X^i, Y^i; \mathbf{P})$, and $\mathsf{F}_{i,j}$ denotes the coefficient array of $f(X_1^i X_2^j, Y_1^i Y_2^j; \mathbf{P})$. The first dimension of $\mathsf{F}_i$ corresponds to $X$, and the second corresponds to $Y$. Similarly, the first two dimensions of $\mathsf{F}_{i,j}$ correspond to $X_1$ and $X_2$, while the last two correspond to $Y_1$ and $Y_2$, respectively. That is,
\[
\mathsf{F}_i[i_1, i_2] = \left[f(X^i, Y^i; \mathbf{P})\right]_{i_1, i_2},
\]
and
\[
\mathsf{F}_{i,j}[i_1, i_2, i_3, i_4] = \left[f(X_1^i X_2^j, Y_1^i, Y_2^j; \mathbf{P})\right]_{i_1, i_2, i_3, i_4}.
\]

These coefficient arrays are computed based on the probability matrix $\mathbf{P}$ and powers of $(X, Y)$ pairs using \eqref{eq: f} in our algorithm. It can also be seen that range of $\mathsf{F}_i$ along the first dimension is $m|i|+1$ and it is $(M-1)|i|+1$ along the second dimension. This can be generalized for $\mathsf{F}_{i,j}$ as well. From this point onward, we will use this notation in our algorithm.

Now, we introduce our self-defined functions in the body of the MD-GRADE algorithm:
\begin{enumerate}
\item $\mathrm{conv}(\cdot)$ is used to compute the discrete convolution of a list of coefficient arrays with the same dimension $n$. If the same input array appears multiple times in the input list, we denote this using superscript notation to indicate the number of repetitions. Specifically, if the coefficient array $\mathsf{G}$ appears $k$ times in the input list, we denote it by $\mathsf{G}^k$.

We adopt two approaches to compute the discrete convolution. The first is the direct approach; when applied to two input arrays, as in $\mathrm{conv}(\mathsf{G}_1,\mathsf{G}_2)$, the resulting array is computed using $2n$ nested loops based on~\eqref{eq: discrete conv}. When the length of the input list is higher, say some $\ell>2$, then the output is evaluated using the recursive relation:
\[
\mathrm{conv}(\mathsf{G}_1,\mathsf{G}_2,\mathsf{G}_3, \dots, \mathsf{G}_{\ell}) = \mathrm{conv}\left(\mathsf{G}_1,\mathrm{conv}(\mathsf{G}_2, \mathsf{G}_3, \dots, \mathsf{G}_{\ell})\right).
\]
The second approach exploits the well-known property that convolution in the time domain is equivalent to element-wise multiplication in the frequency domain. Specifically, by applying the $n$-dimensional discrete Fourier transform (DFT) to the input coefficient arrays, performing element-wise multiplication in the frequency domain, and then applying the inverse DFT (IDFT), we can obtain the convolution result. Here, both DFT and IDFT should be applied with the ranges of the resulting array on each dimension.

The rationale behind preferring the Fourier-based approach lies in computational efficiency. Convolution is the most computationally demanding component of our MD-GRADE algorithm. In the direct case, the number of arithmetic operations is $O(N_1 N_2)$, where $N_1$ and $N_2$ are the sizes of the two input arrays. In contrast, the Fourier-based method can be implemented using the fast Fourier transform (FFT) algorithm, reducing the complexity to $O(N \log N)$, where $N$ is the size of the resulting array. This logarithmic speed-up is particularly beneficial in the case of cycle concatenation, where higher-dimensional coefficient arrays are used, and the array sizes increase with the number of dimensions. The array size also grows with the $m$ and $M$. Hence, this approach allows us to efficiently evaluate probability-distribution matrices with significantly larger $m$ and $M$ values.

Nevertheless, we continue to use the direct approach for small-sized input arrays, as the overhead of FFT may outweigh its benefits in such cases, and use the FFT-based method when input sizes are large to optimize the overall efficiency.

\item $\mathrm{force(\cdot)}$ is a function specifically defined for $\mathbf{P}$. It takes the inputs $\mathbf{P}$ and $\mathbf{p}^*$. It projects each row of $\mathbf{P}$ indexed by $t$ onto the hyperplane defined by the orthogonal vector $\mathbf{1}_{1\times M}$ and adds the vector $\mathbf{p}^{*}[t]/M \cdot \mathbf{1}_{1\times M}$ via
$$\mathbf{P}[t,k] \gets \mathbf{P}[t,k] + \bigg[ \mathbf{p}^{*}[t]- \sum_{\ell=0}^{M-1} \mathbf{P}[t,\ell] \bigg]/M.$$
Summations are for non-updated values. Then, it corrects its rows by making the negative entries $0$ and by updating the positive entries in order not to change the sum of each row via
\vspace{-0.5em}$$ \mathbf{P}[t,k] \gets \mathbf{P}'[t,k] \cdot \frac{\mathbf{p}^{*}[t]}{\sum_{\ell=0}^{M-1} \mathbf{P}'[t,\ell]},$$
where $\mathbf{P}'[t,k] = \begin{cases} 0 & \textup{ if } \mathbf{P}[t,k] < 0, \\ \mathbf{P}[t,k] & \textup{ otherwise.} \end{cases}$

The purpose of this function is to map any matrix $\mathbf{P}$ into the domain and feasible set of the optimization problem. The first update projects each row of the matrix $\mathbf{P}$ onto the hyperplane where the row-wise sums are equal to the corresponding entry in  $\mathbf{p}^{*}$. The second update assigns zero to any negative entry in $\mathbf{P}$ and then rescales each row to preserve its original sum. In brief, this is our way to satisfy the optimization constraints.
\end{enumerate}

Here, $\mathbf{p}^{*}$ is a locally-optimal edge distribution of an SC code with parameters $(\gamma, \kappa, m)$, obtained by \cite[Algorithm~1]{GRADE}. It should also be noted that if an underlying SC code and its partitioning matrix are already available, the edge distribution vector derived from this partitioning matrix can be used as $\mathbf{p}^{*}$, where here $\mathbf{p}^{*}[i] = \frac{1}{\gamma\kappa} \sum_{j,k} \mathbbm{1}\{\mathbf{K}(j,k) = i\}, \, \forall i \in \{0,1,\dots,m\}$. This approach enables the code designer to generate the relocation matrix in a way that is more specifically tailored to the already-designed underlying SC code.

We now introduce the MD-GRADE algorithm for the simplest case: minimizing the number of cycles-$6$. The algorithm begins by initializing the matrix $\mathbf{P}$ such that its first column is equal to the vector $\mathbf{p}^*$ and all other entries are set to zero. This corresponds to the case of $0\%$ relocation.

In each iteration, the algorithm evaluates the coefficient matrices of the polynomials $f(X, Y)$ and $f(X^{-1}, Y^{-1})$ based on the current matrix $\mathbf{P}$. It then computes the coefficient arrays of the polynomials that appear in the objective function and its gradient expressions in \eqref{eqn: expected number of cycle-6} and \eqref{eq_6} using the convolution operation described earlier.

Next, the algorithm sums up the necessary entries of these arrays to evaluate both the objective function and its gradient. To ensure numerical stability, the gradient is normalized using its $L_2$ norm. The matrix $\mathbf{P}$ is then updated in the opposite direction of the gradient. Since the gradient update may push $\mathbf{P}$ outside the domain or the feasible set, $\mathrm{force}(\cdot)$ is applied to map it back into these sets. The algorithm checks termination criteria at each iteration: whether the MD density exceeds a predefined maximum, or whether the change in the objective function between consecutive iterations falls below a certain threshold. If either condition is met, the algorithm terminates; otherwise, it proceeds to the next iteration. Finally, the algorithm returns the locally-optimal probability-distribution matrix and the corresponding expected number of cycles-$6$ in the resulting Tanner graph of $\mathbf{H}_{\textup{MD}}$ obtained by Theorem \ref{thm: forecast}.

\begin{algorithm}
\caption{Multi-Dimensional Gradient-Descent Distributor (MD-GRADE) for Cycle-$6$ Optimization} \label{algo: MD-GRADE6}
\begin{algorithmic}[1]
\Statex \textbf{Inputs:} $\gamma,\kappa, L, m, M$: parameters of the MD code; $\mathbf{p}^{*}$: locally-optimal SC edge distribution; $\mathcal{T}_{\textup{max}}$: maximum MD density; $\epsilon, \alpha$: termination threshold and step size of gradient descent.
\Statex \textbf{Outputs:} $\mathbf{P}$: a (joint) locally-optimal probability-distribution matrix;  $E_6$: the expected number of cycles-$6$.
\Statex \textbf{Intermediate Variables:} $F_{\textup{prev}}, F_{\textup{cur}}$: the values of the objective function in \eqref{eqn: expected number of cycle-6} at the previous and current iterations; $\mathsf{F}_1, \mathsf{F}_{-1}$ coefficient arrays for the polynomials $f(X, Y), f(X^{-1}, Y^{-1})$; $\mathsf{G}_{\textup{obj}}, \mathsf{G}_{\textup{grad}}$ coefficient arrays of the polynomials appearing in the objective function and its gradient expressions;  $\mathbf{G}$: the gradient matrix (of size $(m+1) \times M$) of the objective function, where $\mathbf{g}^{\textup{con}}=\nabla_{\mathbf{p}^{\textup{con}}} N_6(\mathbf{p}^{\textup{con}})$; $\mathcal{T}_{\textup{cur}}$: current MD density.
\State $\mathbf{P}[i, 0] \gets \mathbf{p}^*[i]$, $\mathbf{P}[i,j] \gets 0$,  $\forall i\in \{0,1,\dots,m\}$ and $\forall j\in \{1,\dots,M-1\}$. 
\State $F_{\textup{prev}}=0$.
\State Evaluate $\mathsf{F}_1$ and $\mathsf{F}_{-1}$ based on $\mathbf{P}$.
\State $\mathsf{G}_{\textup{obj}} \gets  6\binom{\gamma}{3}\binom{\kappa}{3} \cdot \mathrm{conv}(\mathsf{F}_1^3,\mathsf{F}_{-1}^3)$.
\State $\mathsf{G}_{\textup{grad}} \gets  36\binom{\gamma}{3}\binom{\kappa}{3} \cdot \mathrm{conv}(\mathsf{F}_1^3,\mathsf{F}_{-1}^2)$.
\State $F_{\textup{cur}} \gets \sum_{M \vert b} \mathsf{G}_{\textup{obj}}\left[0,b\right]$. 
\For{$\forall i\in \{0,1,\dots,m\}$ and $\forall j\in \{0,\dots,M-1\}$}
\State $\mathbf{G}[i,j] \gets \sum_{M \vert b} \mathsf{G}_{\textup{grad}}[i, b+j]$.
\EndFor
\State $\mathbf{P} \gets \mathbf{P} - \alpha \displaystyle \frac{\mathbf{G}}{\| \mathbf{g}^{\textup{con}}  \|_2}$.
\State $\mathbf{P} \gets \mathrm{force}(\mathbf{P}, \mathbf{p}^*)$.
\State $\mathcal{T}_{\textup{cur}} \gets 1 - \sum_{i=0}^m \mathbf{P}[i,0]$.
\If{$(\mathcal{T}_{\textup{cur}} < \mathcal{T}_{\textup{max}})$ and $(|F_{\textup{cur}}-F_{\textup{prev}}| > \epsilon)$} 
\State $F_{\textup{prev}} \gets F_{\textup{cur}}$.
\State \textbf{go to} Step 3.
\EndIf
\State $E_6 \gets \frac{(2L-m)}{2} \cdot M \cdot F_{\textup{cur}}$.
\State \textbf{return} $\mathbf{P}$ and $E_6$.
\end{algorithmic}
\end{algorithm}

For the case of cycle-$8$, the overall structure of the algorithm remains mostly the same as in the cycle-$6$ case. The only difference lies in the computation of the objective function and the gradient. Specifically, we evaluate four distinct coefficient arrays for the objective function and seven distinct coefficient arrays for the gradient. Each of these arrays corresponds to a polynomial term in \eqref{eqn: expected number of cycle-8} and \eqref{eq_8}, in the same order as they appear in the equations.

\begin{algorithm}
\caption{Multi-Dimensional Gradient-Descent Distributor (MD-GRADE) for Cycle-$8$ Optimization} \label{algo: MD-GRADE8}
\begin{algorithmic}[1]
\Statex \textbf{Inputs:} $\gamma,\kappa, L, m, M$: parameters of the MD code; $\mathbf{p}^{*}$: locally-optimal SC edge distribution; $\mathcal{T}_{\textup{max}}$: maximum MD density; $w_1, w_2, w_3, w_4$: coefficients appearing in the objective function expression; $\epsilon, \alpha$: termination threshold and step size of gradient descent.
\Statex \textbf{Outputs:} $\mathbf{P}$: a (joint) locally-optimal probability-distribution matrix;  $E_8$: the expected number of cycles-$8$.
\Statex \textbf{Intermediate Variables:} $F_{\textup{prev}}, F_{\textup{cur}}$: the values of the objective function in \eqref{eqn: expected number of cycle-8} at the previous and current iterations; $\mathsf{F}_1, \mathsf{F}_{-1}, \mathsf{F}_2, \mathsf{F}_{-2}$: coefficient arrays for the polynomials $f(X, Y), f(X^{-1}, Y^{-1}), f(X^2, Y^2), f(X^{-2}, Y^{-2})$; $\mathsf{G}_{\textup{obj}}^{\{i\}}, \mathsf{G}_{\textup{grad}}^{\{j\}}$ for $i\in\{1,2,3,4\}, j\in\{1,2,\dots,7\}$: coefficient arrays of each polynomial term appearing in the objective function and its gradient expressions; $\mathbf{G}$: the gradient matrix (of size $(m+1) \times M$) of the objective function, where $\mathbf{g}^{\textup{con}}=\nabla_{\mathbf{p}^{\textup{con}}} N_8(\mathbf{p}^{\textup{con}})$; $\mathcal{T}_{\textup{cur}}$: current MD density.
\State $\mathbf{P}[i,0] \gets \mathbf{p}^*[i]$, $\mathbf{P}[i,j] \gets 0$,  $\forall i\in \{0,1,\dots,m\}$ and $\forall j\in \{1,\dots,M-1\}$. 
\State $F_{\textup{prev}}=0$.
\State Evaluate $\mathsf{F}_1, \mathsf{F}_{-1}, \mathsf{F}_2, \mathsf{F}_{-2}$ based on $\mathbf{P}$.
\State $\mathsf{G}_{\textup{obj}}^{\{1\}} \gets  w_1 \cdot \mathrm{conv}(\mathsf{F}_2^2, \mathsf{F}_{-2}^2)$, $\mathsf{G}_{\textup{obj}}^{\{2\}} \gets  w_2 \cdot \mathrm{conv}(\mathsf{F}_2, \mathsf{F}_{-2}, \mathsf{F}_1^2, \mathsf{F}_{-1}^2)$, $\mathsf{G}_{\textup{obj}}^{\{3\}} \gets  w_3 \cdot \mathrm{conv}(\mathsf{F}_2, \mathsf{F}_1^2, \mathsf{F}_{-1}^4)$, \Statex $\mathsf{G}_{\textup{obj}}^{\{4\}} \gets  w_4 \cdot \mathrm{conv}(\mathsf{F}_1^4, \mathsf{F}_{-1}^4)$.
\State $\mathsf{G}_{\textup{grad}}^{\{1\}} \gets  4 w_1 \cdot \mathrm{conv}(\mathsf{F}_2^2,\mathsf{F}_{-2})$, $\mathsf{G}_{\textup{grad}}^{\{2\}} \gets  2 w_2 \cdot \mathrm{conv}(\mathsf{F}_2,\mathsf{F}_1^2, \mathsf{F}_{-1}^2)$, $\mathsf{G}_{\textup{grad}}^{\{3\}} \gets  4 w_2 \cdot \mathrm{conv}(\mathsf{F}_2, \mathsf{F}_{-2}, \mathsf{F}_1^2, \mathsf{F}_{-1})$, \Statex $\mathsf{G}_{\textup{grad}}^{\{4\}} \gets  w_3 \cdot \mathrm{conv}(\mathsf{F}_1^2, \mathsf{F}_{-1}^4)$, $\mathsf{G}_{\textup{grad}}^{\{5\}} \gets  2 w_3 \cdot \mathrm{conv}(\mathsf{F}_2, \mathsf{F}_1, \mathsf{F}_{-1}^4)$, $\mathsf{G}_{\textup{grad}}^{\{6\}} \gets  4 w_3 \cdot \mathrm{conv}(\mathsf{F}_2, \mathsf{F}_1^2, \mathsf{F}_{-1}^3)$, $\mathsf{G}_{\textup{grad}}^{\{7\}} \gets  8 w_4 \cdot \mathrm{conv}(\mathsf{F}_1^4, \mathsf{F}_{-1}^3)$.
\State $F_{\textup{cur}} \gets \sum_{M \vert b} \sum_{i=1}^4 \mathsf{G}_{\textup{obj}}^{\{i\}}\left[0,b\right]$. 
\For{$\forall i\in \{0,1,\dots,m\}$ and $\forall j\in \{0,\dots,M-1\}$}
\State $\mathbf{G}[i,j] \gets \sum_{M \vert b} \mathsf{G}_{\textup{grad}}^{\{1\}}[2i, b+2j] + \mathsf{G}_{\textup{grad}}^{\{2\}}[2i, b+2j] + \mathsf{G}_{\textup{grad}}^{\{3\}}[i, b+j] + \mathsf{G}_{\textup{grad}}^{\{4\}}[-2i, b-2j] + \mathsf{G}_{\textup{grad}}^{\{5\}}[-i, b-j] + \mathsf{G}_{\textup{grad}}^{\{6\}}[i, b+j] + \mathsf{G}_{\textup{grad}}^{\{7\}}[i, b+j]$.
\EndFor
\State $\mathbf{P} \gets \mathbf{P} - \alpha \displaystyle \frac{\mathbf{G}}{\| \mathbf{g}^{\textup{con}}  \|_2}$.
\State $\mathbf{P} \gets \mathrm{force}(\mathbf{P}, \mathbf{p}^*)$.
\State $\mathcal{T}_{\textup{cur}} \gets 1 - \sum_{i=0}^m \mathbf{P}[i,0]$.
\If{$(\mathcal{T}_{\textup{cur}} < \mathcal{T}_{\textup{max}})$ and $(|F_{\textup{cur}}-F_{\textup{prev}}| > \epsilon)$} 
\State $F_{\textup{prev}} \gets F_{\textup{cur}}$.
\State \textbf{go to} Step 3.
\EndIf
\State $E_8 \gets (L-m) \cdot M \cdot F_{\textup{cur}}$.
\State \textbf{return} $\mathbf{P}$ and $E_8$.
\end{algorithmic}
\end{algorithm}

Lastly, for the case of cycle concatenation, the overall structure of the algorithm remains similar to the previous two cases, with the key difference lying in how the objective function and gradient are evaluated. Unlike the cycle-$6$ and cycle-$8$ cases, the cycle concatenation case involves minimizing a weighted sum of expected counts corresponding to multiple concatenated cycle configurations. Specifically, the objective function is formed by summing the expected counts of $6$-$6$, $6$-$8$, and $8$-$8$ concatenations, each multiplied by a corresponding weight. In our implementation, we typically set the weights for cycle concatenations as $(w_{6,6}, w_{6,8}, w_{8,8}) = (1, 10^{-2}, 10^{-4})$, which is natural given how detrimental each of them could be.

\begin{remark}
A key difference in the cycle concatenation case is that the polynomials appearing in the objective function involve four variables. This increased dimensionality is a limiting factor due to the associated increase in computational complexity. Although the convolution operations are optimized using an FFT-based method, they still remain computationally expensive.

To make the problem tractable, we focus only on the dominant object patterns for each cycle concatenation configuration, since all of them share the same characteristic polynomial, and their total count in the base matrix is significantly higher compared with the other object pattern classes. Including all possible configurations would dramatically increase the number of required convolution operations, making the evaluation of the probability matrix for large values of $m$ and $M$ computationally prohibitive. As a sanity check, we also applied the algorithm including all the configurations (without approximation) for the $6$-$6$ case and observed that the approximated solution closely matches the actual result in terms of both the final objective value and the probability-distribution matrix.
\end{remark}  

\begin{algorithm}
\caption{Multi-Dimensional Gradient-Descent Distributor (MD-GRADE) for Cycle-Concatenation Optimization} \label{algo: MD-GRADE-CON}
\begin{algorithmic}[1]
\Statex \textbf{Inputs:} $\gamma,\kappa, L, m, M$: parameters of the MD code; $\mathbf{p}^{*}$: locally-optimal SC edge distribution; $\mathcal{T}_{\textup{max}}$: maximum MD density; $w_{6,6}, w_{6,8}, w_{8,8}$: weights appearing in the objective function expression; $\Lambda_{6,6}, \Lambda_{6,8}, \Lambda_{8,8}$: coefficients appearing in the objective function expression; $\epsilon, \alpha$: termination threshold and step size of gradient descent.
\Statex \textbf{Outputs:} $\mathbf{P}$: a (joint) locally-optimal probability-distribution matrix.
\Statex \textbf{Intermediate Variables:} $F_{\textup{prev}}, F_{\textup{cur}}$: the values of the objective function in \eqref{eqn: expected number of cycle-con} at the previous and current iterations; $\mathsf{F}_{1,1}, \mathsf{F}_{-1,-1}, \mathsf{F}_{1,0}, \mathsf{F}_{-1,0}, \mathsf{F}_{0,1}, \mathsf{F}_{0,-1}$: coefficient arrays for the polynomials $f(X_1 X_2, Y_1 Y_2), f(X_1^{-1} X_2^{-1}, Y_1^{-1} Y_2^{-1}),$ $f(X_1 , Y_1), f(X_1^{-1}, Y_1^{-1}), f(X_2 , Y_2), f(X_2^{-1}, Y_2^{-1})$; $\mathsf{G}_{\textup{obj}}^{\{2k,2l\}},\mathsf{G}_{\textup{grad}}^{\{2k,2l,r\}}$ for $(k,l)\in\{(3,3),(3,4),(4,4)\}$ and $r \in \{1,2,3\}$: coefficient arrays of each polynomial term appearing in the objective function and its gradient expressions for $6$-$6$, $6$-$8$, $8$-$8$ cycle concatenation; $\mathbf{G}$: the gradient matrix (of size $(m+1) \times M$) of the objective function, where $\mathbf{g}^{\textup{con}}=\nabla_{\mathbf{p}^{\textup{con}}} N^{\mathcal{D}}(\mathbf{p}^{\textup{con}})$; $\mathcal{T}_{\textup{cur}}$: current MD density.
\State $\mathbf{P}[i,0] \gets \mathbf{p}^*[i]$, $\mathbf{P}[i,j] \gets 0$,  $\forall i\in \{0,1,\dots,m\}$ and $\forall j\in \{1,\dots,M-1\}$. 
\State $F_{\textup{prev}}\gets 0$.
\State $F_{\textup{cur}}\gets 0$ and $\mathbf{G}=\mathbf{0}_{(m+1)\times M}$.
\State Evaluate $\mathsf{F}_{1,1}, \mathsf{F}_{-1,-1}, \mathsf{F}_{1,0}, \mathsf{F}_{-1,0}, \mathsf{F}_{0,1}, \mathsf{F}_{0,-1}$ based on $\mathbf{P}$.
\For{$(k,l)\in\{(3,3),(3,4),(4,4)\}$}
\State $\mathsf{G}_{\textup{obj}}^{\{2k,2l\}} \gets \Lambda_{2k,2l} \cdot \mathrm{conv}(\mathsf{F}_{1,1}, \mathsf{F}_{-1,-1}, \mathsf{F}_{1,0}^{k-1}, \mathsf{F}_{-1,0}^{k-1}, \mathsf{F}_{0,1}^{l-1}, \mathsf{F}_{0,-1}^{l-1})$
\State $\mathsf{G}_{\textup{grad}}^{\{2k,2l,1\}} \gets 2 \Lambda_{2k,2l} \cdot \mathrm{conv}(\mathsf{F}_{1,1}, \mathsf{F}_{1,0}^{k-1}, \mathsf{F}_{-1,0}^{k-1}, \mathsf{F}_{0,1}^{l-1}, \mathsf{F}_{0,-1}^{l-1})$.
\State $\mathsf{G}_{\textup{grad}}^{\{2k,2l,2\}} \gets 2 (k-1) \Lambda_{2k,2l} \cdot \mathrm{conv}(\mathsf{F}_{1,1}, \mathsf{F}_{-1,-1},  \mathsf{F}_{1,0}^{k-1}, \mathsf{F}_{-1,0}^{k-2}, \mathsf{F}_{0,1}^{l-1}, \mathsf{F}_{0,-1}^{l-1})$.
\State $\mathsf{G}_{\textup{grad}}^{\{2k,2l,3\}} \gets 2 (l-1) \Lambda_{2k,2l} \cdot \mathrm{conv}(\mathsf{F}_{1,1}, \mathsf{F}_{-1,-1},  \mathsf{F}_{1,0}^{k-1}, \mathsf{F}_{-1,0}^{k-1}, \mathsf{F}_{0,1}^{l-1}, \mathsf{F}_{0,-1}^{l-2})$.
\EndFor 
\For{$(k,l)\in\{(3,3),(3,4),(4,4)\}$}
\State $F_{\textup{cur}} \gets  F_{\textup{cur}} + w_{2k,2l} \cdot \sum_{M \mid b_1} \sum_{M \mid b_2} \mathsf{G}_{\textup{obj}}^{\{k,l\}}\left[0,0,b_1,b_2\right]$. 
\For{$\forall i\in \{0,1,\dots,m\}$ and $\forall j\in \{0,\dots,M-1\}$}
\State $\mathbf{G}[i,j] \gets \mathbf{G}[i,j] + w_{2k,2l} \cdot \sum_{M \mid b_1} \sum_{M \mid b_2} \mathsf{G}_{\textup{grad}}^{\{2k,2l,1\}}[i, i, b_1+j, b_2+j] + \mathsf{G}_{\textup{grad}}^{\{2k,2l,2\}}[i, 0, b_1+j, b_2] + \mathsf{G}_{\textup{grad}}^{\{2k,2l,2\}}[0, i, b_1, b_2+j]$.
\EndFor
\EndFor 
\State $\mathbf{P} \gets \mathbf{P} - \alpha \displaystyle \frac{\mathbf{G}}{\| \mathbf{g}^{\textup{con}}  \|_2}$.
\State $\mathbf{P} \gets \mathrm{force}(\mathbf{P}, \mathbf{p}^*)$.
\State $\mathcal{T}_{\textup{cur}} \gets 1 - \sum_{i=0}^m \mathbf{P}[i,0]$.
\If{$(\mathcal{T}_{\textup{cur}} < \mathcal{T}_{\textup{max}})$ and $(|F_{\textup{cur}}-F_{\textup{prev}}| > \epsilon)$} 
\State $F_{\textup{prev}} \gets F_{\textup{cur}}$.
\State \textbf{go to} Step 3.
\EndIf
\State \textbf{return} $\mathbf{P}$.
\end{algorithmic}
\end{algorithm}

\begin{remark}
One of our key design choices in the proposed code construction is to separate the SC partitioning and the MD relocation stages in both the FL and the probabilistic continuous (the MD-GRADE) frameworks. In the FL setting, the partitioning and relocation matrices are obtained in two distinct stages. In the continuous case, we first compute the vector $\mathbf{p}^*$, and then compute the matrix $\mathbf{P}$ while enforcing the constraint described in \eqref{eq: row sum constraint}. This constraint ensures that the row sums of $\mathbf{P}$ remain constant, preserving the edge distribution found for partitioning.

An alternative design approach would be to merge these two stages, at least for the probabilistic framework, and directly optimize the matrix $\mathbf{P}$ without separately computing $\mathbf{p}^*$. To explore this, we conducted a small experiment using only the constraint from \eqref{constraint_overall} and modifying the MD-GRADE algorithm accordingly. In this case, we allowed the row sums of $\mathbf{P}$ to vary, enforcing only that the total sum of all elements in $\mathbf{P}$ equals $1$.

We observe that the result of the joint design approach is very close to that of the separate design. This is true in terms of the final distribution matrix $\mathbf{P}$, the objective function value (which is even slightly worse in the joint design case), and the row sums of $\mathbf{P}$, the inferred edge distribution. This result is significant because it suggests that merging the two stages does not offer any performance improvements. Furthermore, keeping the stages separate is advantageous in terms of computational efficiency, particularly regarding the FL algorithms.
\end{remark}

After obtaining a locally-optimal probability-distribution matrix using the MD-GRADE algorithm, we use this matrix to initialize the FL optimizer. First, we quantize the matrix $\gamma\kappa \mathbf{P}$, where the $(i,j)$-th entry after quantization represents the number of circulants in the base matrix assigned to the $i$-th component matrix and the $j$-th auxiliary matrix. 

Next, for each component matrix, we list the circulants assigned to it, which are already specified from the SC code design. Based on the quantized entries of $\gamma\kappa \mathbf{P}$, we determine the number of circulants to be relocated from this component matrix. We then randomly select that number of circulants from the list, and randomly redistribute them across the auxiliary matrices according to the relevant row in the quantized $\gamma\kappa \mathbf{P}$, excluding the auxiliary matrix $\mathbf{X}_0$.

Instead of assigning a uniform relocation probability to each circulant, i.e., picking the circulants to relocate uniformly at random, we adopt an informed initialization strategy based on the underlying SC code. For each circulant, we count how many times it appears in the objects (e.g., cycles or cycle concatenations) that we aim to reduce their count. We then scale the relocation probabilities of the circulants proportionally to these counts, so that those most involved in harmful configurations have a higher likelihood of being relocated during the random initialization phase. We observe that this informed initialization strategy yields a better reduction in the final number of objects in the MD Tanner graph.

\begin{remark}\label{remark:rel_per_dist} Note that the MD-GRADE algorithm provides us with a non-trivial solution. In particular, the relocation percentages for each component matrix $\mathbf{H}_i$ are not necessarily the same. For instance, for the input parameters $(\gamma,\kappa,L,m,M) = (3,20,13,20,4,7)$, the vector of component-wise relocation percentages is $[ 25.90 \textup{ } 65.24 \textup{ } 63.78 \textup{ } 45.65 \textup{ } 28.58],$ where the total relocation percentage is $35.23\%$ (see Fig.~\ref{fig: md_code_6} for the actual probability-distribution matrix). This vector shows that more circulants (percentage-wise) need to be relocated from the middle component matrices of the SC code than the side ones. These percentages are quite close to the final relocation percentages produced by the MCMC algorithm to design GD-GD MD Code~6, where the vector of component-wise relocation percentages is $[ 21.74 \textup{ } 80.00 \textup{ } 75.00 \textup{ } 50.00 \textup{ } 25.00]$.
\end{remark}

\subsection{Markov Chain Monte Carlo Finite-Length Optimizer} \label{subsec: MCMC_algo}

We now introduce our FL algorithm for performing relocations and constructing a GD-MD code. This algorithm adapts the MCMC-based approach introduced in \cite{reins_mcmc} for the relocation stage. Details related to FL partitioning and lifting are skipped since they do not represent novel contributions in this paper. However, the same MCMC approach is adopted for both stages.

Our method relies on an MCMC technique known as \textit{Gibbs sampling} \cite{geman_gibbs, mackay_it}, which is used for sampling from joint probability distributions. For the optimization algorithm, we define a probability-distribution for the inputs of the problem such that the probability is inversely proportional to the objective function. In particular, our approach assigns higher probability to better solutions with lower objective function values and uses Gibbs sampling to sample from this distribution, thereby increasing the likelihood of sampling an optimal/sub-optimal solution eventually.

Let the input of the optimization problem be the vector $\mathbf{x}$, and let the objective function be $C(\mathbf{x})$. Here, $\mathbf{x}$ is the row-concatenated version of the relocation matrix $\mathbf{M}$, and $C(\mathbf{x})$ is the number of objects of interest (cycles or cycle concatenations), or a weighted sum of the numbers of different types of objects, in the Tanner graph of the GD-MD code that we aim to minimize. The probability assigned to $\mathbf{x}$ is defined as:
\begin{align}
P(\mathbf{x}) = \frac{P^*(\mathbf{x})}{Z(\beta)},
\end{align}
where
\begin{align}
P^*(\mathbf{x}) = 
\begin{cases}
e^{-\beta C(\mathbf{x})}, & \text{if } \mathbf{x}\in F , \\
0, & \text{otherwise,}
\end{cases}
\end{align}
and $Z(\beta)$ is the normalization constant given by
\begin{align}
Z(\beta) = \sum_{\mathbf{x} \in F} e^{-\beta C(\mathbf{x})}.
\end{align}
Here, $\beta \in \mathbb{R}^{\geq 0}$ is a hyperparameter of the algorithm, and $F$ is the feasible set of the problem.

The update rule of Gibbs sampling is as follows:
\begin{align} \label{eq:conditional_probability}
x_i^{(t+1)} \sim P\left(x_i \mid x_1^{(t+1)}, \ldots, x_{i-1}^{(t+1)}, x_{i+1}^{(t)}, \ldots, x_n^{(t)}\right) = \frac{P^*([x_1^{(t+1)} \textup{ } \ldots \textup{ } x_{i-1}^{(t+1)} \textup{ } x_i \textup{ } x_{i+1}^{(t)} \textup{ } \ldots \textup{ } x_n^{(t)}])}{\sum_{x'_i=0}^{M-1} P^*([x_1^{(t+1)} \textup{ } \ldots \textup{ } x_{i-1}^{(t+1)} \textup{ } x'_i \textup{ } x_{i+1}^{(t)} \textup{ } \ldots \textup{ } x_n^{(t)}])},
\end{align}
where $x_i$ is the $i$-th element of $\mathbf{x}$, $t$ denotes the iteration index, and $n = \gamma\kappa$ denotes the length of $\mathbf{x}$. While Gibbs sampling traditionally updates one variable at a time in sequential index order (i.e., first index, then second, and so on), it can be adapted to update a small number of variables simultaneously in a randomized order. Here, each feasible $\mathbf{x} \in F$ represents a state in a Markov chain, with the transition probabilities between states determined by~\eqref{eq:conditional_probability}.

In our algorithm, we update a small set of entries in $\mathbf{x}$ per iteration, where the cardinality of updated sets is chosen as some integer $\delta$. The algorithm maintains a list of index sets whose length is the same as that of the input vector $\mathbf{x}$. Each set in the list includes one index $\zeta$ from $\{1,2,\ldots, n\}$ and the $\delta-1$ indices of the most correlated entries, where correlation is determined based on the number of common objects of interest involving the entries corresponding to these indices.

The algorithm iterates over a list of such index sets. For each set, it evaluates the conditional probability mass function (PMF) and sample the new values of the input variables using an extension of \eqref{eq:conditional_probability} to the multi-variable case (see \cite[(11)]{reins_mcmc}). This list is shuffled at each pass to avoid introducing systematic Markov chain transition patterns.

The algorithm keeps track of the minimum objective function value observed until its current iteration, along with its corresponding input vector. These records are updated whenever a better solution is encountered during the evaluation of objective function values within the conditional probability computation.

The pseudo-code of the algorithm is provided in Algorithm~\ref{algo:MCMC}. The final relocation matrix $\mathbf{M}$ is obtained by reshaping the optimal input vector found by the algorithm, $\mathbf{x}_{\textup{opt}}$, to the matrix format.

\begin{algorithm}
\caption{Markov Chain Monte Carlo (MCMC) Finite-Length Optimizer} \label{algo:MCMC}
\begin{algorithmic}[1]
\Statex \textbf{Inputs:} $\mathcal{L}$: list of objects of interest; $\delta$: cardinality of index sets; $\mathbf{x}_{\textup{init}}$: initial input vector; $\beta_{\text{init}}$: initial value of $\beta$; $T$: maximum number of iterations; $M$: number of auxiliary matrices.
\Statex \textbf{Outputs:} $C_{\textup{opt}}$: minimum value of the objective function recorded during iterations; $\mathbf{x}_{\textup{opt}}$: optimal value of the input vector resulting in the minimum objective function.
\Statex \textbf{Intermediate Variables:} $S$: list of index sets; $\mathbf{x}$: input vector for the current iteration; $i$: number of samples; $\beta$: hyper-parameter of the main distribution; $\nu$: index set of entries that are updated; $Z$: normalizing factor for probabilities; $\mathbf{x}', \mathbf{x}_{\textup{prev}}$: input vectors used for intermediate calculations; $P_\nu ,P_\nu^*$: normalized and non-normalized mappings of conditional PMF values of transition probabilities between $\mathbf{x}'$ and $\mathbf{x}$, where $\mathbf{x}'$ and $\mathbf{x}$ differ only at entries indexed by $\nu$.
\State Initialize $S$ by going over $\mathcal{L}$ and calculating the common object counts for all entry pairs, then list $d$ indices of the most correlated entries for each entry. 
\State $\mathbf{x} \gets \mathbf{x}_{\text{init}}$.
\State Initialize $C_{\textup{opt}} \gets \infty$, $i \gets 0$, $\beta \gets \beta_{\text{init}}$.
\While{$i < T$}
	\State Shuffle the order of $S$ randomly.
    \ForEach{$\nu \in S$}
        \State $i \gets i+1$, $Z \gets 0$, $\mathbf{x}_{\textup{prev}}\gets \mathbf{x}$.
        \ForEach{$\mathbf{x}'_{\nu} \in \{0,1,\dots,M-1\}^{\delta}$}
            \State Obtain $\mathbf{x}'$ by assigning its entries indexed by $\nu$ from $\mathbf{x}'_{\nu}$, and taking the remaining entries from $\mathbf{x}$.
            \If {$\mathbf{x}'$ does not satisfy the constraints}
            	\State $P_\nu(\mathbf{x}') \gets 0$.
            \Else
                \State Evaluate $C(\mathbf{x}')$.
                \If{$C(\mathbf{x}') < C_{\textup{opt}}$}
                    \State $C_{\textup{opt}} \gets C(\mathbf{x}')$ and $\mathbf{x}_{\textup{opt}}\gets \mathbf{x}'$.
                    \If{$C_{\textup{opt}} = 0$} \textbf{go to} Step~28.
                    \EndIf
                \EndIf
                \State $P_\nu^*(\mathbf{x}') \gets \exp(-\beta C(\mathbf{x}'))$.
                \State $Z \gets Z + P_\nu^*(\mathbf{x}')$.
            \EndIf
        \EndFor
        \State $P_\nu \gets P_\nu^* / Z$. \textit{// Normalization to reach a probability distribution, for all $\mathbf{x}$.}
        \State Sample from the distribution $P_\nu$ to reach $\mathbf{x}$.
    \EndFor
    \State Update MCMC internal and distribution variables. \textit{// This is to control the transition rate, see \cite{reins_mcmc} for details.}
\EndWhile
\State \textbf{return} $C_{\textup{opt}}$, and $\mathbf{x}_{\textup{opt}}$.
\end{algorithmic}
\end{algorithm}

Before executing the MCMC algorithm, we extract a list of all active objects in the underlying SC code whose instances we aim to minimize \cite{hashemi}. We emphasize that only elementary objects are considered in this step. The evaluation of the objective function involves iterating over this list and checking whether each object remains active for the given relocation arrangement. Specifically, a cycle concatenation is considered active if both cycles in its basis remain active after relocations \cite{rohith2}.

As discussed earlier, the initial relocation matrix is derived from the MD-GRADE algorithm. To limit the search space, allowable input vectors are constrained to lie within fixed $L_1$ and $L_\infty$ distances from the initial vector. Additionally, we impose an upper bound on the MD density, typically set to $0.35$, i.e., to $35\%$ relocations. The main reason behind limiting the MD density is to enable low-latency windowed decoding that can take advantage of that most of the non-zero entries are within $\mathbf{X}_0$ and its copies, which is part of our future work. Any input vector that violates these constraints is assigned zero probability, and thus excluded from the sampling process.

When the algorithm targets cycles, it first attempts to eliminate all instances of cycle-$6$. If successful, it then proceeds to minimize the number of cycle-$8$ instances, while ensuring that all cycle-$6$ candidates are inactive by assigning zero probability to any vector that would activate any of them. In the cycle concatenation case, the objective function is defined as the weighted sum of the number of active $6$-$6$, $6$-$8$, and $8$-$8$ configurations. The weights are chosen as in Algorithm~\ref{algo: MD-GRADE-CON}. Observe that all underlying SC codes we adopt do not have any cycles-$4$.

%%%%%%%%%%%%%%%%%%%%%%%%%%%%%%%%%%%%%%%%%%%%%%%%%%%%%%%%%%%%%%%%%%%%%%%%%%%%%%%%%%%%%

\vspace{-0.1em}
\section{Numerical Results and Comparisons} \label{sec: numeric}

In this section, we present MD-SC codes constructed using our probabilistic framework, report their short cycle and object counts, and demonstrate their error-rate performance with simulation plots. Moreover, we report the estimation results for short cycles as derived in Theorem~\ref{thm: forecast}. We design eight sets of MD-SC codes (MD Codes~1--8) for this section, with their parameters, codeword lengths, and design rates listed in Table~\ref{table: code_parameters}. As presented in Table~\ref{table: code_parameters}, our framework supports the design of MD-SC codes across a wide spectrum of code lengths and design rates. These codes are suitable for use in diverse application areas, including communication systems (wireless, optical, etc.) and modern data storage systems.

\begin{table}
\caption{Parameters, Lengths and Design Rates of Proposed MD-SC Codes Where Each Entry Refers to a Set of Codes}
\vspace{-0.5em}
\centering
\scalebox{1.00}
{
\renewcommand{\arraystretch}{1.2}
\begin{tabular}{|c|c|c|c|c|c|c|c|c|}
\hline
\makecell{Code name} & \makecell{$\gamma$} & \makecell{$\kappa$} & \makecell{$z$} & \makecell{$L$} & \makecell{$m$} & \makecell{$M$} & \makecell{Length} & \makecell{Design \\ rate} \\
\hline
MD Code~1 & $4$ & $17$ & $17$ & $10$ & $1$ & $3$ & $8{,}670$ & $0.74$ \\
\hline
MD Code~2 & $3$ & $19$ & $23$ & $10$ & $2$ & $4$ & $17{,}480$ & $0.81$ \\
\hline
MD Code~3 & $3$ & $17$ & $7$ & $50$ & $9$ & $13$ & $77{,}350$ & $0.79$ \\
\hline
MD Code~4 & $4$ & $29$ & $29$ & $50$ & $19$ & $5$ & $210{,}250$ & $0.81$ \\
\hline
MD Code~5 & $4$ & $24$ & $17$ & $30$ & $6$ & $3$ & $36{,}720$ & $0.80$ \\
\hline
MD Code~6 & $3$ & $20$ & $13$ & $20$ & $4$ & $7$ & $36{,}400$ & $0.82$ \\
\hline
MD Code~7 & $4$ & $13$ & $5$ & $10$ & $3$ & $5$ & $3{,}250$ & $0.60$ \\
\hline
MD Code~8 & $4$ & $17$ & $7$ & $15$ & $4$ & $5$ & $8{,}925$ & $0.70$ \\
\hline
\end{tabular}}
\label{table: code_parameters}
\vspace{-0.5em}
\end{table}

We use MD Codes~1--4 to examine optimization targeting cycle counts, whereas MD Codes~5--8 are generated to examine optimization targeting object counts. We refer to the FL construction method performed by the MCMC algorithm guided by a GD-based distribution as the GD-MCMC method, both for partitioning and relocation design stages. Occasionally, we drop the word ``method'' from ``the MCMC method'' and ``the GD-MCMC method'' for brevity. We also include codes constructed with a uniform distribution for comparison purposes. Here, by a uniform distribution for the partitioning stage, we mean that edge distribution for the underlying SC code is $\mathbf{p}_{\textup{UNF}} = \frac{1}{m+1}\cdot \mathbf{1}$. On the other hand, by a uniform distribution for the relocation stage, we mean that the relocation density is applied evenly across all component matrices and distributed uniformly among all auxiliary matrices. If we denote the edge distribution for partitioning by $\mathbf{p}$ and the relocation density by $\mathcal{T}$, then the uniform probability-distribution matrix takes the following form:
\begin{align}
\mathbf{P}_{\textup{UNF}} = \begin{bmatrix}
(1 - \mathcal{T})\mathbf{p}[0] & \frac{\mathcal{T}}{M-1} \mathbf{p}[0] & \dots & \frac{\mathcal{T}}{M-1} \mathbf{p}[0] \\
(1 - \mathcal{T})\mathbf{p}[1] & \frac{\mathcal{T}}{M-1} \mathbf{p}[1] & \dots & \frac{\mathcal{T}}{M-1} \mathbf{p}[1] \\
\vdots & \vdots & \ddots & \vdots \\
(1 - \mathcal{T})\mathbf{p}[m] & \frac{\mathcal{T}}{M-1} \mathbf{p}[m] & \dots & \frac{\mathcal{T}}{M-1} \mathbf{p}[m] \\
\end{bmatrix}.
\end{align}

The design methods for each set of codes are listed below in more detail:
\begin{enumerate}
    \item MD Codes~1--2: Partitioning matrices are generated using the optimal overlap (OO) method~\cite{oocpo}, since the memory is low for both codes. Lifting is performed via the MCMC, and relocations are conducted using the GD-MCMC. All three stages of the code design aim to reduce short cycle counts.
    
    \item MD Codes~3--4: Partitioning and relocations are both done using the GD-MCMC, while lifting is again performed via the MCMC. The design goal in all three stages is to reduce short cycle counts.
    
    \item MD Code~5: Three variants of the code are created, differing mainly in their partitioning and relocation distributions. One is generated by GD-based distributions for both stages (GD-GD MD Code~5), another is generated by GD distribution for partitioning and uniform distribution for relocation (GD-UNF MD Code~5), and the last is generated by uniform distributions in both stages (UNF-UNF MD Code~5). The underlying SC codes are obtained from~\cite[$(4,24)$-BSC]{GRADE}, where partitioning is performed using the algorithmic optimizer of~\cite{GRADE} and lifting is done using the circulant power optimizer (CPO) of~\cite{GRADE}, which is based on that of~\cite{oocpo}. Relocations are done using the MCMC for all variants. In all cases, the optimization targets objects (concatenated cycles) for all three design stages.
    
    \item MD Codes~6--8: Similar to MD Code~5, three variants (GD-GD, GD-UNF, and UNF-UNF) are generated for each code of the three. Partitioning and relocations are guided by the respective distributions and done using the MCMC, and lifting is also performed via the MCMC. Here, the design targets short cycles during the partitioning and lifting stages, and concatenated cycles (objects) during the relocation stage. The rationale is that by focusing on short cycles initially, we can eliminate most cycles-$6$, effectively removing $6$-$6$ and $6$-$8$ configurations. This allows the relocation stage to focus on reducing $8$-$8$ configurations, which appear in many detrimental absorbing sets.
\end{enumerate}

For MD Codes~5--8, it should be noted that our main contribution lies in the design of the GD-GD MD codes. The GD-UNF and UNF-UNF variants are generated for comparison purposes, in order to highlight the improvement offered by our proposed framework compared with a straightforward distribution.

\begin{table}
\centering
\caption{Reduced Short Cycle Counts of Proposed MD-SC Codes}
\scalebox{1.00}
{
\renewcommand{\arraystretch}{1.2}
\begin{tabular}{|c|c|c|c|}
\hline
\makecell{Code name} & \makecell{Cycle-$6$ \\ count} & \makecell{Relocation \\ percentage} & \makecell{Underlying \\ SC count} \\ 
\hline
\cite{homa-lev} & $9{,}078$ & $33.82\%$ & $-$ \\
\hline
MD Code~1 & $3{,}366$ & $33.82\%$ & $25{,}211$ \\ 
\hline
\addlinespace
\hline
\makecell{Code name} & \makecell{Cycle-$8$ \\ count} & \makecell{Relocation \\ percentage} & \makecell{Underlying \\ SC count} \\ 
\hline
\cite{homa-lev} & $249{,}320$ & $33.33\%$ & $-$ \\
\hline
MD Code~2 & $206{,}356$ & $33.33\%$ & $282{,}693$ \\ 
\thickhline
MD Code~3 & $15{,}652$ & $31.37\%$ & $164{,}794$ \\ 
\thickhline
MD Code~4 & $1{,}581{,}080$ & $34.48\%$& $2{,}740{,}036$ \\ 
\hline
\end{tabular}}
\label{table: cycle_counts} 
\end{table}

In Table~\ref{table: cycle_counts}, we list the short cycle counts for MD Codes~1--4, along with the relocation percentages, and the cycle counts of the corresponding underlying SC codes (which share the same design rate). All codes listed in the table are free of cycles-$4$, and those shown in the lower panel of the table are also free of cycles-$6$. For comparison, we also include two codes from~\cite{homa-lev}, aligned with MD Code~1 and MD Code~2, respectively. In particular, each of these literature codes shares the same design parameters, partitioning matrices, and relocation densities with our corresponding GD-MD code. It should be noted, however, that the lifting matrices differ: ours are obtained via the MCMC algorithm, whereas theirs are generated using the CPO method.

In~\cite{homa-lev}, the authors introduce the concept of the depth of an MD-SC code, where during the relocation stage, not all $M$ auxiliary matrices are used. Instead, only the first $d\leq M$ of them are utilized, where $d$ is referred to as the depth. This approach could enable the use of a windowed decoder in multiple dimensions, which would reduce decoding latency, while also mitigating the effect of short cycles. They report the number of cycles-$6$ in the final Tanner graph of an MD-SC code with parameters $(\gamma, \kappa, z, L, m, M) = (4, 17, 17, 10, 1, 5)$, while varying $d$ from $2$ to $5$. For comparison, we also generated MD-SC codes using the same underlying SC code, with slight modifications to our method. Specifically, in the MD-GRADE algorithm, we force the columns of the probability-distribution matrix indexed by all $i$ such that $i\geq d$ to be zero at all iterations. During FL optimization, we constrain the MCMC algorithm to assign relocation values only from the set $\{0,1, \dots, d-1\}$ (Algorithm~\ref{algo:MCMC}, Step-8 changes accordingly). In Table~\ref{table: cycle_counts_with_depth}, we present our resulting cycle counts (right side) alongside those reported in~\cite{homa-lev} (left side) for direct comparison. It should be noted that all of these codes share the same relocation density, which is $26.47\%$. The notable reduction in cycle counts achieved by our GD-MD codes is demonstrated by the table.

\begin{remark}
Although we generated MD-SC codes using the first $d$ auxiliary matrices, the selection of which auxiliary matrices to use can be made in a more informed way. This is similar to the case of topologically-coupled codes \cite{GRADE}, where some component matrices are deliberately set to all zeros for an SC code, and the selection of the non-zero component matrices is guided by a similar probabilistic framework through the formulation of an optimization problem. Exploring this problem could be a promising direction for future research.
\end{remark}

In Table~\ref{table: estimation_results}, we list the lower and upper bounds as well as the estimated cycle counts based on Theorem~\ref{thm: forecast} for MD Codes~1--4, along with their corresponding actual cycle counts.

In Table~\ref{table: object_counts}, we list the elementary $6$-$6$, $6$-$8$, and $8$-$8$ configuration counts in the final Tanner graphs of MD Codes~5--8, for all three distribution variants (GD-GD, GD-UNF, UNF-UNF). Relocation percentages and counts of these configurations for the underlying SC codes are also provided. It should be noted that we only provide underlying SC counts for the GD-GD MD codes (where the SC code is designed using the GD algorithm).

\begin{remark}\label{remark:extended_sc}
Instead of constructing an MD-SC code via relocations, one can simply increase the number of replicas to $L \times M$ of an SC code while keeping all other parameters and partitioning/lifting matrices fixed; the resulting code is referred as the extended SC code. Extended SC codes have the same length as their MD-SC counterparts, and have close design rates when $L$ is sufficiently larger than $m$, ensuring a fair comparison. For two of the MD-SC codes considered in this work, we also constructed their extended SC counterparts (note that the rate loss of MD-SC codes is not higher than $0.025$ for these cases). For MD Code~1, the extended SC code counterpart has $78{,}591$ cycles-$6$, yielding a reduction of approximately $95.7\%$ achieved by our MD-SC construction. Similarly, the extended SC counterpart of GD-GD MD Code~6 contains $17{,}287{,}933$ many $8$-$8$ configurations, corresponding to a reduction of about $99.3\%$ achieved by our MD-SC construction.
\end{remark}

\begin{table}
\centering
\caption{Cycle-$6$ Counts of MD-SC Codes With Parameters $(\gamma, \kappa, z, L, m, M) = (4, 17, 17, 10, 1, 5)$ and Varying Depth}
\scalebox{1.00}{
\renewcommand{\arraystretch}{1.2}
\begin{tabular}{|c|c|c|}
\hline
\multirow{2}{*}{Depth} & \multicolumn{2}{c|}{Cycle-$6$ count} \\
\cline{2-3}
    & \cite{homa-lev} & GD-MD Code \\
\hline
$2$ & $20{,}825$ & $14{,}195$ \\
\hline
$3$ & $9{,}775$ & $4{,}080$ \\
\hline
$4$ & $5{,}695$ & $1{,}615$ \\
\hline
$5$ & $5{,}610$ & $765$ \\
\hline
\end{tabular}}
\label{table: cycle_counts_with_depth}
\end{table}

\begin{table}
\centering
\caption{Cycle Counts With Their Estimated Values, Lower And Upper Bounds}
\scalebox{1.00}
{
\renewcommand{\arraystretch}{1.2}
\begin{tabular}{|c|c|c|c|c|}
\hline
\makecell{Code Name} & \makecell{Cycle count} & \makecell{Estimate} & \makecell{Lower bound} & \makecell{Upper bound} \\
\hline 
MD Code~1 & $3{,}366$ & $49{,}782$ & $47{,}162$ & $52{,}402$ \\
\hline
MD Code~2 & $206{,}356$ & $226{,}650$ & $169{,}990$ & $283{,}310$ \\
\hline
MD Code~3 & $15{,}652$ & $203{,}340$ & $158{,}700$ & $247{,}970$ \\
\hline
MD Code~4 & $1{,}581{,}080$ & $2{,}553{,}500$ & $988{,}450$ & $4{,}118{,}600$ \\
\hline
\end{tabular}}
\label{table: estimation_results} 
\end{table}

\begin{table*}
\centering
\caption{Object Counts of Proposed MD-SC Codes and Their Underlying SC Codes}
\scalebox{1.00}
{
\renewcommand{\arraystretch}{1.2}
\begin{tabular}{|c|c|c|c|c|c|}
\hline
 & \makecell{Code} & \makecell{$6$-$6$ configuration \\ count} & \makecell{$6$-$8$ configuration \\ count} & \makecell{$8$-$8$ configuration \\ count} &  \makecell{Relocation \\ percentage} \\
\hline
\multirow{4}{*}{MD Code~5} & Underlying SC & $1{,}292$ & $859{,}010$ & $88{,}781{,}378$ & $-$ \\
\cline{2-6}
    & UNF-UNF MD & $0$ & $580{,}992$ & $49{,}166{,}448$ & $33.33\%$ \\
\cline{2-6}
	& GD-UNF MD & $0$ & $357{,}459$ & $33{,}641{,}283$ & $34.38\%$ \\	
\cline{2-6}
	  & GD-GD MD & $0$ & $44{,}115$ & $23{,}612{,}898$ & $33.33\%$ \\
\hline
\addlinespace
\hline
\multirow{4}{*}{MD Code~6} & Underlying SC & $0$ & $0$ & $2{,}001{,}493$ & $-$ \\
\cline{2-6}
	& UNF-UNF MD & $0$ & $0$ & $535{,}535$ & $35.00\%$ \\
\cline{2-6}
	& GD-UNF MD & $0$ & $0$ & $331{,}058$ & $35.00\%$ \\
\cline{2-6}
    & GD-GD MD & $0$ & $0$ & $112{,}931$ & $35.00\%$ \\
\hline
\addlinespace
\hline
\multirow{4}{*}{MD Code~7} & Underlying SC & $4{,}305$ & $261{,}280$ & $5{,}984{,}110$ & $-$ \\
\cline{2-6}
	& UNF-UNF MD & $3{,}700$ & $142{,}700$ & $2{,}053{,}375$ & $34.62\%$ \\
\cline{2-6}
	& GD-UNF MD & $3{,}250$ & $81{,}875$ & $1{,}454{,}000$ & $34.62\%$ \\
\cline{2-6}
    & GD-GD MD & $0$ & $11{,}775$ & $980{,}750$ & $34.62\%$ \\
\hline
\addlinespace
\hline
\multirow{4}{*}{MD Code~8} & Underlying SC & $5{,}243$ & $619{,}150$ & $23{,}417{,}464$ & $-$ \\
\cline{2-6}
	& UNF-UNF MD & $9{,}870$ & $483{,}105$ & $11{,}861{,}570$ & $33.82\%$ \\
\cline{2-6}
	& GD-UNF MD & $1{,}225$ & $124{,}250$ & $5{,}585{,}685$ & $35.29\%$ \\
\cline{2-6}
    & GD-GD MD & $0$ & $0$ & $4{,}653{,}845$ & $35.29\%$ \\
\hline
\end{tabular}}
\label{table: object_counts} 
\end{table*}

All variants of MD Code~6, GD-GD/GD-UNF variants of MD Code~7, and GD-GD/UNF-UNF variants of MD Code~8 are simulated over the additive white Gaussian noise channel (AWGNC) using a sum-product (fast Fourier transform based) decoder~\cite{GRADE}. The results are presented in Fig.~\ref{fig: gamma3_code}, Fig.~\ref{fig: shorter_code}, and Fig.~\ref{fig: short_code}, respectively. Note that the GD-UNF MD Code~7 exhibits an error floor, and we observe that the main cause of the faulty frames is the presence of two detrimental objects in the Tanner graph of the code, namely a $(4,4)$ UAS and a $(5,4)$ elementary unlabeled oscillating set (UOS).\footnote{A UOS is a special case of a UTS, with the condition that for each VN, the number of neighboring even-degree CNs is greater than or equal to the number of neighboring odd-degree CNs, and there exists at least one VN for which the equality holds \cite{harredy_nb}.} The graph representations of these two objects are also presented in Fig.~\ref{fig: error_floor}.

Below are some takeaways from these results summarized in the tables and figures: 
\begin{enumerate} 
	\item Our probabilistic framework offers a significant reduction in the number of cycles-$6$ and cycles-$8$ for our GD-MD codes. These reductions can be as high as $95.7\%$ relative to the extended SC counterparts, even when the rate loss is negligible (see Remark~\ref{remark:extended_sc}). More intriguingly, the GD-MD codes have remarkably lower cycle counts compared with their underlying SC codes, despite the latter having notably shorter lengths. In particular, the proposed framework offers an $86.6\%$ reduction in cycle-$6$ count in the case of MD Code~1 relative to its underlying SC code, as well as reductions in cycle-$8$ counts ranging from $27.0\%$ to $90.5\%$ in the cases of MD Codes~2--4.

	\item Our GD-MD codes outperform relevant MD-SC codes in the literature \cite{homa-lev} in terms of cycle counts. Starting with MD Code~1, which has $3{,}366$ cycles-$6$, surpassing the relevant MD-SC code designed by \cite[Algorithm~2]{homa-lev}, which has $9{,}078$ cycles-$6$. Similarly, MD Code~2 has $206{,}356$ cycles-$8$, whereas the relevant MD-SC code designed in \cite{homa-lev} has $249{,}320$ cycle-$8$. Additionally, by incorporating the concept of MD depth, where only a subset of auxiliary matrices are utilized for relocations, our method is able to generate MD-SC codes with significantly lower cycle counts compared with~\cite{homa-lev}, under identical code parameters, underlying SC code, and MD density, which could be observed from Table~\ref{table: cycle_counts_with_depth}.

	\item Since all three stages of the code design are performed to minimize the number of cycles for MD Codes~1--4, these numbers are lower than the estimates obtained based on Theorem \ref{thm: forecast} for our GD-MD codes. For MD Code~2, Theorem \ref{thm: forecast} gives a good indicator of the outcome of the MCMC algorithm, where the (rounded) expected number of cycles-$8$ is $226{,}650$ and the actual final number of cycles-$8$ is $206{,}356$. Similarly, for MD Code~4, the expected number of cycles-$8$ is $2{,}553{,}500$, and the actual final number of cycles-$8$ is $1{,}581{,}080$. However, the expected number of cycles-$8$ of MD Code~3 is $203{,}340$ and the actual number of cycles-$8$ obtained is $15{,}652$, highlighting the strength of our FL design.	
	
	\item When we compare our GD-GD MD Codes with their underlying SC and their GD-UNF/UNF-UNF MD counterparts, we observe a substantial reduction in the number of detrimental objects. Key findings are summarized below:
    \begin{itemize}
        \item All GD-GD MD Codes~5--8 are entirely free of elementary $6$-$6$ configurations. In contrast, their underlying SC codes and their GD-UNF/UNF-UNF variants do contain such objects in the case of MD Codes~7 and 8.

        \item Regarding elementary $6$-$8$ configurations, GD-GD MD Code~8 contains none of these objects, even though its counterparts exhibit large populations. Additionally, for MD Code~5 and MD Code~7, the reduction percentages are:
        \begin{itemize}
            \item Relative to GD-UNF: $87.65\%$ and $85.62\%$, respectively.
            \item Relative to UNF-UNF: $92.40\%$ and $91.75\%$, respectively.
            \item Relative to SC codes: $94.86\%$ and $95.49\%$, respectively.
        \end{itemize}

        \item For elementary $8$-$8$ configurations, the reduction percentages range as follows:
        \begin{itemize}
            \item Compared with GD-UNF: from $16.68\%$ to $65.89\%$.
            \item Compared with UNF-UNF: from $51.97\%$ to $78.91\%$.
            \item Compared with SC codes: from $73.40\%$ to $94.35\%$.
        \end{itemize}
        Moreover, the reduction reaches up to $99.3\%$ relative to the extended SC counterparts, while incurring only a minimal rate loss (as discussed in Remark~\ref{remark:extended_sc}).
    \end{itemize}

\begin{figure}
\centering
\includegraphics[width=0.45\textwidth]{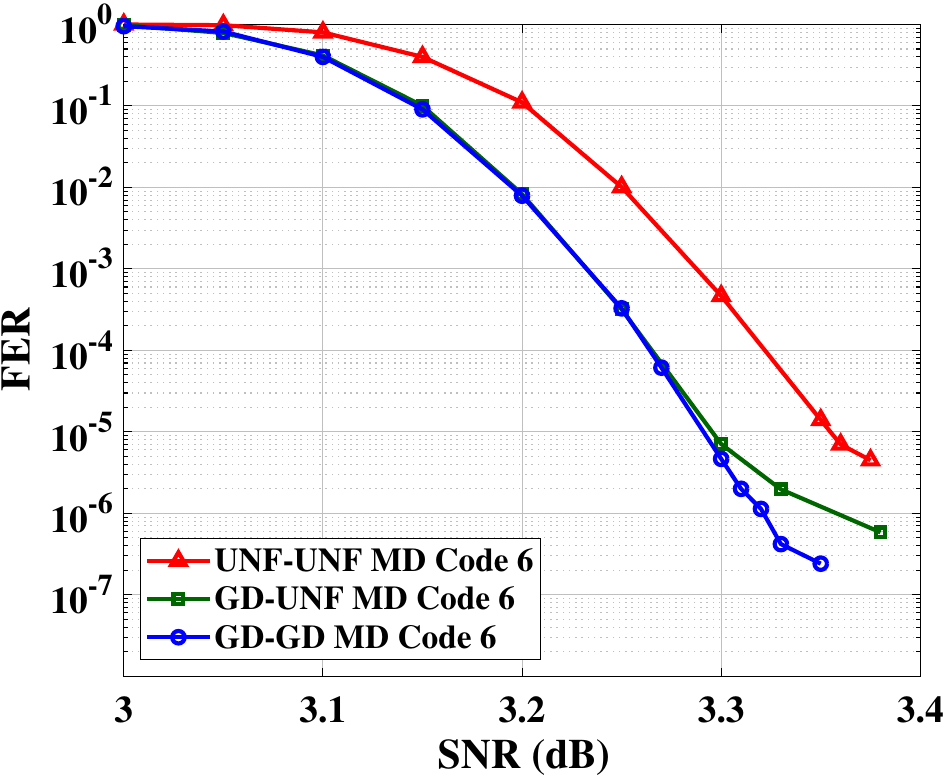} \vspace{-0.8em}
\caption{FER curves of GD-GD/GD-UNF/UNF-UNF MD Code 6 over the AWGNC. All curves correspond to codes with parameters $\left(\gamma, \kappa, z, L, m, M\right) = \left(3,20,13,20,4,7\right)$.} \vspace{-1.0em}
\label{fig: gamma3_code}
\end{figure}

\begin{figure}
\centering
\includegraphics[width=0.45\textwidth]{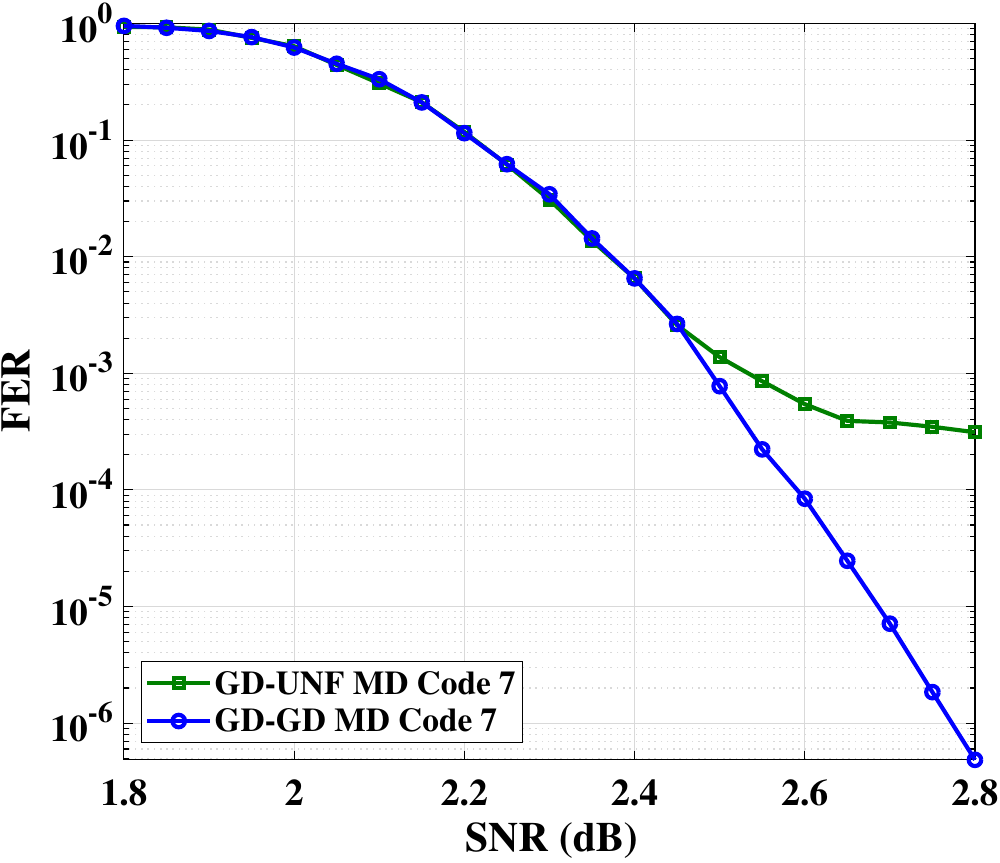} \vspace{-0.8em}
\caption{FER curves of GD-GD/GD-UNF MD Code 7 over the AWGNC. Both curves correspond to codes with parameters $\left(\gamma, \kappa, z, L, m, M\right) = \left(4,13,5,10,3,5\right)$.} \vspace{-1.0em}
\label{fig: shorter_code}
\end{figure}

	\item Based on the frame error rate (FER) simulations, we observe that the GD-GD MD codes exhibit considerable performance improvements over their GD-UNF and UNF-UNF counterparts. When compared with the UNF-UNF MD codes, the GD-GD MD codes consistently outperform them across the entire SNR range. In particular, the FER gains reach up to $1.76$ orders of magnitude at $3.35$~dB SNR for MD Code~6, and up to $3.08$ orders of magnitude at $2.8$~dB SNR for MD Code~8. Furthermore, at $10^{-6}$ FER level, GD-GD MD Code~8 has approximately $0.2$~dB SNR improvement over the corresponding UNF-UNF code.

    We also observe that the GD-GD MD codes outperform their GD-UNF counterparts in the high-SNR region. Both GD-UNF MD Codes~6--7 exhibit an error floor, whereas the GD-GD variants achieve significant gains of FER up to $0.67$ orders of magnitude at $3.33$~dB SNR for MD Code~6, and up to $2.81$ orders of magnitude at $2.8$~dB SNR for MD Code~7.

    It is important to note that the error floor is directly linked to the large number of detrimental objects present in the GD-UNF codes. As illustrated in Fig.~\ref{fig: error_floor}, the dominant problematic structures responsible for the error floor in GD-UNF MD Code~7 Tanner graph contain elementary $6$-$6$ and $6$-$8$ configurations.
	
	\item Since the number of possible relocation matrices is $M^{\gamma\kappa}$, the design space of our problem is extremely large and it grows rapidly with the code parameters. The main goal of our probabilistic framework is to guide the FL optimization by providing an informed initialization point that is likely near an optimal solution, effectively reducing the search space and improving the efficiency of heuristic methods such as the MCMC. Moreover, in~\cite{reins_mcmc}, the existence of multiple recurrent classes in the Markov chain constructed by the MCMC algorithm is discussed, and it is shown that the difference between the final result of the algorithm and the locally-optimal solution within the recurrent class that the algorithm operates on is bounded by a small margin. This implies that initialization strongly influences the final outcome, particularly for the MCMC method. Our results show that when the MCMC is initialized with the MD-GRADE distribution, it achieves significantly better results than when initialized with the uniform distribution, both in object counts (see Table~\ref{table: object_counts}) and in error-rate performance (see Fig.~\ref{fig: gamma3_code}, Fig.~\ref{fig: shorter_code}, and Fig.~\ref{fig: short_code}). This suggests that our probabilistic framework enables the MCMC to operate within a ``good'' recurrent class and to reach a better result.

    \item In general, for SC codes, the locally-optimal edge distribution obtained through the framework of~\cite{GRADE} indicates that the component matrices near the edges should contain a higher concentration of circulants, while those in the middle should exhibit lower populations. Intriguingly, for the locally-optimal probability-distribution matrices produced by the MD-GRADE algorithm, we observe that the component-wise relocation percentages tend to be lower for the edge component matrices and higher for the middle ones (as discussed in Remark~\ref{remark:rel_per_dist} for the specific case of GD-GD MD Code-6).  

    Taken together, these two probabilistic frameworks suggest that to optimize FL performance, fewer circulants should be allocated to the middle component matrices during partitioning, and larger portion of them should be relocated to auxiliary matrices during relocations. Our numerical analyses, including both object counts and error-rate results, support this conclusion by demonstrating that such a non-trivial distribution of edges yields significantly improved MD-SC codes compared with those constructed with a straightforward edge distribution (such as uniform). The shape of this optimized distribution is itself noteworthy, and further investigating the underlying reasons for its effectiveness using different tools could be a promising direction for future research.

\end{enumerate}

\begin{figure}
\centering
\includegraphics[width=0.45\textwidth]{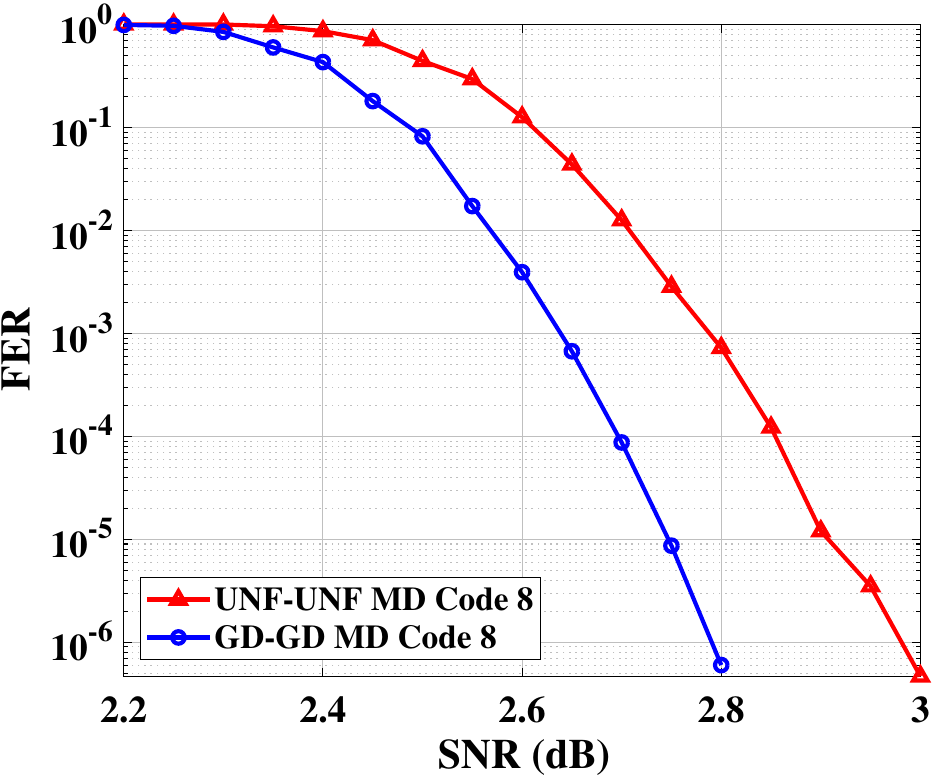} \vspace{-0.8em}
\caption{FER curves of GD-GD/UNF-UNF MD Code 8 over the AWGNC. Both curves correspond to codes with parameters $\left(\gamma, \kappa, z, L, m, M\right) = \left(4,17,7,15,4,5\right)$.} \vspace{-1.0em}
\label{fig: short_code}
\end{figure}

\begin{figure}
\centering
\includegraphics[width=0.5\textwidth]{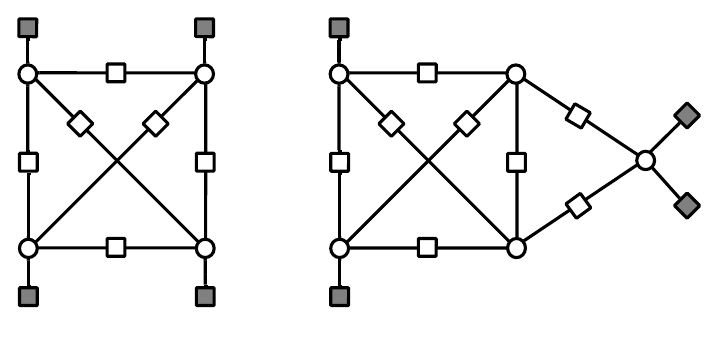}
\vspace{-1.5em}
\caption{A $(4,4)$ UAS (left) and $(5,4)$ UOS (right) in the Tanner graph of GD-UNF Code 7, with parameters $\left(\gamma, \kappa, z, L, m, M\right) = \left(4, 13, 5, 10, 3, 5\right)$. These configurations are identified as the main reason behind the faulty frames in the error-floor region.}
\label{fig: error_floor}
\end{figure}

%%%%%%%%%%%%%%%%%%%%%%%%%%%%%%%%%%%%%

\begin{figure*}
\centering

% ------------ LEFT COLUMN ------------------
\begin{minipage}[t]{0.47\textwidth}
\centering

% Partitioning table
\resizebox{\textwidth}{!}{
\begin{tabular}{|c|c|c|c|c|c|c|c|c|c|c|c|c|c|c|c|c|}
\hline
0 & 1 & 0 & 1 & 0 & 1 & 0 & 1 & 0 & 1 & 0 & 1 & 0 & 1 & 0 & 1 & 1 \\
\hline
1 & 0 & 1 & 0 & 1 & 0 & 1 & 0 & 1 & 0 & 1 & 0 & 1 & 0 & 1 & 0 & 0 \\
\hline
0 & 0 & 0 & 0 & 0 & 0 & 0 & 0 & 1 & 1 & 1 & 1 & 1 & 1 & 1 & 1 & 1 \\
\hline
1 & 1 & 1 & 1 & 1 & 1 & 1 & 1 & 0 & 0 & 0 & 0 & 0 & 0 & 0 & 0 & 0\\
\hline
\end{tabular}}
\vspace{6pt}

% Lifting table
\resizebox{1.05\textwidth}{!}{
\begin{tabular}{|c|c|c|c|c|c|c|c|c|c|c|c|c|c|c|c|c|}
\hline
9 & 4 & 14 & 3 & 6 & 4 & 10 & 5 & 0 & 3 & 2 & 3 & 8 & 4 & 12 & 13 & 2 \\
\hline
15 & 1 & 10 & 7 & 5 & 6 & 1 & 14 & 9 & 4 & 6 & 1 & 6 & 3 & 12 & 1 & 10 \\
\hline
6 & 4 & 14 & 12 & 0 & 3 & 16 & 1 & 7 & 8 & 2 & 12 & 16 & 3 & 1 & 6 & 9 \\
\hline
0 & 12 & 1 & 0 & 12 & 9 & 2 & 3 & 13 & 1 & 1 & 6 & 0 & 16 & 1 & 10 & 4 \\
\hline
\end{tabular}}

\end{minipage}
\hfill
% ------------ RIGHT COLUMN ------------------
\begin{minipage}[t]{0.47\textwidth}
\centering

% Relocation table
\resizebox{\textwidth}{!}{
\begin{tabular}{|c|c|c|c|c|c|c|c|c|c|c|c|c|c|c|c|c|}
\hline
2 & 0 & 0 & 0 & 0 & 1 & 0 & 1 & 0 & 1 & 2 & 0 & 0 & 2 & 1 & 2 & 0 \\
\hline 
0 & 0 & 0 & 0 & 0 & 0 & 0 & 1 & 0 & 0 & 1 & 0 & 1 & 1 & 0 & 1 & 0 \\
\hline 
0 & 0 & 2 & 0 & 1 & 2 & 0 & 0 & 0 & 0 & 0 & 0 & 0 & 0 & 1 & 0 & 1 \\
\hline 
0 & 0 & 1 & 1 & 1 & 0 & 0 & 0 & 0 & 0 & 0 & 0 & 2 & 0 & 0 & 0 & 1 \\
\hline 
\end{tabular}}
\vspace{-8pt}

% Probability-Distribution Matrix
\[ 
\begin{bmatrix} 
\text{\footnotesize 0.3309} & \text{\footnotesize 0.0846} & \text{\footnotesize 0.0846} \\ \text{\footnotesize 0.3309} & \text{\footnotesize 0.0846} & \text{\footnotesize 0.0846} 
\end{bmatrix} 
\]

\end{minipage}

\caption{Partitioning (top left), lifting (bottom left), relocation (top right) matrices of MD Code~1 with parameters $(\gamma,\kappa,z,L,m,M) = (4,17,17,10,1,3)$, and probability-distribution matrix (bottom right) output of MD-GRADE algorithm for this code.}
\label{fig: md_code_1}
\end{figure*}

\begin{figure*}
\centering

% ------------ LEFT COLUMN ------------------
\begin{minipage}[t]{0.47\textwidth}
\centering

% Partitioning table
\resizebox{\textwidth}{!}{%
\begin{tabular}{|c|c|c|c|c|c|c|c|c|c|c|c|c|c|c|c|c|c|c|}
\hline
0 & 1 & 1 & 0 & 1 & 2 & 0 & 2 & 2 & 0 & 1 & 1 & 0 & 1 & 2 & 0 & 2 & 2 & 2 \\
\hline
1 & 0 & 0 & 1 & 0 & 0 & 1 & 0 & 0 & 2 & 2 & 2 & 2 & 2 & 1 & 2 & 1 & 1 & 1 \\
\hline
2 & 2 & 2 & 2 & 2 & 1 & 2 & 1 & 1 & 1 & 0 & 0 & 1 & 0 & 0 & 1 & 0 & 0 & 0 \\
\hline 
\end{tabular}}
\vspace{15pt}

% Lifting table
\resizebox{1.1\textwidth}{!}{%
\begin{tabular}{|c|c|c|c|c|c|c|c|c|c|c|c|c|c|c|c|c|c|c|}
\hline
7 & 0 & 0 & 15 & 0 & 2 & 8 & 19 & 0 & 0 & 12 & 21 & 0 & 7 & 0 & 0 & 0 & 0 & 0 \\
\hline
2 & 2 & 21 & 0 & 22 & 10 & 2 & 0 & 10 & 17 & 2 & 17 & 1 & 6 & 8 & 6 & 0 & 11 & 18 \\
\hline
1 & 8 & 3 & 3 & 20 & 12 & 13 & 13 & 1 & 15 & 3 & 6 & 18 & 0 & 14 & 10 & 13 & 10 & 18 \\
\hline
\end{tabular}}

\end{minipage}
\hfill
% ------------ RIGHT COLUMN ------------------
\begin{minipage}[t]{0.47\textwidth}
\centering

% Relocation table
\resizebox{\textwidth}{!}{%
\begin{tabular}{|c|c|c|c|c|c|c|c|c|c|c|c|c|c|c|c|c|c|c|}
\hline
0 & 0 & 0 & 0 & 3 & 0 & 1 & 1 & 0 & 0 & 0 & 2 & 1 & 2 & 3 & 3 & 0 & 0 & 0 \\
\hline
0 & 0 & 2 & 0 & 0 & 2 & 0 & 0 & 1 & 0 & 1 & 0 & 0 & 0 & 0 & 0 & 0 & 0 & 0 \\
\hline
0 & 3 & 0 & 0 & 0 & 0 & 0 & 0 & 1 & 3 & 0 & 0 & 0 & 0 & 1 & 1 & 3 & 2 & 0 \\
\hline
\end{tabular}}
\vspace{-6pt}

% Matrix
\[
\begin{bmatrix}
\text{\footnotesize 0.2257} & \text{\footnotesize 0.0366} & \text{\footnotesize 0.0343} & \text{\footnotesize 0.0366} \\
\text{\footnotesize 0.2153} & \text{\footnotesize 0.0404} & \text{\footnotesize 0.0372} & \text{\footnotesize 0.0404} \\
\text{\footnotesize 0.2257} & \text{\footnotesize 0.0366} & \text{\footnotesize 0.0343} & \text{\footnotesize 0.0366} \\
\end{bmatrix}
\]

\end{minipage}

\caption{Partitioning (top left), lifting (bottom left), relocation (top right) matrices of MD Code~2 with parameters $(\gamma,\kappa,z,L,m,M) = (3,19,23,10,2,4)$, and probability-distribution matrix (bottom right) output of MD-GRADE algorithm for this code.}
\label{fig: md_code_2}
\end{figure*}

\begin{figure*}
\centering
\resizebox{0.48\textwidth}{!}{\begin{tabular}{|c|c|c|c|c|c|c|c|c|c|c|c|c|c|c|c|c|c|c|c|}
\hline
4 & 1 & 0 & 3 & 0 & 1 & 0 & 3 & 0 & 0 & 3 & 2 & 4 & 0 & 4 & 2 & 0 & 4 & 3 & 3 \\
\hline
0 & 0 & 2 & 4 & 3 & 0 & 3 & 0 & 4 & 0 & 4 & 4 & 0 & 4 & 0 & 4 & 4 & 1 & 0 & 0 \\
\hline
4 & 4 & 4 & 2 & 4 & 4 & 4 & 1 & 0 & 4 & 1 & 0 & 0 & 3 & 4 & 0 & 4 & 0 & 0 & 0 \\
\hline
\end{tabular}}
\vspace{5pt}

\resizebox{0.6\textwidth}{!}{\begin{tabular}{|c|c|c|c|c|c|c|c|c|c|c|c|c|c|c|c|c|c|c|c|}
\hline
12 & 6 & 4 & 12 & 0 & 1 & 5 & 12 & 6 & 0 & 7 & 0 & 4 & 9 & 3 & 11 & 11 & 2 & 12 & 7 \\
\hline
6 & 6 & 6 & 4 & 2 & 0 & 9 & 1 & 3 & 0 & 1 & 10 & 6 & 0 & 2 & 4 & 6 & 10 & 10 & 7 \\
\hline
1 & 8 & 8 & 11 & 6 & 1 & 2 & 2 & 5 & 7 & 5 & 10 & 1 & 0 & 1 & 3 & 11 & 11 & 1 & 8 \\
\hline
\end{tabular}}
\vspace{5pt}

\resizebox{0.48\textwidth}{!}{\begin{tabular}{|c|c|c|c|c|c|c|c|c|c|c|c|c|c|c|c|c|c|c|c|c|c|}
\hline
0 & 4 & 0 & 4 & 0 & 6 & 2 & 0 & 4 & 5 & 2 & 0 & 6 & 0 & 0 & 1 & 0 & 0 & 0 & 1 \\
\hline 
0 & 0 & 4 & 0 & 0 & 1 & 3 & 0 & 0 & 0 & 0 & 0 & 0 & 1 & 0 & 0 & 5 & 5 & 0 & 0 \\
\hline 
0 & 1 & 0 & 3 & 0 & 0 & 0 & 1 & 0 & 0 & 0 & 0 & 0 & 0 & 1 & 0 & 0 & 0 & 3 & 0 \\
\hline 
\end{tabular}}

\[ \begin{bmatrix}
\text{\footnotesize 0.2840} & \text{\footnotesize 0.0166} & \text{\footnotesize 0.0166} & \text{\footnotesize 0.0166} & \text{\footnotesize 0.0166} & \text{\footnotesize 0.0166} & \text{\footnotesize 0.0166} \\
\text{\footnotesize 0.0290} & \text{\footnotesize 0.0091} & \text{\footnotesize 0.0091} & \text{\footnotesize 0.0091} & \text{\footnotesize 0.0091} & \text{\footnotesize 0.0091} & \text{\footnotesize 0.0091} \\
\text{\footnotesize 0.0241} & \text{\footnotesize 0.0071} & \text{\footnotesize 0.0071} & \text{\footnotesize 0.0071} & \text{\footnotesize 0.0071} & \text{\footnotesize 0.0071} & \text{\footnotesize 0.0071} \\
\text{\footnotesize 0.0725} & \text{\footnotesize 0.0101} & \text{\footnotesize 0.0101} & \text{\footnotesize 0.0101} & \text{\footnotesize 0.0101} & \text{\footnotesize 0.0101} & \text{\footnotesize 0.0101} \\
\text{\footnotesize 0.2381} & \text{\footnotesize 0.0159} & \text{\footnotesize 0.0159} & \text{\footnotesize 0.0159} & \text{\footnotesize 0.0159} & \text{\footnotesize 0.0159} & \text{\footnotesize 0.0159}
\end{bmatrix} \]
\caption{Partitioning (first), lifting (second), relocation (third) matrices of GD-GD MD Code~6 with parameters $(\gamma,\kappa,z,L,m,M) = (3,20,13,20,4,7)$, and probability-distribution matrix (bottom) output of MD-GRADE algorithm for this code.}
  \label{fig: md_code_6} 
\end{figure*}

\begin{figure*}
\centering

% ------------ LEFT COLUMN ------------------
\begin{minipage}[t]{0.49\textwidth}
\centering

% Partitioning table
\hspace{1pt}
\resizebox{0.7\textwidth}{!}{
\begin{tabular}{|c|c|c|c|c|c|c|c|c|c|c|c|c|}
\hline
0 & 0 & 3 & 2 & 1 & 3 & 2 & 0 & 3 & 0 & 3 & 3 & 0 \\
\hline
3 & 3 & 0 & 1 & 3 & 0 & 1 & 0 & 3 & 0 & 2 & 0 & 3 \\
\hline
3 & 0 & 1 & 0 & 3 & 1 & 0 & 3 & 0 & 3 & 0 & 3 & 2 \\
\hline 
0 & 3 & 2 & 3 & 0 & 2 & 3 & 3 & 0 & 3 & 0 & 1 & 0 \\
\hline
\end{tabular}}
\vspace{6pt}

% Lifting table
\resizebox{0.7\textwidth}{!}{
\begin{tabular}{|c|c|c|c|c|c|c|c|c|c|c|c|c|}
\hline
1 & 2 & 2 & 0 & 4 & 0 & 2 & 3 & 3 & 1 & 4 & 1 & 1 \\
\hline
1 & 0 & 0 & 1 & 2 & 1 & 2 & 3 & 0 & 4 & 1 & 3 & 2 \\
\hline
2 & 0 & 0 & 2 & 1 & 0 & 1 & 3 & 1 & 0 & 0 & 1 & 2 \\
\hline
2 & 0 & 0 & 3 & 4 & 4 & 3 & 0 & 0 & 0 & 3 & 0 & 0 \\
\hline
\end{tabular}}

\end{minipage}
\hfill
% ------------ RIGHT COLUMN ------------------
\begin{minipage}[t]{0.49\textwidth}
\centering

% Relocation table
\resizebox{0.7\textwidth}{!}{
\begin{tabular}{|c|c|c|c|c|c|c|c|c|c|c|c|c|}
\hline
1 & 0 & 0 & 1 & 2 & 0 & 0 & 1 & 1 & 0 & 4 & 0 & 0 \\
\hline
2 & 0 & 0 & 3 & 0 & 0 & 2 & 3 & 0 & 0 & 0 & 0 & 0 \\
\hline
0 & 2 & 0 & 0 & 1 & 0 & 3 & 0 & 0 & 1 & 0 & 4 & 0 \\
\hline
0 & 0 & 1 & 2 & 0 & 0 & 0 & 0 & 0 & 0 & 0 & 3 & 0 \\
\hline
\end{tabular}}
\vspace{-2pt}

% Probability-Distribution Matrix
\[ \begin{bmatrix}
\text{\footnotesize 0.2753} & \text{\footnotesize  0.0273} & \text{\footnotesize  0.0273} & \text{\footnotesize  0.0273} & \text{\footnotesize  0.0273} \\
\text{\footnotesize 0.0535} & \text{\footnotesize 0.0155} & \text{\footnotesize 0.0155} & \text{\footnotesize 0.0155} & \text{\footnotesize 0.0155} \\
\text{\footnotesize 0.0535} & \text{\footnotesize 0.0155} & \text{\footnotesize 0.0155} & \text{\footnotesize 0.0155} & \text{\footnotesize 0.0155} \\
\text{\footnotesize 0.2753} & \text{\footnotesize  0.0273} & \text{\footnotesize  0.0273} & \text{\footnotesize  0.0273} & \text{\footnotesize  0.0273} \\
\end{bmatrix} \]

\end{minipage}

\caption{Partitioning (top left), lifting (bottom left), relocation (top right) matrices of GD-GD MD Code~7 with parameters $(\gamma,\kappa,z,L,m,M) = (4,13,5,10,3,5)$, and probability-distribution matrix (bottom right) output of MD-GRADE algorithm for this code.}
\label{fig: md_code_7}
\end{figure*}

%%%%%%%%%%%%%%%%%%%%%%%%%%%%%%%%%%%%%
\section{Conclusion} \label{sec: conclusion}

MD-SC codes are being increasingly adopted in a wide range of applications. FL optimization for MD-SC code design is a computationally challenging problem, especially as more degrees of freedom become available with increasing code parameters. In this work, we introduced a probabilistic framework that efficiently exploits these available degrees of freedom to obtain improved MD-SC codes. More specifically, we defined a probability-distribution matrix whose entries represent the distribution of edges across the component and auxiliary matrices, and expressed the expected number of active detrimental objects as a function of this distribution matrix entries. Our initial analysis considered the simpler case of the expected number of active short cycles. This was then complemented by the expected number of cycles in the final Tanner graph, which acts as an estimate of the FL outcome. Such estimates can be used to predict the relocation density required to eliminate all instances of a cycle (when possible), or to determine the best achievable reduction for a given cycle under a maximum relocation percentage imposed by decoding-latency constraints. We then generalized the analysis to derive the expected number of active objects for arbitrary topologies, and finally addressed the case of two short-cycle concatenations. This latter analysis allows us to prioritize objects that are more detrimental to performance, allocating the available degrees of freedom more efficiently, since isolated cycles might not be problematic enough to justify the use of all our relocation resources focusing on them.

We provided expressions of the gradients of the proposed objective functions as well as the solution form of locally-optimal points derived using the KKT conditions. Based on these, we developed a GD algorithm to reach these locally-optimal probability distributions. The resulting probabilistic output is then used to guide an MCMC-based FL optimization algorithm, which is itself probabilistic in nature. Combined with the MCMC method, our framework enables the design of high-performance MD-SC codes, even if the memory and the number of auxiliary matrices is high. Furthermore, when our method was considered together with the earlier probabilistic framework introduced for SC code design~\cite{GRADE}, we observed that improved FL performance is achieved when fewer edges are allocated to the middle component matrices during the partitioning stage, and more edges from these middle component matrices are relocated into auxiliary matrices during the relocation stage (percentage wise). Codes designed based on this distribution, referred to as GD-MD codes, consistently outperform their literature counterparts, underlying SC codes, and MD-SC codes obtained via uniform edge distributions. We observed and reported improvements both in the number of detrimental objects and in error-rate performance. For example, we identified a case where a comparison code designed using a uniform distribution exhibits higher population of the targeted detrimental objects, causing a noticeable error floor, while our GD-MD code has significantly less number of these objects, and thus exhibits an improved FER performance by up to $2.81$ orders of magnitude.

We suggest that our probabilistic MD code design framework can offer a valuable tool to scholars and engineers concerned with the reliability of modern data transmission and storage systems. Future work includes developing a deeper understanding of the edge distribution described above as well as its effectiveness in the FL design using asymptotic tools. Another promising direction is offering a joint design that integrates our GD-MD codes with state-of-the-art constrained coding schemes, such as lexicographically-ordered constrained (LOCO) codes~\cite{reins_loco}, where such a combined construction has potential to be highly effective in two-dimensional magnetic recording (TDMR) systems. Finally, similar concepts can also be applied to extend the work in \cite{ahh_rmc_sc} in order to design rate-compatible MD-SC codes.

\section*{Appendix} \label{sec: appendix}

Figures~\ref{fig: md_code_1}, \ref{fig: md_code_2}, \ref{fig: md_code_6}, and \ref{fig: md_code_7} show the partitioning, lifting, and relocation matrices of MD~Code~1, MD~Code~2, GD-GD MD~Code~6, and GD-GD MD~Code~7, respectively, along with the corresponding probability-distribution matrices generated by the MD-GRADE algorithm.

\section*{Acknowledgment}

The authors would like to thank Canberk \.{I}rima\u{g}z{\i} for his assistance in carrying out this research and for his contributions in the prior version of this work~\cite{reins_gdmd}.

\ifCLASSOPTIONcaptionsoff
  \newpage
\fi

\end{document}